\documentclass[hidelinks, 11 pt, a4paper]{article}
\usepackage{amsfonts,amsmath,amssymb,amsthm}
\usepackage{subfigure}
\usepackage[pdftex]{graphicx}
\usepackage{natbib}
\usepackage[hypertexnames=false]{hyperref}
\usepackage{setspace}
\usepackage{pdfsync}
\usepackage{lscape}
\usepackage{verbatim}
\usepackage[hmargin=1.25in,vmargin=1.25in]{geometry}
\usepackage{footmisc}
\usepackage{rotating}
\usepackage{titling}
\usepackage{mathrsfs,dsfont}
\usepackage[dvipsnames,svgnames, x11names]{xcolor}
\usepackage[capposition=top]{floatrow}
\usepackage{booktabs}
\usepackage{multirow}
\usepackage{xcolor}
\usepackage{comment}

\usepackage[normalem]{ulem}

\usepackage{tikz}
\usepackage{tikz-3dplot}
\usepackage{tkz-euclide}
\usepackage{mathtools}

\hypersetup{colorlinks, linkcolor = {teal}, citecolor = {blue}, urlcolor ={blue!80!black}}

\makeatletter
\renewcommand\paragraph{\@startsection{paragraph}{4}{\z@}%
            {-2.5ex\@plus -1ex \@minus -.25ex}%
            {1.25ex \@plus .25ex}%
            {\normalfont\normalsize\bfseries}}
\makeatother
\usetikzlibrary{patterns,shapes, calc, positioning}

\DeclareGraphicsExtensions{.jpg}

\def\qed{\rule{2mm}{2mm}}
\def\indep{\perp \!\!\! \perp}

\newtheorem{condition}{Condition}
\newtheorem{proposition}{Proposition}[section]
\newtheorem{theorem}{Theorem}[section]
\newtheorem{lemma}{Lemma}[section]

\newtheorem{corollary}{Corollary}[section]
\newtheorem{remark}{Remark}[section]
\newtheorem{assumption}{Assumption}[section]

\let\oldmarginpar\marginpar
\renewcommand{\marginpar}[2][rectangle,draw,fill=black, text=white,text width= 2cm,rounded corners]{
    \oldmarginpar{
    \tiny \tikz \node at (0,0) [#1]{#2};}
    }

\newcommand\Sh{\mathcal Z}
\newcommand\R{\mathbb R}
\newcommand\supp{\text{supp}}
\DeclareMathOperator*{\Times}{\scalebox{1.8}{$\times$}}

\begin{document}


\title{Overidentification in Shift-Share Designs}
\author{Jinyong Hahn\\ UCLA\\ hahn@econ.ucla.edu
\and Guido Kuersteiner\\ University of Maryland \\ gkuerte@umd.edu
\and Andres Santos\\ UCLA \\ andres@econ.ucla.edu
\and Wavid Willigrod\\UCLA \\ dwwilligrod@gmail.com}
\maketitle

\begin{abstract}
This paper studies the testability of identifying restrictions commonly employed to assign a causal interpretation to two stage least squares (TSLS) estimators based on Bartik instruments.
For homogeneous effects models applied to short panels, our analysis yields testable implications previously noted in the literature for the two major available identification strategies.
We propose overidentification tests for these restrictions that remain valid in high dimensional regimes and are robust to heteroskedasticity and clustering.
We further show that homogeneous effect models in short panels, and their corresponding overidentification tests, are of central importance by establishing that:
(i) In heterogenous effects models, interpreting TSLS as a positively weighted average of treatment effects can impose implausible assumptions on the distribution of the data; and
(ii) Alternative identifying strategies relying on long panels can prove uninformative in short panel applications.
We highlight the empirical relevance of our results by examining the viability of Bartik instruments for identifying the effect of rising Chinese import competition on US local labor markets.
\end{abstract}

\newpage

\section{Introduction}

Shift-share designs typically rely on linear combinations of unit specific variables  (the ``shares") and aggregate level variables (the ``shocks") as instruments to obtain identification.
Originally employed in the work of \cite{bartik1991} and \cite{blanchard1992yregional}, these instruments have become known in the literature as ``Bartik" instruments.
Bartik instruments have proven to be remarkably versatile yielding insights into, among others, the impact of immigration on labor markets \citep{card2001immigrant}, the consequence of trade liberalization for poverty \citep{topalova2010factor}, the effect of import competition on labor markets \citep{david2013china}, and the welfare implications of geographic sorting by education \citep{diamond2016determinants}. 


In this paper, we examine the testable implications of identifying restrictions commonly employed to assign a causal interpretation to two stage least squares (TSLS) estimators based on Bartik instruments.
We largely focus our analysis on two recent complementary identification strategies.
The first, studied by \cite{goldsmith2020bartik}, attributes the exogeneity of the Bartik instrument to the exogeneity of the shares.
The second, examined by \cite{adao2019shift} and \cite{borusyak2022quasi}, instead attributes the exogeneity of the Bartik instrument to the exogeneity of aggregate shocks.
We systematically study the overidentifying content of both approaches by noting that their differing asymptotic frameworks either implicitly or explicitly require \emph{conditional} moment restrictions to hold.
Specifically, \cite{goldsmith2020bartik} necessitates the Bartik instrument to be exogenous conditional on the realization of the aggregate shocks, while \cite{adao2019shift} and \cite{borusyak2022quasi} require the Bartik instrument to be exogenous conditional on the realization of the shares (among other variables).

The conditional moment restrictions implied by the different identification strategies yield previously noted overidentifying restrictions in homogeneous treatment effects models.
For instance, when applied to the framework of \cite{goldsmith2020bartik}, our analysis implies that the entire vector of shares must be exogenous.
As noted by \cite{goldsmith2020bartik}, the model is therefore overidentified because any deterministic linear combination of the shares yields a valid instrument.
In contrast, when applied to the framework of \cite{adao2019shift} and \cite{borusyak2022quasi}, our results imply that the conditional mean of the Bartik instrument given controls, shares, and residuals must be a linear function of the controls only.
As a result, the model is overidentified because any appropriately centered deterministic linear combination of the aggregate shocks yields a valid instrument.
Our impression is that the latter overidentifying restrictions have not received the same level of attention by practitioners as the testable implications of \cite{goldsmith2020bartik}.

Building on our overidentification analysis, we develop a framework for testing the overidentifying restrictions corresponding to both identification strategies.
In many of the applications that motivate us, the number of overidentifying restrictions can be ``large" relative to the sample -- e.g., in many applications the numbers of sectors exceeds the number of clusters.
We therefore focus on developing tests that remain valid in such high dimensional settings by employing the high dimensional central limit theorem of \cite{chernozhuokov2022improved}.
Our tests utilize bootstrap based critical values and are robust to the presence of heteroskedasticity, clustering, and weighting.\footnote{In contrast, \cite{goldsmith2020bartik} emply overidentification tests based on LIML (which requires homoskedasticity) and \cite{CHAO201415} (which precludes clustering). Neither test allows the number of restrictions to exceed the sample size.}

A natural way to proceed, for instance after rejecting the validity of the homogeneous effects model, is to consider weaker assumptions that still enable us to attribute a causal interpretation to TSLS.
To this end, we study a generalization of the model and an alternative identification strategy, but find both approaches to be potentially empirically limited in scope.
As a generalization of the model, we consider linear heterogeneous treatment effects models previously employed in the literature \citep{goldsmith2020bartik, adao2019shift}.
In such a setting, TSLS can readily be shown to estimate a weighted average of group specific treatment effects.
As forcibly argued by \cite{kolesa2013estimation} and \cite{blandhol2022tsls} among others, however, the averaging weights should be positive in order to attribute a causal interpretation to TSLS.
Unfortunately, we find that in the case of Bartik instruments the weights can naturally be negative.
In particular, within the identification framework of \cite{goldsmith2020bartik}, we show that a necessary condition for the weights to be positive is that shares of different sectors be uncorrelated with each other -- a requirement that a-fortiori fails when shares sum up to one.
Conversely, within the identification framework of \cite{adao2019shift} and \cite{borusyak2022quasi},  we show that a necessary condition for the weights to be positive is that aggregate shocks to different sectors be (weakly) positively correlated with each other.

As an alternative identification strategy, we consider the possibility that the exogeneity of the instrument is not justified by the shares alone \citep{goldsmith2020bartik} or the aggregate shocks alone \citep{adao2019shift, borusyak2022quasi} but by the shares and shocks together.
Formally, such an identification strategy corresponds to requiring that the Bartik instrument be uncorrelated with the error term when expectations are evaluated over both the time series distribution of aggregate shocks and the cross sectional distribution of shares.
While this identification strategy renders the model ``just identified," it also necessitates a long panel in order to estimate the time series distribution of the aggregate shocks.
We show that as a result, the effective number of observations for computing TSLS standard errors is governed by the number of time periods in the panel.
Because the majority of empirical studies relying on Bartik instruments have relied on short panels, we expect TSLS to be statistically uninformative under this identification strategy.
Our asymptotic analysis relies on a novel simultaneous time series and cross sectional study of long panels.
These results significantly extend related work in \cite{HKM, RePEc:cup:etheor:v:38:y:2022:i:5:p:942-958_5} and may be of independent interest.

We highlight the empirical relevance of our analysis by revisiting the study by \cite{david2013china} on the impact of rising Chinese import competition on local US labor markets.
In this setting, our overidentification tests find evidence against the validity of the identification framework of \cite{goldsmith2020bartik} as well as the identification framework of \cite{adao2019shift} and \cite{borusyak2022quasi}.
Moreover, since shares are empirically correlated across sectors, our analysis further suggests that TSLS may not have a causal interpretation in a heterogeneous effects model under the identification strategy of \cite{goldsmith2020bartik}.
Likewise, TSLS may not possess a causal interpretation in a heterogeneous effects model under the identification strategy of \cite{adao2019shift} and \cite{borusyak2022quasi} when shocks within clusters are negatively correlated.
Finally, we note that applying the proposed alternative just identified long panel identification strategy is not viable in this application because there are only two time periods.

The rest of the paper is organized as follows.
Section \ref{sec:overid} characterizes the overidentifying restrictions implied by the identification strategies of \cite{goldsmith2020bartik} as well \cite{borusyak2022quasi} and \cite{adao2019shift}.
Section \ref{sec:justid} examines the scope for attributing a causal interpretation to TSLS under heterogeneous effects models or long panel identification strategies.
Our overidentification tests are developed in Section \ref{sec:test}, while Section \ref{sec:china} illustrates the  relevance of our analysis by revisiting \cite{david2013china}.
Section \ref{sec:recs} briefly concludes by providing recommendations for empirical practice.
A series of appendices contain proofs of our results and a Monte Carlo study.

\section{Overidentification}\label{sec:overid}

Following the main instrumental variables specifications of \cite{goldsmith2020bartik} and \cite{adao2019shift}, we begin by studying homogeneous treatment effects models in which the parameter of interest is common across all individuals.
In the next section, we will examine heterogenous treatment effects models instead.
For ease of exposition, we focus on cross sectional applications though note that an extension to short panel data settings is immediate; see, e.g., Remarks \ref{rm:panelGSS} and \ref{rm:panelakm} below.

In what follows we let $Y_{i}\in\R$, $X_{i}\in\R$, $Z_i\in \R$, and $W_{i}\in\R^d$ denote the outcome variable, a scalar regressor of interest, a scalar instrument, and a vector of controls.
We will refer to each ``$i$" as an individual, though note that in applications $i$ may represent, for example, a location.
The variables $(Y_i,X_i,W_i)$ are assumed to satisfy
\begin{equation}\label{model:eq1}
Y_{i}=X_{i}\beta+W_{i}^{\prime}\gamma_{\mathtt{s}}+\varepsilon_{i},
\end{equation}
where the restrictions on $\varepsilon_i$ will be stated shortly.
We will further require that the instrument $Z_{i}$ follow a Bartik structure in the sense
that it satisfies
\begin{equation*}
Z_{i}=\Sh^{\prime}S_{i}
\end{equation*}
with $\Sh \in\R^p$ an aggregate variable and $S_{i}\in\R^p$ an individual specific variable.

For instance, in canonical applications $\Sh$ equals a vector of aggregate industry specific shocks and $S_i$ equals the share of each industry in the economy of location $i$.
For this reason, in what follows we refer to $\Sh$ as ``shocks" and $S_i$ as ``shares."
Other variables, such as $\varepsilon_i$, may have a Bartik structure as well though we only make the Bartik structure explicit for $Z_i$.
Finally, we complement equation \eqref{model:eq1} with the first stage equation
\begin{equation}\label{model:eq3}
X_{i}=Z_{i} \lambda +W_{i}^{\prime}\gamma_{\mathtt{f}%
}+\eta_{i}.
\end{equation}
Here, the $\mathtt{f}$ and $\mathtt{s}$ subscripts on $\gamma$ refer to the first and the second stage respectively.
While we let the first stage coefficients in \eqref{model:eq3} be fixed for simplicity, it is worth emphasizing that all the results in this section continue section hold if the first stage coefficients are instead random as in \cite{adao2019shift}.
The observable variables are the aggregate shocks $\Sh$ and the cross sectional variables $\left\{Y_{i},X_{i},W_{i},S_{i}\right\}_{i=1}^n$.

Different asymptotic frameworks either implicitly or explicitly condition on different variables when delivering asymptotic promises.
To reflect these differences, it will prove convenient to introduce the notation $\mathcal G_n$ to denote the relevant conditioning variables.
For instance, we will argue below that asymptotic approximations that rely on only the cross section being large implicitly condition on aggregate variables such as the shocks $\Sh$.
Hence, in such instances we would set $\mathcal G_n$ to include all aggregate variables, including $\Sh$.
Alternatively, when specialized to equations \eqref{model:eq1} and \eqref{model:eq3}, the asymptotic framework of \cite{adao2019shift} and \cite{borusyak2022quasi} explicitly conditions on the cross-sectional variables $\{W_i,S_i,\varepsilon_i\}_{i=1}^n$ so we would in that case simply set $\mathcal G_n$ to equal $\{W_i,S_i,\varepsilon_i\}_{i=1}^n$.

We will first examine conditions under which the two stage least squares (TSLS) estimand equals $\beta$ in the homogeneous effects model defined by \eqref{model:eq1} and \eqref{model:eq3}.
To this end, it is helpful to define the ``population" residualized instrument $\dot Z_{i,n}$ according to
\begin{equation}\label{model:eq4}
\dot Z_{i,n} \equiv Z_i - W_i^\prime \pi_n \hspace{0.5 in} \pi_{n}   \equiv\left(  \frac{1}{n}\sum_{i=1}^{n}E\left[  \left.  W_{i} W_{i}^{\prime}\right\vert \mathcal{G}_{n}\right]  \right)  ^{-1}\left( \frac{1}{n}\sum_{i=1}^{n}E\left[  \left.  W_{i}Z_{i}\right\vert \mathcal{G}_{n}\right]  \right),
\end{equation}
which allows us to define the TSLS estimand $\beta_{0,n}$ as the solution to the equation
\begin{equation}\label{model:eq5}
\frac{1}{n}\sum_{i=1}^{n}E\left[  \left.  \left(  Y_{i}-X_{i}\beta
_{0,n}\right)  \dot{Z}_{i,n}\right\vert \mathcal{G}_{n}\right]  =0.
\end{equation}
We note that in the notation we have let $\beta_{0,n}$ depend on $n$ to reflect that the TSLS estimand may depend on $n$ through the conditioning on $\mathcal G_n$.
Under an appropriate rank condition, plugging the outcome equation \eqref{model:eq1} into the moment condition in \eqref{model:eq5} readily yields that the estimand $\beta_{0,n}$ equals $\beta$ if and only if we have
\begin{equation}\label{model:eq6}
\frac{1}{n}\sum_{i=1}^n E\left[\varepsilon_i(Z_i-W_i^\prime \pi_n)|\mathcal G_n\right] = 0    .
\end{equation}
We next examine the implications of this exogeneity condition under the different choices of conditioning set $\mathcal G_n$ that correspond to different asymptotic approximations.

\subsection{Conditioning on Shocks}\label{sec:homoggss}

Traditionally, empirical work employing shift-share designs has reported standard errors that are motivated by asymptotic approximations in which only the cross section becomes large (i.e.\ $n\to \infty$).
Probabilistic statements associated with these asymptotic approximations correspond to a thought experiment in which we only re-sample the individual specific variables $\{Y_i,X_i,W_i,S_i\}_{i=1}^n$.
For instance, in this context, the level of a confidence interval for $\beta$ refers to the probability that a randomly drawn sample $\{Y_i,X_i,W_i,S_i\}_{i=1}^n$ yields a confidence interval that indeed includes $\beta$.

Probabilistic statements based on re-sampling only individual specific variables a fortiori keep aggregate variables such as the shocks $\Sh$ fixed.
Consequently, in an asymptotic approximation in which only $n$ becomes large we are only able to identify the distribution of individual specific variables \emph{conditionally} on the realizations of aggregate variables such as $\Sh$. 
When evaluating the exogeneity condition on our instrument (i.e.\ restriction \eqref{model:eq6}) under this asymptotic approximation,  we should therefore view all aggregate variables as belonging to the conditioning set  $\mathcal G_n$.

In order to illustrate the implications of these observations, we will for simplicity assume that the  variables $\{S_i,W_i,\varepsilon_i\}_{i=1}^n$ are i.i.d.\ and independent of all aggregate variables.\footnote{This setting rules out, for example, that $(W,\varepsilon)$ have a Bartik structure as in \cite{adao2019shift}. We note, however, that the main points made in this section carry over if we allow $(W,\varepsilon)$ to have a Bartik structure as well, though such an extension requires additional notation and assumptions. }
Provided that the controls $W_i$ are uncorrelated with the error $\varepsilon_i$, it then follows from a bit of algebra that the exogeneity condition in \eqref{model:eq6} is equivalent to
\begin{equation}\label{model:eq7}
\Sh^\prime E\left[S_i \varepsilon_i\right] = 0    .
\end{equation}
Equation \eqref{model:eq7} highlights that the relevant exogeneity condition for this asymptotic framework depends on the realization of the aggregate shocks $\Sh$ and the correlation between the shares $S_i$ and the error $\varepsilon_i$.
Lacking a justification as to why \eqref{model:eq7} should hold at the \emph{actual} realization of $\Sh$ and not others, we should in the interest of robustness demand that \eqref{model:eq7} hold at \emph{all} possible realizations of $\Sh$.
Provided the aggregate shocks $\Sh$ exhibit sufficient variation, however, it then follows that condition \eqref{model:eq7} can hold for all possible realizations of $\Sh$ if and only if in fact $S_i$ itself is uncorrelated with the errors $\varepsilon_i$.

Our next simple proposition formalizes the preceding discussion.

\begin{proposition}\label{prop:homoggss}
Suppose that $\mathcal Z\in \mathcal G_n$ and that  $\{S_i,W_i,\varepsilon_i\}_{i=1}^n$ are i.i.d.\ and independent of $\mathcal G_n$.
If $E[W\varepsilon] = 0$ and the support of $\mathcal Z\in \R^p$ has dimension $p$, then
\begin{equation}\label{prop:homoggssdisp}
\frac{1}{n}\sum_{i=1}^n E\left[\varepsilon_i(Z_i - W_i^\prime \pi_n)|\mathcal G_n \right] =0
\end{equation}
with probability one (over $\Sh$) if and only if $E[S_i \varepsilon_i ] = 0$.
\end{proposition}

Proposition \ref{prop:homoggss} establishes that the entire vector of shares $S_i$ must be a valid instrument in order for the scalar instrument $Z_i$ itself to be valid.
This conclusion echoes arguments in \cite{goldsmith2020bartik}, who similarly argue that the exogeneity of the instrument $Z_i$ should be understood in terms of the exogeneity of the shares $S_i$.
The principal implication of Proposition \ref{prop:homoggss} for our purposes is that the validity of the Bartik instrument renders the model overidentified.
For instance, it follows that we may construct an overidentification test by employing the sample moments
\begin{equation*}
\frac{1}{n}\sum_{i=1}^n S_i \hat \varepsilon_i ,
\end{equation*}
where $\{\hat \varepsilon_i\}_{i=1}^n$ denotes the fitted residuals from the estimated model.
One complication that arises in building overidentification tests from these moments is that the number of moments (i.e.\ the number of shares) is often too ``large" for standard asymptotic approximations to remain accurate.
In Section \ref{sec:test}, we address this challenge by developing overidentification tests that rely on high dimensional asymptotics instead.

\begin{remark}\label{rm:panelGSS} \rm
While we have focused on a cross-sectional setting for ease of exposition, our conclusions readily extends to panel data applications.
Specifically, suppose there are $1\leq t \leq T$ time periods, each with an aggregate shock $\Sh_t$ and a cross-section $\{Y_{it},X_{it},S_{it},W_{it}\}_{i=1}^n$.
Under a short panel asymptotic approximation (i.e.\ $n\to\infty$ and $T$ fixed), conditions similar to those employed in Proposition \ref{prop:homoggss} readily imply that as the relevant exogeneity condition we should require that
\begin{equation*}
E[S_{it}\varepsilon_{it}] = 0 \quad \quad \text{ for } 1\leq t \leq T.
\end{equation*}
Thus, in this context, the time dimension yields additional overidentifying restrictions to those available in the cross sectional setting. \qed
\end{remark}

\subsection{Identification Through Shocks}\label{sec:homogakm}

As an alternative identification strategy, \cite{adao2019shift} and \cite{borusyak2022quasi} advocate interpreting the exogeneity of the Bartik instrument $Z_i$ as originating from the exogeneity of the aggregate shocks $\Sh$.
Within our context, the assumptions employed by their asymptotic approximations correspond to setting  $\mathcal G_n$ to equal $\{S_i,W_i,\varepsilon_i\}_{i=1}^n$ and letting the number of shocks $p$ increase with the sample size $n$.\footnote{As noted by \cite{adao2019shift}, including $\{\varepsilon_i\}_{i=1}^n$ in the conditioning set is important for obtaining asymptotically valid standard errors in applications in which $\{\varepsilon_i\}_{i=1}^n$ possesses a Bartik structure.}
Under this asymptotic framework, the relevant exogeneity condition (i.e.\ restriction \eqref{model:eq6}) is equivalent to
\begin{equation}\label{model:eq9}
\frac{1}{n}\sum_{i=1}^n \varepsilon_i\left(S_i^\prime E\left[\Sh|\mathcal G_n \right] - W_i^\prime \pi_n\right) = 0  .
\end{equation}

To gain intuition into this exogeneity requirement, it is instructive to consider the case in which there are no controls $W_i$. 
Letting $\Sh_j$ denote the $j^{th}$ coordinate of the vector of aggregate shocks $\Sh\in \R^p$ and $S_{ij}$ the $j^{th}$ coordinate of the shares $S_i\in \R^p$, it then follows that the exogeneity condition in \eqref{model:eq9} simplifies to the expression
\begin{equation}\label{model:eq10}
\sum_{j=1}^p \left(E[\mathcal Z_j|\mathcal G_n] \times  \left(\frac{1}{n}\sum_{i=1}^n \varepsilon_i S_{ij}\right)\right) = 0 .
\end{equation}
Moreover, since the moment restriction in \eqref{model:eq10} must hold for any possible realization of the sample $\{S_i,\varepsilon_i\}_{i=1}^n$, the exogeneity condition in this context implies, under appropriate assumptions, that the aggregate shocks $\{\Sh_j\}_{j=1}^p$ must have mean zero and be mean independent of the individual specific variables $\{S_i,\varepsilon_i\}_{i=1}^n$.
However, if all the aggregate shock $\{\Sh_j\}_{j=1}^p$ are mean independent of the individual specific variables, then the instrument $Z_i = S_i^\prime \mathcal Z$ must itself be mean independent of $\{S_i,\varepsilon_i\}_{i=1}^n$ because
\begin{equation*}
E[Z_i|\{S_i,\varepsilon_i\}_{i=1}^n] = \sum_{j=1}^p E[\Sh_j|\mathcal G_n]S_{ij} = 0.
\end{equation*}
Intuitively, if the aggregate shocks $\{\Sh_j\}_{j=1}^p$ are exogenous in the sense that \eqref{model:eq10} holds, then any suitable linear combination of the shocks, such as $Z_i = S_i^\prime \Sh$, must be mean independent of $\{S_i,\varepsilon_i\}_{i=1}^n$ as well.
In particular, we obtain the overidentifying restriction that $Z_i$ must be uncorrelated with any functions of the shares $S_i$ and the error $\varepsilon_i$.\footnote{Equivalently, note that since any suitably linear combination of $\Sh$ is also a valid instrument, it follows that there is a surplus of instruments and hence that the model is overidentified.}

While the inclusion of controls $W_i$ allows us to relax the requirements on the aggregate shocks $\Sh$, the model remains overidentified through its implications on the conditional mean of the instrument $Z_i$.
Our next proposition illustrates this conclusion in the simple setting in which $\Sh$ is independent of all individual specific variables.

\begin{proposition}\label{prop:homogakm}
Suppose $\mathcal G_n = \{S_i,W_i,\varepsilon_i\}_{i=1}^n$ and $\Sh$ is independent of $\mathcal G_n$.
If the support of $\{\varepsilon_i\}_{i=1}^n$ conditional on $\{S_i,W_i\}_{i=1}^n$ has dimension $n$, then
\begin{equation}\label{prop:homogakmdisp}
\frac{1}{n}\sum_{i=1}^n E[\varepsilon_i(Z_i - W_i^\prime \pi_n)|\mathcal G_n] = 0
\end{equation}
with probability one (over $\{S_i,W_i,\varepsilon_i\}_{i=1}^n$) if and only if $E[Z_i|\{S_i,W_i,\varepsilon_i\}_{i=1}^n] = W_i^\prime \pi_n$.
\end{proposition}

Proposition \ref{prop:homogakm} establishes that the exogeneity requirement needed to obtain identification through the aggregate shocks $\Sh$ effectively imposes that the conditional mean of the instrument $Z_i$ given $\{S_i,W_i,\varepsilon_i\}_{i=1}^n$ be equal to the linear projection of $Z_i$ onto $W_i$.
This conclusion can be shown to hold even if $\Sh$ is not independent of the individual specific variables $\{S_i,W_i,\varepsilon_i\}_{i=1}^n$ at the cost of additional notation and assumptions.
For our purposes, the principal implication of Proposition \ref{prop:homogakm} is that the residualized instrument $(Z_i - W_i^\prime \pi_n)$ must be uncorrelated with any function of the shares $S_i$, controls $W_i$, and errors $\varepsilon_i$.
This observation suggests employing the sample moments
\begin{equation*}
\frac{1}{n}\sum_{i=1}^n g(\hat \varepsilon_i,W_i,S_i)(Z_i - W_i^\prime \hat \pi_n)
\end{equation*}
to construct an overidentification test -- here $g$ is an arbitrary vector valued function and $\{Z_i - W_i^\prime \hat \pi_n\}_{i=1}^n$ denotes the residuals from regressing the instrument on the controls.
In Section \ref{sec:test} we will develop such an overidentification test while allowing the number of sample moments to potentially be high dimensional.

\begin{remark}\label{rm:panelakm} \rm
The overidentifying restrictions of Proposition \ref{prop:homogakm} generalize to panel data applications.
Specifically, suppose that at each time period $1\leq t \leq T$ we observe a  shock $\Sh_t \in \R^p$ and individual specific variables $\{Y_{it},X_{it},S_{it},W_{it}\}_{i=1}^n$.
In this context, the asymptotic framework in \cite{adao2019shift} and \cite{borusyak2022quasi} corresponds to setting $\mathcal G_n$ to equal $\{S_{it},W_{it},\varepsilon_{it}\}_{i=1,t=1}^{n,T}$ and letting $T$ be fixed while $n$ and $p$ become large.
Under conditions similar to those imposed in Proposition \ref{prop:homogakm}, it follows that the exogeneity requirements imposed on the shocks $\{\Sh_t\}_{t=1}^T$ to obtain identification imply
\begin{equation*}
E[Z_{it}|\mathcal G_n] = W_{it}^\prime \pi_n \hspace{0.5 in} \pi_n \equiv \left(\sum_{i=1}^n\sum_{t=1}^T W_{it}W_{it}^\prime \right)^{-1} \left(\sum_{i=1}^n\sum_{t=1}^T W_{it}E[Z_{it}|\mathcal G_n]\right).
\end{equation*}
As a result, the model is now overidentified by the restriction that the residualized instrument $(Z_{it}-W_{it}^\prime \pi_n)$ be uncorrelated with any function of $\{W_{it},S_{it},\varepsilon_{it}\}_{t=1}^T$. \qed
\end{remark}

\begin{remark}\label{rm:akmrel} \rm
\cite{adao2019shift} and \cite{borusyak2022quasi} impose sufficient conditions for the overidentifying restrictions derived in Proposition \ref{prop:homogakm} to hold.
For instance, \cite{adao2019shift} require that the controls satisfy $W_i = \mathcal W^\prime S_i$, with $\mathcal W$ a $p\times d$ matrix of aggregate shocks, and impose that $E[\Sh|\{S_i,W_i,\varepsilon_i\}_{i=1}^n,\mathcal W] = \mathcal W\Gamma$ for some vector $\Gamma\in \R^d$.\footnote{More precisely, \cite{adao2019shift} impose that $W_i = \mathcal W^\prime S_i + U_i$ and require $U_i$ to be asymptotically negligible in a suitable sense; see their Assumptions 3(iii)(iv).}
Under this structure, it follows that $\pi_n$ equals $\Gamma$ and hence that
\begin{multline*}
E[Z_i|\{S_i,W_i,\varepsilon_i\}_{i=1}^n] \\
= E[S_i^\prime E[\Sh | \{S_i,W_i,\varepsilon_i\}_{i=1}^n,\mathcal W]|\{S_i,W_i,\varepsilon_i\}_{i=1}^n] = E[S_i^\prime \mathcal W\Gamma | \{S_i,W_i,\varepsilon_i\}_{i=1}^n] = W_i^\prime \pi_n,
\end{multline*}
which implies the overidentifying restrictions of Proposition \ref{prop:homogakm} are indeed satisfied. \qed
\end{remark}

\section{Just Identification?}\label{sec:justid}

We have so far shown that the homogeneous effects model is overidentified in shift-share designs.
One possible way to proceed, for instance after statistically rejecting the model, is to weaken assumptions in a manner that renders the model potentially just identified  (e.g., in the sense of \citep{chen2018overidentification}).
In this section, we show that two natural relaxations to the model of Section \ref{sec:overid} are severely limited in the set of empirical contexts to which they may be successfully applied.
As a result, we conclude that the homogeneous effects model of Section \ref{sec:overid} is of central empirical importance in shift-share designs relying on Bartik instruments.

First, we examine the possibility of relaxing the model in Section \ref{sec:overid} to a linear heterogeneous effects model.
In this context, the TSLS estimand can be interpreted as a weighted average of causal effects for different population subgroups.
We find, however, that ensuring  the estimand equals a \emph{positively} weighted average of causal effects requires strong, often unrealistic, assumptions on the distribution of the data.

Second, we examine the possibility of relaxing the model in Section \ref{sec:overid} by only requiring that the Bartik instrument be exogeneous when expectations are evaluated over \emph{both} the shocks and the shares.
We show that, under this exogeneity requirement, identification crucially relies on the time series variation of the shocks and, as a result, that the standard errors of TSLS depend on the time dimension.
Because the majority of empirical studies relying on shift-share designs have employed short panels, we expect TSLS estimates to be highly imprecise under this identification strategy.

\subsection{Heterogeneous Effects}

We first re-examine the shift-share design of Section \ref{sec:overid} in the presence of heterogenous treatment effects.
To this end, we generalize the second stage equation by setting
\begin{equation}\label{model:eq13}
Y_i = X_i \beta_i + W_i^\prime \gamma_{\mathtt{s}} + \varepsilon_i;
\end{equation}
i.e.\ the ``treatment effect" $\beta_i$ is allowed to depend on the individual but, for simplicity, we keep the coefficient for the controls fixed.
We complement the second stage by following \cite{adao2019shift} and \cite{goldsmith2020bartik} in imposing a linear first stage
\begin{equation}\label{model:eq14}
X_i = \Sh^\prime \Lambda_i S_i + W_i^\prime \gamma_{\mathtt{f}} + \eta_i \hspace{0.5 in} \Lambda_i = \text{diag}\left\{\left(\lambda_{i1},\ldots, \lambda_{ip}\right)\right\},
\end{equation}
where ``$\text{diag}\{(a_1,\ldots, a_p)\}$" denotes a diagonal matrix with diagonal entries $(a_1,\ldots, a_p)$ and again, for simplicity, we set the coefficient for the controls $W_i$ to be fixed.
In analogy to \cite{imbens1994identification}, we refer to $\Lambda_i$ as the ``type" of the individual.

Within this context, we examine conditions under which the TSLS estimand retains a causal interpretation.
Formally, we remain interested in the parameter $\beta_{0,n}$ solving
\begin{equation*}
\frac{1}{n}\sum_{i=1}^n E\left[(Y_i - X_i \beta_{0,n})\dot Z_{i,n}\Big|\mathcal G_n\right] = 0,
\end{equation*}
but now study conditions under which $\beta_{0,n}$ equals a positively weighted average of the average treatment effects for different subgroups.
Under suitable exogeneity and rank conditions on the instrument $Z_i$, it is possible to show that the estimand $\beta_{0,n}$ satisfies
\begin{equation}\label{model:eq16}
    \beta_{0,n} =  E\left[\left(\sum_{i=1}^n\omega_{i,n}E\left[\beta_i|\Lambda_i,W_i,\mathcal G_n\right]\right)\Bigg|\mathcal G_n\right] \hspace{0.3 in} \omega_{i,n} \equiv \frac{E\left[X_i\dot Z_{i,n}|\Lambda_i,W_i,\mathcal G_n\right]}{\sum_{i=1}^n E\left[X_i \dot Z_{i,n}|\mathcal G_n\right]};
\end{equation}
i.e.\ $\beta_{0,n}$ equals the expectation of a $\{\omega_{i,n}\}_{i=1}^n$-weighted average of the average treatment effects for subgroups determined by $\Lambda_i$ and $W_i$; see Lemma \ref{lm:decomp} for a formal statement.
As forcibly argued in the literature, a minimal requirement for $\beta_{0,n}$ to possess a causal interpretation is that the weights $\{\omega_{i,n}\}_{i=1}^n$ be positive \citep{kolesa2013estimation, blandhol2022tsls}.
In what follows, we study the implications of requiring that such a positivity condition hold under different asymptotic frameworks -- i.e.\ under different choices of $\mathcal G_n$.
For succinctness, we will refer to TSLS as having a causal interpretation whenever the weights $\{\omega_{i,n}\}_{i=1}^n$ are  positive with probability one.

\subsubsection{Conditioning on Shocks}\label{sec:heterogss}

We begin by studying the conditions under which TSLS has a causal interpretation in an asymptotic setting in which only the cross section grows (i.e.\ $n \to \infty$).
Recall that in Section \ref{sec:homoggss} we showed that this setting implicitly conditions on the shocks $\Sh$ and requires the shares $S_i$ to be a valid instrument.

The analysis in this section relies on the following assumption and additional regularity conditions that we formally state in the Appendix; see Assumption \ref{ass:heterogssextra}.

\begin{assumption}\label{ass:heterogss}
(i) $\Sh$ and $\{Y_i,X_i,W_i,S_i\}_{i=1}^n$ satisfy equations \eqref{model:eq13} and \eqref{model:eq14};
(ii) $(\Sh,\Lambda_i,\beta_i,\varepsilon_i,\eta_i)\indep S_i$ conditionally on $W_i$;
(iii) $E[Z_i|W_i,\mathcal G_n] = W_i^\prime \pi_n$ with $\mathcal G_n = \{\Sh\}$.
\end{assumption}

Assumption \ref{ass:heterogss}(i) formally imposes the structure of the model, while Assumption \ref{ass:heterogss}(ii) requires the shares $S_i$ to be suitably exogenous.
In Assumption \ref{ass:heterogss}(iii) we further demand that the conditional mean of the instrument given the covariates be linear.\footnote{Under regularity conditions, Assumption \ref{ass:heterogss}(iii) is equivalent to $E[S_i|W_i]$ being linear in $W_i$.}
\cite{blandhol2022tsls} showed in a related context that Assumption \ref{ass:heterogss}(iii) is a necessary condition for TSLS to have a causal interpretation.
We therefore impose Assumption \ref{ass:heterogss}(iii)  not because it is innocuous, but because it enables us to examine what additional conditions are necessary for TSLS to have a causal interpretation.
Assumption \ref{ass:heterogss}(iii) is of course testable, and practitioners may wish to examine its validity in assessing whether the TSLS estimator can be assigned a causal interpretation.

Our next proposition characterizes the weights $\{\omega_{i,n}\}_{i=1}^n$ and obtains necessary and sufficient conditions for them to be positive with probability one.

\begin{proposition}\label{prop:heterogss}
If Assumptions \ref{ass:heterogss} holds and $\mathcal G_n = \Sh$, then it follows that
\begin{equation}\label{prop:heterogssdisp}
\omega_{i,n} = \frac{\Sh^\prime (\Lambda_i {\rm Var}\{S_i|W_i\})\Sh}{\sum_{j=1}^n \Sh^\prime (\Lambda_j {\rm Var}\{S_j|W_j\})\Sh}.
\end{equation}
Moreover, if in addition Assumption \ref{ass:heterogssextra} holds, then for $n$ large $\omega_{i,n}$ is positive for all $i$ with probability one (over $\Lambda_i,W_i,\Sh$) if and only if the matrix $\Lambda_i {\rm Var}\{S_i|W_i\}$ is either positive semi-definite with probability one, or negative semi-definite with probability one.
\end{proposition}

The first part of Proposition \ref{prop:heterogss} shows that the weights may be expressed as a quadratic form in the shocks $\Sh$ (see \eqref{prop:heterogssdisp}).
Intuitively, provided that the support of the shocks $\Sh$ is sufficiently rich, it therefore follows that the weights can only be positive for all possible realizations of the shocks $\Sh$ if in fact the matrices $\Lambda_i\text{Var}\{S_i|W_i\}$ are positive semi-definite (or negative semi-definite) with probability one.
The second part of Proposition \ref{prop:heterogss} formalizes this intuition under the requirement that the support of the shocks $\Sh$ contain a neighborhood of zero.\footnote{This requirement is formally stated in Assumption \ref{ass:heterogssextra} in the Appendix.}

The necessary and sufficient conditions derived in Proposition \ref{prop:heterogss} for TSLS to have a causal interpretation are both highly restrictive and testable.
Our next corollary shows that, under mild conditions on the support of $\Lambda_i$, these conditions in fact necessarily fail whenever shares are correlated with each other.

\begin{corollary}\label{cor:heteroimp}
Let the conditions of Proposition \ref{prop:heterogss} hold and suppose that the support of $\lambda_{ij}/\lambda_{ik}$ is unbounded for any $j\neq k$.
If for some $j\neq k$ we have that
\begin{equation}\label{cor:heteroimpdisp}
P\left({\rm Cov}\{S_{ij},S_{ik}|W_i\}\neq 0\right) > 0,
\end{equation}
then the weights $\{\omega_{i,n}\}_{i=1}^n$ are negative with positive probability (over $(\Lambda_i,W_i,\Sh)$).
\end{corollary}

Corollary \ref{cor:heteroimp} implies that, whenever shares are correlated, TSLS necessarily lacks a causal interpretation under certain realizations of the shocks $\Sh$.
Because shares must be correlated whenever they sum up to one, we view Corollary \ref{cor:heteroimp} as a  warning that TSLS based on the Bartik instrument can easily fail to have a causal interpretation under a heterogenous effects model.
However, it may be worth emphasizing that, as argued by \cite{goldsmith2020bartik}, empirical researchers may still be able to estimate causal parameters by carefully employing the shares as separate instruments instead of combining them into a single scalar $Z_i$ by employing the shocks $\Sh$.

\begin{remark}\label{rm:heterotest} \rm
Corollary \ref{cor:heteroimp} yields testable implications beyond restricting the correlation between shares.
Notably, because $\text{Var}\{S_i|W_i\}$ must be diagonal and $\Lambda_i$ must have constant sign in order for TSLS to have a causal interpretation, it follows that the sign of the first stage must be constant under certain choices of instrument.
In particular, letting $\tilde {\Sh}_{ij} = f_j(W_i)\Sh_j$ for any positive $f_j$ and $\tilde {\Sh}_{i} = (\tilde {\Sh}_{i1},\ldots, \tilde {\Sh}_{ip})^\prime$, it follows that
\begin{equation}\label{rm:heterotest1}
\text{sign}\left\{E[X_i \dot Z_{i,n}|\Sh]\right\} = \text{sign}\left\{E[X_i (S_i - E[S_i|W_i])^\prime \tilde {\Sh}_{i}|\Sh ]\right\}.
\end{equation}
Since the conditional mean of the shares is linear in the controls in our setting, we obtain a simple diagnostic: Compare the sign of the covariances between $X_i$ and $\dot S_{i,n}^\prime \tilde {\Sh}_i$ for different choices of $\tilde {\Sh}_i$ -- here, $\dot S_{i,n}$ denotes the residual from regressing $S_i$ on $W_i$. \qed
\end{remark}


\subsubsection{Identification Through Shocks}\label{sec:heteroakm}

We next examine the conditions under which TSLS retains a causal interpretation in an asymptotic framework in which identification is driven by the exogeneity of the aggregate shocks.
To this end, we follow \cite{adao2019shift} and \cite{borusyak2022quasi} in augmenting the conditioning set $\mathcal G_n$ of Section \ref{sec:homogakm} to include all unobserved heterogeneity -- i.e., we set $\mathcal G_n$ to equal $\{S_i,W_i,\Lambda_i,\beta_i,\varepsilon_i,\eta_i\}_{i=1}^n$.

Our analysis in this section relies on the following assumption and a set of regularity conditions that we formally state in the Appendix; see Assumption \ref{ass:heteroakmextra}.

\begin{assumption}\label{ass:heteroakm}
(i) $\Sh$ and $\{Y_i,X_i,W_i,S_i\}_{i=1}^n$ satisfy equations \eqref{model:eq13} and \eqref{model:eq14};
(ii) $E[Z_i|W_i,\mathcal G_n] = W_i^\prime \pi_n$ with $\mathcal G_n = \{S_i,W_i,\Lambda_i,\beta_i,\varepsilon_i,\eta_i\}_{i=1}^n$.
\end{assumption}

Assumption \ref{ass:heteroakm} imposes the structure of our model and requires that the conditional mean of the instrument given the controls be linear.\footnote{Sufficient conditions for Assumption \ref{ass:heteroakm}(ii) were introduced by \cite{adao2019shift}; see Remark \ref{rm:akmrel}.}
As previously noted in Section \ref{sec:heterogss}, linearity has been shown to be a necessary condition for TSLS to preserve a causal interpretation in related contexts.
We reiterate that we therefore impose linearity to examine what additional conditions are needed for TSLS to have a causal interpretation, and not because we view linearity as an inoccuous assumption.
It is also worth emphasizing the connection between Assumption \ref{ass:heteroakm}(ii) and Proposition \ref{prop:homogakm}, which established that linearity was a necessary condition for the instrument to be exogenous in a homogeneous effects model.
As a result, the overidentification test for the homogenous effects model that we develop in Section \ref{sec:test} may be readily adapted to test Assumption \ref{ass:heteroakm}(ii).

Our next proposition obtains a characterization of the weights $\{\omega_{i,n}\}_{i=1}^n$ and derives a necessary and sufficient condition for TSLS to have a causal interpretation.

\begin{proposition}\label{prop:heteroakm}
If Assumption \ref{ass:heteroakm} holds, and $\mathcal G_n = \{ S_i, W_i, \Lambda_i, \beta_i,\varepsilon_i,\eta_i\}_{i=1}^n$, then
\begin{equation}\label{prop:heteroakmdisp}
\omega_{i,n} = \frac{S_i^\prime (\Lambda_i {\rm Var}\{\Sh|\mathcal G_n\})S_i}{\sum_{j=1}^n S_j^\prime (\Lambda_j {\rm Var}\{\Sh|\mathcal G_n\})S_j}.
\end{equation}
Moreover, if in addition Assumption \ref{ass:heteroakmextra} holds, then for $n$ large $\omega_{i,n}$ is positive for all $i$ with probability one (over $\mathcal G_n$) if and only if the distribution of $\Lambda_i{\rm Var}\{\Sh|\mathcal G_n\}$ satisfies
\begin{equation}\label{lm:heteroakmdisp2}
P\left(\sup_{s\in \R^p_+} s^\prime (\Lambda_i{\rm Var}\{\Sh|\mathcal G_n\}) s \leq 0\right) = 1  \text{ or } P\left(\inf_{s\in \R^p_+} s^\prime (\Lambda_i{\rm Var}\{\Sh|\mathcal G_n\}) s \geq 0\right) = 1.
\end{equation}
\end{proposition}

The first part of Proposition \ref{prop:heteroakm} establishes that the weights $\{\omega_{i,n}\}_{i=1}^n$ equal a quadratic form in the shares $S_i$.
Intuitively, the expression for the weights as a quadratic form in the shares implies that the support of the random matrix $\Lambda_i \text{Var}\{\Sh|\mathcal G_n\}$ must be restricted in order for the weights to be positive for any possible realization of the shares.
The second part of Proposition \ref{prop:heteroakm} formalizes this intuition under a requirement that the support of the shares is suitably rich; see Assumption \ref{ass:heteroakmextra} for a formal statement.

Because, unlike the shocks, the shares are always positive, the necessary and sufficient condition for TSLS to have causal interpretation (i.e.\ \eqref{lm:heteroakmdisp2}) is weaker than requiring the random matrix $\Lambda_i\text{Var}\{\Sh|\mathcal G_n\}$ to be positive (or negative) semidefinite with probability one.
The conditions derived by Proposition \ref{prop:heteroakm} are nonetheless restrictive and testable.
Our next corollary, for instance, employs Proposition \ref{prop:heteroakm} to establish that TSLS fails to have a causal interpretation whenever shocks are negatively correlated.

\begin{corollary}\label{cor:heteroimpakm}
Let the conditions of Proposition \ref{prop:heteroakm} hold and suppose $P(\lambda_{ij}=0)=0$ for all $j$
and the support of $\lambda_{ij}/\lambda_{ik}$ is unbounded for any $j\neq k$.
If in addition
\begin{equation}\label{cor:heteroimpakmdisp}
P\left({\rm Cov}\{\Sh_j,\Sh_k|\mathcal G_n\} < 0\right) > 0
\end{equation}
for some $j\neq k$, then the weights $\{\omega_{i,n}\}_{i=1}^n$ are negative with positive probability.
\end{corollary}

Traditionally, the literature has assumed shocks to either be uncorrelated or clustered for asymptotic purposes \citep{adao2019shift,borusyak2022quasi}.
Corollary \ref{cor:heteroimpakm} highlights that the covariance structure of the shocks is important not only for delivering asymptotic approximations and computing standard errors, but also for assigning TSLS a causal interpretation.
Corollary \ref{cor:heteroimpakm} additionally has important implications for applications in which standard errors are clustered.
In particular, since clustered standard errors reflect a concern that shocks within a cluster are correlated, Corollary \ref{cor:heteroimpakm} implies that TSLS can only retain a causal interpretation if in fact all correlations within a cluster are positive conditionally on $\mathcal G_n$.
Empirical researchers clustering standard errors should therefore argue that all correlations within a cluster are positive or, alternatively, rely on a homogeneous effects model.

\begin{remark}\label{rm:heteroakmtest} \rm
Because all entries of $\text{Var}\{\Sh|\mathcal G_n\}$ must be (weakly) positive and the sign of $\Lambda_i$ must be constant  for TSLS to possess a causal interpretation, it follows that
\begin{equation}\label{rm:heteroakmtest1}
\text{sign}\left\{\sum_{i=1}^n E[X_i\dot Z_{i,n}|\mathcal G_n]\right\} = \text{sign}\left\{\sum_{i=1}^n E\left[X_i(\Sh - E[\Sh|\mathcal G_n])^\prime f(S_i,W_i)|\mathcal G_n \right]\right\}
\end{equation}
for any positive $f(S_i,W_i)\in \R^p$ and $\mathcal G_n$ as in Corollary \ref{cor:heteroimpakm}.
In parallel to Remark \ref{rm:heterotest}, we may interpret the restriction in \eqref{rm:heteroakmtest1} as demanding that the sign of the first stage be constant under certain alternative choices of instrument.
Empirically evaluating this restriction can require an estimator for the conditional mean of the shocks.
For this purpose, we note that multiple such estimators have been proposed in the literature as they are also needed to compute the standard errors of the TSLS estimator when conditioning on $\mathcal G_n$ \citep{adao2019shift,borusyak2022quasi}. \qed
\end{remark}

\subsection{Long Panel} \label{sec:longpanel}

Our cross sectional analysis has so far conditioned on either the aggregate shocks or the shares to obtain identification.
Both approaches yield overidentifying restrictions for homogeneous effects models that carry over to short panel settings (i.e.\ $T$ fixed); see Remarks \ref{rm:panelGSS} and \ref{rm:panelakm}.
In this section, we conclude by exploring the implications of employing \emph{unconditional} moment restrictions for identification instead.

For illustrative purposes, we consider a homogeneous effects models with no controls and a scalar endogenous variable.
We assume that at each time period $1\leq t \leq T$ we observe an aggregate shock $\Sh_t$ and variables $\{Y_{it},X_{it},Z_{it}, S_{it}\}_{i=1}^n$ satisfying
\begin{equation*}
Y_{it} = X_{it}\beta + \varepsilon_{it} \hspace{0.5 in}  Z_{it} = S_{it}^\prime \Sh_t.
\end{equation*}
As an identifying assumption, we now employ the just identified moment restriction
\begin{equation}\label{eq:long2}
\frac{1}{nT} \sum_{i=1}^n\sum_{t=1}^T E[\varepsilon_{it} Z_{it}] = 0.
\end{equation}
Crucially, the expectation in \eqref{eq:long2} is taken over both the cross section and the time series -- e.g., over both $\Sh_t$ and $(S_{it},\varepsilon_{it})$.
In particular, the moment restriction in \eqref{eq:long2} contrasts with \cite{goldsmith2020bartik} and \cite{adao2019shift} who instead consider conditional expectations given $\Sh_t$ and $(S_{it},\varepsilon_{it})$ respectively; see Section \ref{sec:overid}.

The natural estimator for $\beta$ continues to be the same TSLS estimator that we have considered so far.
However, relying on the unconditional moment restriction in \eqref{eq:long2} for identification now requires us to let both $n$ and $T$ grow so that we may approximate expectations over both the cross section and the time series.
As a result, when employing the just identified moment restriction in \eqref{eq:long2} to obtain identification, we need to rely on ``long panel" asymptotic approximations and their corresponding standard errors.
The main message of this section is that the long panel standard errors for TSLS decrease at a rate of $1/\sqrt T$.
Therefore, long panel standard errors, and hence the identification restriction in \eqref{eq:long2}, are unlikely to prove informative in applications in which $T$ is small.

In the rest of the section, we provide a summary of the long panel asymptotic properties of the TSLS estimator.
Due to the technical nature of the analysis, we relegate formal statements to the appendix and focus instead on providing intuition for the results.
To this end, we begin by noting that the TSLS estimator $\hat \beta_n$ satisfies
\begin{equation}\label{eq:long3}
\hat{\beta}_n  -\beta = \left(\frac{1}{nT} \sum_{t=1}^T\sum_{i=1}^n X_{it}Z_{it}\right)^{-1} \left(\frac{1}{nT} \sum_{i=1}^n \sum_{t=1}^T \Sh_t^\prime S_{it}\varepsilon_{it}\right).
\end{equation}
As the cross section and time series grow, the denominator in \eqref{eq:long3} converge in probability to some constant ${\rm D}\neq 0$ under standard conditions; i.e., the denominator satisfies
\begin{equation}\label{eq:long4}
\frac{1}{nT} \sum_{t=1}^T\sum_{i=1}^n X_{it}Z_{it} \stackrel{p}{\rightarrow} {\rm D}.
\end{equation}
The asymptotic distribution of the TSLS estimator is therefore governed by the numerator in \eqref{eq:long3}.
In order to derive this asymptotic distribution, we let $\mathcal C_t$ denote all aggregate shocks at time $t$ (which includes $\Sh_t$) and assume that $\{S_{it}\varepsilon_{it}\}_{i=1}^n$ are i.i.d.\ across $i$ (but not $t$) conditionally on the aggregate shocks $\mathcal C_t$.
Defining the variables
\begin{equation*}
\zeta_t \equiv E\left[S_{it}\varepsilon_{it}|\mathcal C_t\right], \hspace{0.4 in}  \nu_{it} \equiv S_{it}\varepsilon_{it} - \zeta_t,
\end{equation*}
we then obtain a decomposition into a time series and a panel process by noting that
\begin{equation}\label{eq:long6}
\frac{1}{nT} \sum_{i=1}^n \sum_{t=1}^T \Sh_t^\prime S_{it} \varepsilon_{it} = \frac{1}{T}\sum_{t=1}^T \Sh_t^\prime \zeta_t + \frac{1}{nT}\sum_{t=1}^T \sum_{i=1}^n \Sh_t^\prime \nu_{it} .
\end{equation}
Crucially, the time series process $\{\Sh_t^\prime \zeta_t\}_{t=1}^T$ has mean zero because of the moment restriction in \eqref{eq:long2}, while the panel process $\{\Sh_t^\prime \nu_{it}\}_{i,t=1}^{n,T}$ has mean zero by construction.
Thus, under appropriate restrictions, we should expect both the time series and panel processes to be asymptotically normally distributed.

The preceding discussion is formalized in the appendix, where we establish:
\begin{proposition}\label{prop:long-panel}
Under the technical conditions presented in Appendix \ref{App:sec:longpanel} we have
\begin{equation}
\left(  \frac{1}{\sqrt{T}}\sum_{t=1}^{T}\mathcal{Z}_{t}^{\prime}\zeta
_{t},\frac{1}{\sqrt{nT}}\sum_{t=1}^{T}\sum_{i=1}^{n}\mathcal{Z}_{t}^{\prime
}\nu_{i,t}\right)  \stackrel{d}{\rightarrow} \left( \mathbb G_{\zeta}, \mathbb G_{\nu}\right) \label{vector-of-interest}%
\end{equation}
as $n,T \to \infty$, with $\mathbb G_{\zeta}$ and $\mathbb G_{\nu}$ independent Gaussian random variables.
\end{proposition}

The long panel asymptotic properties of the TSLS estimator immediately follow from Proposition \ref{prop:long-panel}.
For instance, combining Proposition \ref{prop:long-panel} with results \eqref{eq:long3}, \eqref{eq:long4}, and \eqref{eq:long6} allows us to approximate the variance of $\hat \beta_n$ by the expression
\begin{equation}\label{eq:long7}
\text{Var}\left\{\hat \beta_n - \beta\right\} \approx \left(\frac{1}{{\rm D}}\right)^2\left\{ \frac{1}{T} \text{Var}\{\mathbb G_{\zeta}\} + \frac{1}{nT} \text{Var}\{\mathbb G_\nu\}\right\}.
\end{equation}
In particular, it follows that the long panel standard errors for TSLS decrease at a rate of $1/\sqrt T$.
Intuitively, this dependence on the time dimension results from the dependence of the moment restriction on the distribution of the time series.
One important exception to this phenomenon arises when the time series process has zero variance (i.e.\ $\text{Var}\{\mathbb G_{\zeta}\} = 0$).
This exception corresponds to the case in which $\zeta_t = 0$ or, equivalently,
\begin{equation}\label{eq:long8}
 E[S_{it}\varepsilon_{it}|\mathcal C_t] = 0.
\end{equation}
However, as we discussed in Section \ref{sec:homoggss}, the conditional moment restriction in \eqref{eq:long8} is exactly the identifying assumption made by \cite{goldsmith2020bartik} that renders the model overidentified.
In summary, the long panel analysis suggests that practitioners should either report standard errors based on \eqref{eq:long7} or rely on the short panel identification strategies discussed in Section \ref{sec:overid} and report the corresponding overidentification tests instead.

\section{Overidentification Tests}\label{sec:test}

We next build on our results by developing overidentification tests for the homogeneous effects model studied in Section \ref{sec:overid}.
While we focus on homogeneous effects models due to their demonstrated importance in shift-share designs, we note that it is also possible to test the overidentifying restrictions for heterogeneous effects models derived in Remarks \ref{rm:heterotest} and \ref{rm:heteroakmtest} by adapting the insights in, e.g.,  \cite{bai:2step}.

\subsection{The Approach}\label{sec:clt}

The identification strategies commonly employed for homogeneous effects models yield testable implications in the form of moment equality restrictions.
In this section, we present a general inference approach that is portable across the moment restrictions yielded by different identification strategies.

In what follows, we let $V_{i}\equiv (Y_{i},X_{i},W_{i},S_{i},\mathcal{Z})$ for notational simplicity.
The overidentifying moment restrictions we derived in Section \ref{sec:overid} have the structure
\begin{equation*}
E[\left.  f\left(  V_{i},\theta\right)  \right\vert \mathcal{G}_{n}]=0,
\end{equation*}
where $\theta$ represent some unknown parameter (e.g., $\beta$) and $f(V_{i},\theta)\in\mathbb{R}^{q}$ is a vector valued function.
In many applications, the number of moment equality restrictions is high dimensional in the sense that it is ``large" relative to the sample size.
Because the distributional approximations implicit in chi-squared overidentification tests can be unreliable in high dimensional problems, we instead employ test statistics based on the high dimensional central limit theorem of \cite{chernozhuokov2022improved}.
Specifically, for an estimator $\hat \theta_n$ of $\theta$ and $f_j(V_i,\hat \theta_n)$ denoting the $j^{th}$ coordinate of $f(V_i,\hat \theta_n)\in \R^q$ we set
\begin{equation}\label{test:eq2}
T_n \equiv \max_{1\leq j \leq q} \left|\sum_{i=1}^n f_j\left(V_i,\hat \theta_n\right)\right|
\end{equation}
as the test statistic for the null hypothesis that the moment restrictions hold.

The main assumption we impose to construct a test is that the sample moments be asymptotically linear; see Appendix \ref{section-proof-test} for a formal statement of the assumptions and results for this section.
Intuitively, we require that for some collection of independent mean zero random variables $\{\psi_i\}_{i=1}^{b_n}\subset \R^q$ the test statistic approximately equals
\begin{equation*}
T_n \approx \max_{1\leq j \leq q} \left|\sum_{i=1}^{b_n} \psi_{ij}\right|,
\end{equation*}
where $\psi_{ij}$ denotes the $j^{th}$ coordinate of $\psi_i$.
We refer to the number $b_n$ of random variables $\psi_i$ as the ``effective number of observations," and note that it need not equal the sample size.
Such a distinction is useful, for example, when the data is clustered (in which case $b_n$ equals the number of clusters) or when relying in the asymptotic framework in \cite{adao2019shift} (in which case $b_n$ equals the number of shocks $p$).

We obtain critical values for our test  by relying on a bootstrap approximation.
Specifically, given estimates $\{\hat \psi_i\}_{i=1}^{b_n}$ for $\{\psi_i\}_{i=1}^{b_n}$, we define a bootstrap statistic
\begin{equation*}
T_{n}^{\star}\equiv\max_{1\leq j \leq q}\left| \sum_{i=1}^{b_{n}}\omega_{i}\left(  \hat{\psi
}_{ij}-\frac{1}{b_{n}}\sum_{k=1}^{b_{n}}\hat{\psi}_{kj}\right)  \right|,
\end{equation*}
where $\hat \psi_{ij}$ denotes the $j^{th}$ coordinate of $\hat \psi_i \in \R^q$ and the random weights $\{\omega_i\}_{i=1}^{b_n}$ are drawn independently of the data $\{V_i\}_{i=1}^n$.
For instance, we may set the bootstrap weights $\{\omega_i\}_{i=1}^{b_n}$ to be i.i.d.\ and drawn from a standard normal or Rademacher distribution.\footnote{The Rademacher distirbution corresponds to setting $\omega_i$ to satisfy $P(\omega_i = 1) = P(\omega_i=-1) = 1/2$.}
The bootstrap critical value for a level $\alpha$ test is then given by the bootstrap quantile
\begin{equation*}
\hat{c}_{n}\equiv\inf\{c:P\left(  \left.  T_{n}^{\ast}\leq c\right\vert
\{V_{i}\}_{i=1}^{n}\right)  \geq1-\alpha\};
\end{equation*}
i.e.\ we employ the $1-\alpha$ quantile of $T_n^*$ conditional on the data $\{V_i\}_{i=1}^n$ but unconditionally on the bootstrap weights $\{\omega_i\}_{i=1}^n$.
The critical value $\hat c_n$ can as usual be obtained through simulation by employing multiple draws of the bootstrap weights $\{\omega_i\}_{i=1}^n$ to approximate the distribution of $T_n^*$ conditionally on the data.

Our next result establishes that a test that rejects whenever $T_n$ exceeds the critical value $\hat c_n$ has the correct asymptotic size.
We defer the statement of the relevant regularity conditions to Appendix \ref{section-proof-test}, though note that they allow for the number of moment restrictions $q$ to grow together with the effective sample size $b_n$.

\begin{proposition}\label{prop-clt}
Under the regularity conditions discussed in Appendix \ref{section-proof-test}, we have
$$\lim_{n\to \infty} P(T_{n}>\hat{c}_{n})=\alpha.$$
\end{proposition}

\subsection{Conditioning on Shocks}\label{test:gss}

The general approach discussed in the preceding section can readily be specialized to develop an overidentification test for application that implicitly condition on aggregate shocks in their analysis.
Recall that such applications require the restriction
\begin{equation*}
E\left[ S_i\varepsilon_i|\mathcal G_n\right] = 0,
\end{equation*}
where $\mathcal G_n$ is understood to contain all aggregate variables, including $\Sh$ \citep{goldsmith2020bartik}.
As our test statistic we therefore employ
\begin{equation}\label{test:gss2}
T_n = \max_{1\leq j \leq p} \left|\frac{1}{\hat \sigma_j} \sum_{i=1}^n S_{ij}\hat \varepsilon_i\right|,
\end{equation}
where $S_{ij}$ denotes the $j^{th}$ coordinate of $S_i\in \R^p$, $\{\hat \varepsilon_i\}_{i=1}^n$ are the residuals from the estimated model, and $\hat \sigma_j$ is an estimated weight that allows us to ensure that all $p$ moments are on a comparable scale.

In order to map this setting into the overidentification test of the preceding section, we simply set $f_j(V_i,\hat \theta_n) = S_{ij}\hat \varepsilon_i/\hat \sigma_j$ so that the test statistic in \eqref{test:gss2} becomes a special case of \eqref{test:eq2}.
Setting $A_i \equiv (Z_i,W_i^\prime)^\prime$, $\sigma_j$ to be the probability limit of $\hat \sigma_j$, and defining
\begin{equation*}
\psi_{ij} = \frac{U_{ij}}{\sigma_j} \hspace{0.3 in} U_{ij} \equiv \left(S_{ij}\varepsilon_{i}-E[\left.  S_{ij}\left(  X_{i},W_{i}^{\prime}\right)  \right\vert \mathcal{G}_{n}]\left(  E[A_{i}\left(
X_{i},W_{i}^{\prime}\right)  |\mathcal{G}_{n}]\right)  ^{-1}A_{i}\varepsilon
_{i}\right),
\end{equation*}
it is then possible to show that, under the null hypothesis, $T_n$ approximately equals
\begin{equation*}
T_n \approx \max_{1\leq j \leq p} \left|\sum_{i=1}^n \psi_{ij}\right|
\end{equation*}
provided the observations $\{Y_i,S_i,X_i\}_{i=1}^n$ are i.i.d.\ conditionally on $\mathcal G_n$.\footnote{
See the supplemental appendix for calculations supporting the claims in this section.}
Our inference approach requires an estimator for the influence function $\psi_i$ and to this end we set
\begin{equation*}
\hat{\psi}_{ij}    \equiv \frac{\hat U_{ij}}{\hat\sigma_j} \hspace{0.3 in} \hat U_{ij} \equiv \left( S_{ij}\hat{\varepsilon}_{i}-\left(  \frac{1}{n}%
\sum_{k=1}^{n}S_{kj}\left(  X_{k},W_{k}^{\prime}\right)  \right)  \left(
\frac{1}{n}\sum_{k=1}^{n}A_{k}\left(  X_{k},W_{k}^{\prime}\right)  \right)
^{-1}A_{i}\hat{\varepsilon}_{i}\right).
\end{equation*}
We also let $\hat \sigma_j^2$ be an estimator of the asymptotic variance of the $j^{th}$ moment by setting
\begin{equation*}
\hat \sigma_j^2 = \frac{1}{n}\sum_{i=1}^n \left(\hat U_{ij} - \frac{1}{n}\sum_{k=1}^n \hat U_{kj}\right)^2.
\end{equation*}

Following the approach in the preceding section, we obtain critical values for our test statistic by: (i) Drawing $b\in \{1,\ldots, B\}$ samples $\{\omega_i^{(b)}\}_{i=1}^n$ of standard normal random variables independent of the data; (ii) For each drawn sample $\{\omega_i^{(b)}\}_{i=1}^n$ computing
\begin{equation*}
T_n^{*(b)} \equiv \max_{1\leq j \leq p} \left|\sum_{i=1}^n \omega_i^{(b)} \left(\hat \psi_{ij} - \frac{1}{n}\sum_{k=1}^{n} \hat \psi_{kj}\right)\right|;
\end{equation*}
and (iii) For a level $\alpha$ test, letting $\hat c_{1-\alpha}$ denote the $1-\alpha$ quantile of $\{T_n^{*(b)}\}_{b=1}^B$.
Our overidentification test then rejects whenever $T_n$ exceeds the critical value $\hat c_{1-\alpha}$.

\begin{remark}\label{rm:gsscluter} \rm
The proposed test can easily be adapted to account for clustered data.
Specifically, suppose observations are divided into clusters $\{c_1,\ldots, c_{|C|}\}\equiv C$ -- e.g., $C$ may represent all states, $|C|$ the total number of states, and each cluster $c\in C$ a specific state.
While the test statistic remains the same, for our bootstrap implementation we now set $\hat \psi_{cj} \equiv \sum_{i\in c} \hat U_{ij}/\hat\sigma_j$ for $\hat \sigma_j^2$ the cluter robust sample variance of $\{\hat U_{ij}\}_{i=1}^n$, and let
\begin{equation*}
T_n^* \equiv \max_{1\leq j \leq p} \left|\sum_{c \in C} \omega_c \left(\hat \psi_{cj} - \frac{1}{|C|} \sum_{\tilde c \in C} \hat \psi_{\tilde c j}\right)\right|,
\end{equation*}
with $\{\omega_i\}_{c\in C}$ an i.i.d.\ sample of standard normal random variables drawn independently of the data -- i.e., the bootstrap procedure is modified by simply employing the same draw of $\omega_c$ for all observations in the same cluster.
We note that this procedure maps into the notation of Section \ref{sec:clt} by letting $b_n$ equal the number of clusters.  \qed
\end{remark}

\subsection{Identification Through Shocks}\label{test:akm}

We conclude our discussion of overidentification tests by specializing our approach to applications in which the exogeneity of the Bartik instrument originates from the exogeneity of the shocks.
Recall from Section \ref{sec:homogakm} that the asymptotic framework designed for such applications relies on setting $\mathcal G_n \equiv \{S_i,W_i,\varepsilon_i\}_{i=1}^n$, which yields the restriction
\begin{equation*}
E[Z_i|\{S_i,W_i,\varepsilon_i\}_{i=1}^n] = W_i^\prime \pi_n
\end{equation*}
for $\pi_n$  the population regression coefficient from regressing $\{Z_i\}_{i=1}^n$ on $\{W_i\}_{i=1}^n$ conditionally on $\mathcal G_n$ (as in \eqref{model:eq4}).
As our overidentification test statistic we therefore employ
\begin{equation}\label{test:akm2}
T_n = \max_{1\leq j \leq q} \left|\frac{1}{\hat \sigma_j} \sum_{i=1}^n g_j(\hat \varepsilon_i,W_i,S_i)(Z_i - W_i^\prime \hat \pi_n)\right|,
\end{equation}
where $\{\hat \varepsilon_i\}_{i=1}^n$ are the residuals from the estimated model, $\hat \pi_n$ is the regression coefficient from regressing $\{Z_i\}_{i=1}^n$ on $\{W_i\}_{i=1}^n$, $\hat \sigma_j$ is again an estimated weight, and $g_j(\hat \varepsilon_i,W_i,S_i)$ denotes the $j^{th}$ coordinate of the vector valued function $g(\hat \varepsilon_i,W_i,S_i)\in \R^q$.

In order to obtain a suitable critical value, we first need a characterization of the influence function for each sample moment.
To this end, we define
\begin{align*}
\delta_j & \equiv \left(\sum_{i=1}^n W_iW_i^\prime\right)^{-1}\sum_{i=1}^n W_{i} g_j(\varepsilon_i,W_i,S_i) \notag \\
\kappa_j & \equiv \left(\sum_{i=1}^n E\left[S_i^\prime \mathcal E X_i|\mathcal G_n\right]\right)^{-1} \left(\sum_{i=1}^n E\left[S_i^\prime \mathcal E X_i|\mathcal G_n\right]\frac{\partial} {\partial \varepsilon}g_j(\varepsilon_i,W_i,S_i)\right)
\end{align*}
with $\mathcal E \equiv \Sh - E[\Sh|\mathcal G_n]$, and for $\sigma_j$ the probability limit of $\hat \sigma_j$ we let $\psi_{ij}$ be given by
\begin{equation*}
\psi_{ij} \equiv \frac{U_{ij}}{\sigma_j} \hspace{0.3 in} U_{ij} \equiv  \mathcal E_i \times \sum_{k=1}^n S_{ki}(g_j(\varepsilon_k,W_k,S_k) - W_k^\prime \delta_j - \varepsilon_k \kappa_j).
\end{equation*}
Under conditions similar to those employed by \cite{adao2019shift}, it is then possible to show that under the null hypothesis the statistic $T_n$ satisfies the approximation\footnote{See the supplemental appendix for calculations supporting the claims in this section.}
\begin{equation*}
T_n \approx \max_{1\leq j \leq q} \left|\sum_{i=1}^p \psi_{ij}\right|.
\end{equation*}
We note that, because identification is driven by the exogeneity of the shocks, the effective number of observations equals $p$ and not $n$ -- i.e.\ the sample $\{\psi_{i}\}$ is of size $p$.

Our bootstrap critical value also relies on an estimate of $\psi_i$, and to this end we define
\begin{align*}
\hat \delta_j & \equiv \left(\sum_{i=1}^n W_iW_i^\prime\right)^{-1}\sum_{i=1}^n W_i g_j(\hat \varepsilon_i,W_i,S_i)\\
\hat \kappa_j & \equiv \left(\sum_{i=1}^n (Z_i - W_i^\prime \hat \pi_n)X_i\right)^{-1}\left(\sum_{i=1}^n (Z_i - W_i^\prime \hat \pi_n)X_i \frac{\partial}{\partial \varepsilon} g_j(\hat \varepsilon_i,W_i,S_i)\right),
\end{align*}
and let $\hat {\mathcal E}\in \R^p$ denote a suitable estimator for $\mathcal E\equiv \mathcal Z - E[\mathcal Z|\mathcal G_n]$; see Remark \ref{rm:akmest} for possible choices of $\hat {\mathcal E}$.
As our estimator for the influence function $\psi_{ij}$ we then set
\begin{equation*}
\hat \psi_{ij}\equiv \frac{\hat U_{ij}}{\hat \sigma_j} \hspace{0.3 in} \hat U_{ij} \equiv \hat {\mathcal E}_i \times \sum_{k=1}^n S_{ki}(g_j(\hat \varepsilon_k,W_i,S_i) - W_k^\prime \hat \delta_j - \hat \varepsilon_k \hat \kappa_j).
\end{equation*}
Finally, for the purposes of studentizing the sample moments we set $\hat \sigma_j^2$ to be given by
\begin{equation*}
\hat \sigma_j^2 \equiv \frac{1}{p}\sum_{i=1}^p \left(\hat U_{ij}-\frac{1}{p}\sum_{k=1}^p\hat U_{kj}\right)^2.
\end{equation*}

Given the estimates $\{\hat \psi_{i}\}$, we can obtain asymptotically valid critical values by following the same procedure as in the preceding section.
Specifically, we: (i) Draw $b\in \{1,\ldots, B\}$ samples $\{\omega_i^{(b)}\}_{i=1}^p$ of standard normal random variables independent of the data; (ii) For each drawn sample $\{\omega_i^{(b)}\}_{i=1}^p$ we compute the bootstrap statistic
\begin{equation*}
T_n^{*(b)} \equiv \max_{1\leq j \leq q} \left|\sum_{i=1}^p \omega_i^{(b)} \left(\hat \psi_{ij} - \frac{1}{p}\sum_{k=1}^{p} \hat \psi_{kj}\right)\right|;
\end{equation*}
and (iii) For a level $\alpha$ test, we let $\hat c_{1-\alpha}$ denote the $1-\alpha$ quantile of $\{T_n^{*(b)}\}_{b=1}^B$.
Our overidentification test then rejects whenever $T_n$ exceeds the critical value $\hat c_{1-\alpha}$.

\begin{remark}\label{rm:akmest} \rm
Multiple estimates $\hat {\mathcal E}$ for $\mathcal E = \mathcal Z - E[\mathcal Z|\mathcal G_n]$ have been previously studied in the literature.
\cite{adao2019shift}, for instance, propose the estimator
\begin{equation}\label{rm:akmtest1}
\hat {\mathcal E} \equiv \left(\sum_{i=1}^n S_{i}S_{i}^\prime\right)^{-1}\sum_{i=1}^n S_{it}(Z_{it}-W_{it}^\prime \hat \pi_n).
\end{equation}
Alternatively, in applications in which $W_i$ contains variables equal to $\mathcal Q^\prime S_i$ for some $p\times k$ observable matrix of aggregate variables $\mathcal Q$, \cite{borusyak2022quasi} advocate setting
\begin{equation}\label{rm:akmtest2}
\hat {\mathcal E} \equiv \Sh - \mathcal Q(\mathcal Q^\prime \mathcal Q)^{-1}\mathcal Q^\prime \Sh
\end{equation}
instead.
In our empirical analysis of \cite{david2013china} we employ a ridge regression version of \eqref{rm:akmtest1} because: (i) The matrix $\sum_i S_{i}S_{i}^\prime$ is ill-conditioned in that application,\footnote{\cite{adao2019shift} instead address this ill-conditioning by dropping sectors from the shares vector.} which prevents the use of \eqref{rm:akmtest1}; and (ii) We focus on the specification employed in \cite{david2013china}, which does not contain controls with the structure required by \eqref{rm:akmtest2}. \qed
\end{remark}

\section{Empirical Application: The China Syndrome}\label{sec:china}

We illustrate the empirical relevance of our analysis by revisiting the seminal work of \cite{david2013china} examining the impact of rising Chinese import competition on local US markets.
We focus on the main specification of \cite{david2013china}, in which
\begin{equation}\label{sec:china1}
Y_{it} = X_{it}\beta + W_{it}^\prime \gamma_{\mathtt s} + \varepsilon_{it}
\end{equation}
with $i$ indexing a commuting zone (CZ), $t$ indexing one of two time periods, $Y_{it}$ and $X_{it}$ denoting the change in employment manufacturing and a measure of import exposure, and $W_{it}$ denoting a set of CZ level characteristics.\footnote{Specifically, we employ the specification in column (6) of Table 3 in \cite{david2013china}.}
\cite{david2013china} estimate the model in \eqref{sec:china1} by TSLS based on a Bartik instrument $Z_{it} = \Sh_t^\prime S_{it}$ in which $S_{it}$ and $\Sh_t$ consist of a vector of lagged industry shares and a vector of Chinese imports growth to other countries for different industries.
Their analysis also employs four digit SIC codes to define sectors, leading to a total of $p = 397$ sectors.

The standard errors in \cite{david2013china} rely on an asymptotic approximation that implicitly conditions on the shocks.
As we argued in Section \ref{sec:homoggss}, such standard errors corresponds to an identification strategy in which shares are viewed as exogeneous; i.e.\
\begin{equation}\label{sec:china2}
E[S_{it}\varepsilon_{it}]  = 0 \text{ for all } t.
\end{equation}
We therefore begin by examining the viability of a constant treatment effects model in this application by employing the overidentification test of Section \ref{test:gss}.
Following \cite{david2013china}, we conduct inference by clustering at the state level leading to only 48 effective observations -- a number much smaller than the $p\times T = 794$ number of moment restrictions represented by \eqref{sec:china2}.
In a Monte Carlo study based on this dataset we found the finite sample rejection probabilities of our test to be below the nominal level when employing all 794 restrictions; see the Supplemental Appendix.
We therefore also test the validity of moment restrictions implied by \eqref{sec:china2} by aggregating shares to higher level SIC codes and focusing on particular time periods.

\begin{table}[t!]
\begin{center}
\caption{Overidentification tests corresponding to conditioning on shocks}													\label{tab:chinagss}
\begin{tabular}{lc c c c}
\hline \hline
\multicolumn{1}{c}{Moment Restrictions} & & $\#$ Moments & & p-value \\ \hline
Four Digit SIC $\&$ all time periods    & & 794          & & 0.127  \\
Four Digit SIC $\&$ $t=1$               & & 397          & & 0.115  \\
Four Digit SIC $\&$ $t=2$               & & 397          & & 0.098  \\ \\

Three Digit SIC $\&$ all time periods    & & 272          & & 0.049  \\
Three Digit SIC $\&$ $t=1$               & & 136          & & 0.009  \\
Three Digit SIC $\&$ $t=2$               & & 136          & & 0.072  \\ \\

Two Digit SIC $\&$ all time periods    & & 40          & & 0.005  \\
Two Digit SIC $\&$ $t=1$               & & 20          & & 0.002  \\
Two Digit SIC $\&$ $t=2$               & & 20          & & 0.366  \\
\hline
\hline
\end{tabular}
\end{center}
\begin{flushleft}
\scriptsize{Table reports p-values for overidentification tests of constant treatment effects models identified by implicitly conditioning on the realization of the shocks; see Section \ref{test:gss}.
All tests are implemented by clustering at the state level and weighting commuting zones by their start of period population as in \cite{david2013china}.
}		
\end{flushleft}\vspace{-0.3 in}
\end{table}

Table \ref{tab:chinagss} reports the p-values for our overidentification tests under different choices of moment restrictions.
Overall, our analysis yields evidence against the validity of the constant effects model under an identification strategy that implicitly conditions on shocks.
\cite{goldsmith2020bartik} reach a similar conclusion, though we note that the overidentification tests they employ do not cluster at the state level and are not robust to a ``large" number of moment restrictions.
Given the implausibility of the constant treatment effects model, it is natural to ask whether TSLS can be attributed a causal interpretation through a heterogeneous treatment effects model instead.
Unfortunately, Corollary \ref{cor:heteroimp} establishes that such an interpretation necessitates that shares be uncorrelated across sectors -- a requirement readily contradicted by the data, which exhibits correlations for some sectors in excess of 0.9.
We note, however, that as advocated by \cite{goldsmith2020bartik} it may still be possible to proceed by employing the shares as instruments without combining them into a Bartik instrument.

We next examine the viability of a constant effects model under an identification strategy in which the exogeneity of the Bartik instrument is due to the exogeneity of the shocks \citep{adao2019shift, borusyak2022quasi}.
To this end, we implement the overidentification test of Section \ref{test:akm}, which requires us to select a vector of moments ($\{g_j\}$ in \eqref{test:akm2}) and an estimator for $\mathcal E \equiv \mathcal Z- E[\mathcal Z|\mathcal G_n]$ (as in Remark \ref{rm:akmest}).
For the former, we select a total of twenty moments and let the functions $\{g_j\}$ depend on $\varepsilon$ only.
In particular, we select one function $g_j$ to simply equal $\varepsilon^2$, and the other nineteen to equal the Logit pdf centered at equispaced points in the support of $\varepsilon$; see the Supplemental Appendix for details.
Finally, as discussed in Remark \ref{rm:akmest}, we estimate $\mathcal E$ through Ridge regression and follow \cite{borusyak2022quasi} in clustering at the three level SIC code.

\begin{table}[t!]
\begin{center}
\caption{Overidentification tests corresponding to identification through shocks}\label{tab:chinaakm}													
\begin{tabular}{l  c  cccc}
\hline \hline
            & & \multicolumn{4}{c}{Ridge Penalty}  \\ \cline{3-6}
            & & 1e-3    & 1e-4      & 1e-5      & 1e-6  \\ \hline
p-value     & & 0.0012  & 0.0074    & 0.0368    & 0.065 \\ \hline \hline
\end{tabular}
\end{center}
\begin{flushleft}
\scriptsize{Table reports p-values for overidentification tests of constant treatment effects models identified through the exogeneity of shocks; see Section \ref{test:akm}.
All tests are implemented by clustering at the three digit SIC level and weighting commuting zones by their start of period population as in \cite{david2013china}.}		
\end{flushleft}\vspace{-0.3 in}
\end{table}

Table \ref{tab:chinaakm} reports the p-values for our test corresponding to different choice of Ridge penalties.
Our results show some sensitivity of the p-value to the choice of Ridge penalty.
This conclusion echoes \cite{adao2019shift} who similarly find that the manner in which they address the singularity of the shares design matrix affects their standard errors.\footnote{\cite{adao2019shift} addressed the singularity of $\sum_i S_iS_i^\prime$ by dropping sectors from their analysis. As stated by the authors, their standard errors are affected by what sectors are dropped; see page 3 in https://github.com/kolesarm/ShiftShareSE/blob/master/doc/ShiftShareSE.pdf}
Overall, we interpret Table \ref{tab:chinaakm} as providing evidence against the plausibility of the constant effects model under the identification strategy that attributes the exogeneity of the Bartik instruments to the shocks.
Our preferred choice of penalty for this application is 1e-5, yielding a p-value of 0.0368, which in a Monte Carlo study based on this dataset yielded finite sample rejection probabilities closest to a $5\%$ nominal level; see the Supplemental Appendix.
It is also worth emphasizing that clustering at the three digit SIC level has important implications for whether TSLS can be attributed a causal interpretation in a heterogeneous effects model.
For example, Corollary \ref{cor:heteroimpakm} implies that such a causal interpretation necessitates all industries within the same 3-digit SIC code to be (weakly) positively correlated.
It is important to note, however, that the latter requirement is necessary but not sufficient for TSLS to have a causal interpretation; see Remark \ref{rm:heteroakmtest} for additional restrictions.

Finally, we note that the ``just identified" long panel identification strategy of Section \ref{rm:heteroakmtest} is not viable in this application.
Specifically, because there are only two time periods in \cite{david2013china}, it is not possible to learn the time series distribution of the shocks as required by such an approach.
In other words, it would be foolhardy to invoke an asymptotic approximation based on $T$ diverging to infinity when $T$ equals two.

\section{Summary and Recommendations} \label{sec:recs}

In this paper, we examined the testable implications of identifying restrictions employed to assign a causal interpretation to TSLS based on Bartik instruments.
We developed specification tests for homogeneous effects models that are robust to heteroskedasticity, clustering, and weighting.
Because our tests are based on the high dimensional central limit theorem, we expect them to be more robust than their alternatives to the ``high" degree of overidentification present in shift-share designs.
Finally, we argued that our overidentification tests are of central importance due to the potentially limited empirical scope of the natural alternatives to the homogeneous effects models.

Our analysis has a number of important implications for applied work:
\begin{itemize}
    \item We recommend that practitioners report overidentification tests corresponding to the identification strategy of \cite{goldsmith2020bartik} as well as that of \cite{adao2019shift} and \cite{borusyak2022quasi}.
    These tests can provide evidence as to which identification strategy is more compatible with the data.

    \item Practitioners relying on the identification strategy of \cite{goldsmith2020bartik} should be wary of assigning a causal interpretation to TSLS by relying on a heterogenous effects model.
    Such an interpretation necessitates that shares be uncorrelated across sectors, which is empirically untenable.
    On the other hand, as advocated by \cite{goldsmith2020bartik}, it may still be possible to employ the shares for causal inference by not combining them into a Bartik instrument.

    \item Practitioners should likewise be careful in assigning TSLS a causal interpretation by relying on a heterogeneous effects model under the identification strategy of \cite{adao2019shift} and \cite{borusyak2022quasi}.
    A necessary condition for TSLS to be causal is that shocks to different sectors not be negatively correlated.
    Reasoning through the correlation structure of shocks is thus important not only for the purposes of obtaining standard errors, but also for interpreting the TSLS estimand.

    \item In applications with long panels, it is possible to rely on an alternative identification strategy in which the randomness of the instrument is understood jointly over the cross sectional distribution of shares and the time series distribution of shocks.
    Practitioners relying on this approach, however, need to employ different standard errors that are governed by the length of the panel.
\end{itemize}

\newpage


\begin{center}
    {\LARGE  {\bf Appendix}}
\end{center}

\renewcommand{\thesection}{A.\arabic{section}}
\renewcommand{\theequation}{A.\arabic{equation}}
\renewcommand{\thecorollary}{A.\arabic{corollary}}
\setcounter{theorem}{0}
\setcounter{corollary}{0}
\setcounter{equation}{0}
\setcounter{remark}{0}
\setcounter{section}{0}

This Appendix contains the proofs of all the results in the main text.

\section{Proofs for Section \ref{sec:overid}}


\noindent \emph{Proof of Proposition \ref{prop:homoggss}.} First note that since $\{S_i,W_i,\varepsilon_i\}_{i=1}^n$ are i.i.d.\ and independent of $\mathcal G_n$ and $\Sh$ is measurable with respect to $\mathcal G_n$ it follows $E[W_i\varepsilon_i] = 0$ that
\begin{equation*}
\frac{1}{n}\sum_{i=1}^n E\left[\varepsilon_i(Z_i-W_i^\prime \pi_n)|\mathcal G_n\right]  = \Sh^\prime [ S_i \varepsilon_i].
\end{equation*}
It is therefore immediate that $E[S_i\varepsilon_i] = 0$ is a sufficient condition for \eqref{prop:homoggssdisp} to hold with probability one over $\Sh$.
To see that this condition is also necessary, suppose by way of contradiction that $E[S_i\varepsilon_i]\neq 0$.
By Lemma \ref{lm:obvious}, it then follows that there is a $z^*$ in the support of $\Sh$ satisfying $(z^*)^\prime E[S_i\varepsilon_i]\neq 0$.
Assuming without loss of generality that $(z^*)^\prime E[S_i\varepsilon_i] > 0$ and setting $B_\delta(z^*)\equiv \{z\in \R^p : \|z-z^*\|_2^2 < \delta\}$, we then obtain that $(E[S_i\varepsilon_i])^\prime z > 0$ for all $z\in B_\delta(z^*)$ provided that $\delta > 0$ is chosen to be sufficiently small.
Hence, employing that $z^*$ is in the support of $\Sh$ yields that
\begin{equation*}
P\left(\Sh^\prime E[S_i\varepsilon_i]\neq 0\right) \geq P\left(\Sh \in B_\delta(z^*)\right) > 0,
\end{equation*}
which implies that for \eqref{prop:homoggssdisp} to hold with probability one we must have $E[S_i\varepsilon_i] = 0$. \qed

\noindent \emph{Proof of Proposition \ref{prop:homogakm}.} 
First note that since $\mathcal G_n = \{S_i,W_i,\varepsilon_i\}_{i=1}^n$ and $\Sh$ is independent of $\mathcal G_n$ by hypothesis, it follows from $Z_i = S_i^\prime \Sh$ by definition that we have
\begin{equation*}
\frac{1}{n}\sum_{i=1}^nE \left[\varepsilon_i\left(Z_i - W_i^\prime \pi_n\right)|\mathcal G_n\right] = \frac{1}{n}\sum_{i=1}^n \varepsilon_i\left(S_i^\prime E[\Sh] - W_i^\prime \pi_n\right).    
\end{equation*}
Since $E[Z_i|\{S_i,W_i,\varepsilon_i\}_{i=1}^n] = S_i^\prime E[\Sh]$, we can conclude that $E[Z_i|\{S_i,W_i,\varepsilon_i\}_{i=1}^n] = W_i^\prime \pi_n$ is a sufficient condition for \eqref{prop:homogakmdisp} to hold with probability one.
In order to establish the converse direction, we next define the event $A_n$ to be given by 
\begin{equation*}
A_n \equiv \left\{ S_i^\prime E[\Sh] \neq W_i^\prime \pi_n \text{ for some } 1\leq i \leq n \right\}
\end{equation*}
and note that $A_n$ is measurable with respect to $\{S_i,W_i\}$ by definition of $\pi_n$.
For any realization of $\{S_i,W_i\}_{i=1}^n$ in the event $A_n$, Lemma \ref{lm:obvious} allows us to conclude that there exists a $\{e_i^*\}_{i=1}^n$ in the support of $\{\varepsilon_i\}_{i=1}^n$ conditional on $\{S_i,W_i\}$ satisfying
\begin{equation*}
\frac{1}{n}\sum_{i=1}^n e_i^*\left(S_i^\prime E[\Sh] - W_i^\prime \pi_n\right) \neq 0.
\end{equation*}
Hence, defining $B_\delta(\{e_i^*\}_{i=1}^n) \equiv \{ \{e_i\}_{i=1}^n : \sum_i (e_i - e_i^*)^2 < \delta^2\}$ we obtain from $\{e_i^*\}_{i=1}^n$ being in the support of $\{\varepsilon_i\}_{i=1}^n$ conditional on $\{S_i,W_i\}_{i=1}^n$ that for $\delta$ sufficiently small 
\begin{multline}\label{prop:homogakm4}
P\left(\frac{1}{n}\sum_{i=1}^n \varepsilon_i\left(S_i^\prime E[\Sh] - W_i^\prime \pi_n\right)\neq 0 \Big|\{S_i,W_i\}_{i=1}^n\right) \\ \geq P\left(\{\varepsilon_i\}_{i=1}^n \in B_\delta(\{e_i^*\}_{i=1}^n)|\{S_i,W_i\}_{i=1}^n \right) > 0.
\end{multline}
In particular, since the conclusion in \eqref{prop:homogakm4} applies to any $\{S_i,W_i\}_{i=1}^n$ in the event $A_n$, it follows that $A_n$ must have probability zero in order for condition \eqref{prop:homogakmdisp} to hold with probability one.
Since $S_i^\prime E[\Sh] = E[Z_i|\{S_i,W_i,\varepsilon_n\}_{i=1}^n]$, the lemma follows. \qed

\begin{lemma}\label{lm:obvious}
If the linear span of a set $V\subseteq \R^k$ has dimension $k$, then it follows that for any $0\neq x\in \R^k$ there is a $v^*\in V$ such that $x^\prime v^*\neq 0$.
\end{lemma}

\noindent {\emph{Proof}.} Since the dimension of the linear span of $V$ equals $k$, it follows that there are vectors $\{v_j\}_{j=1}^k \subset \R^k$ and scalars $\{\alpha_j\}_{j=1}^k$ satisfying $x = \sum_j \alpha_j v_j$. 
In particular, since $x\neq 0$, we have $(\sum_j \alpha_j v_j)^\prime x = \|x\|_2^2 > 0$ and hence that $v_{j^*}^\prime x \neq 0$ for some $j^*$. \qed

\section{Proofs for Section \ref{sec:heterogss}}


The results in Section \ref{sec:heterogss} rely on the next set of regularity conditions. 
In the statement below, the notation $\text{supp}\{V\}$ refers to the support of a random variable $V$.

\begin{assumption}\label{ass:heterogssextra}
(i) $\{\Lambda_i,W_i\}_{i=1}^n$ are i.i.d.\ conditional on $\Sh$;
(ii) ${\rm supp}\{(\Lambda_i,W_i,\Sh)\} = {\rm supp}\{\Lambda_i\}\times {\rm supp}\{W_i\}\times {\rm supp}\{\Sh\}$;
(iii) ${\rm rank}\{{\rm Var}\{S_i|W_i\}\} \geq p-1$ and ${\rm Var}\{S_{ij}|W_i\} > 0$ for all $1\leq j \leq p$ with probability one;
(iv) $(\lambda_{i1},\ldots,\lambda_{ip})$ has a continuous density w.r.t.\ Lebesgue measure;
(v) ${\rm supp}\{\Sh\}$ includes a neighborhood of zero.
\end{assumption}

Assumption \ref{ass:heterogssextra}(i) allows for dependence across observations though it implies that the marginal distribution of $(\Lambda_i,W_i,\Sh)$ is the same for all observations.
Assumptions \ref{ass:heterogssextra}(ii)(iv) further impose support restrictions on $(\Lambda_i,W_i,\Sh)$, while Assumption \ref{ass:heterogssextra}(iii) states our restrictions on the covariance matrix of $S_i \in \R^p$ conditionally on $W_i$ (denoted $\text{Var}\{S_i|W_i\})$.
Assumption \ref{ass:heterogssextra}(iii) allows for $\text{Var}\{S_i|W_i\}$ to be rank-deficient to recognize that the shares may sum up to one.
Finally, Assumptions \ref{ass:heterogssextra}(iv) simply demands that the ``types" vector  $(\lambda_{i1},\ldots,\lambda_{ip})$ be continuously distributed.

\noindent \emph{Proof of Proposition \ref{prop:heterogss}.} First note that Assumption \ref{ass:heterogss}(ii) implies  $(\Lambda_i,\beta_i,\varepsilon_i,\eta_i)$ is independent of $(\Sh,S_i)$ conditionally on $(\Sh,W_i)$. Therefore, by Assumptions \ref{ass:heterogss}(i)(iii), we may apply Lemma \ref{lm:decomp} with $\mathcal G_n = \{\Sh\}$ to conclude that
\begin{multline}\label{prop:heterogss3}
E\left[X_i \dot Z_{i,n}|\Lambda_i, W_i,\Sh\right] =  E\left[(\Sh^\prime \Lambda_iS_i)\dot Z_{i,n}\Big|\Lambda_i,W_i,\Sh\right]    \\
= \Sh^\prime \Lambda_i E\left[ S_i(S_i^\prime - E[S_i^\prime|W_i,\Sh])|\Lambda_i,W_i,\Sh \right]\Sh = \Sh^\prime (\Lambda_i \text{Var}\{S_i|W_i\}) \Sh,
\end{multline}
where the second equality holds by $Z_i = S_i^\prime \Sh$ and Assumption \ref{ass:heterogss}(iii), and the final equality holds by Assumption \ref{ass:heterogss}(ii).
The claim of the proposition therefore follows from result \eqref{prop:heterogss3}, the definition of $\omega_{i,n}$ in \eqref{model:eq16}, and Lemma \ref{lm:sign}. \qed

\noindent \emph{Proof of Corollary \ref{cor:heteroimp}.} 
Suppose that condition \eqref{cor:heteroimpdisp} holds for some $j\neq k$ and define
\begin{equation*}
\sigma_{l}^2(W_i) \equiv \text{Var}\{S_{il}|W_i\} \hspace{0.5 in} \sigma_{jk}(W_i) = {\rm Cov}\{S_{ij},S_{ik}|W_i\}.
\end{equation*}
Next, note that by Proposition \ref{prop:heterogss} a necessary condition for $\omega_{i,n}$ to be positive for all $1\leq i \leq n$ with probability one is that the corresponding $2\times 2$ matrix 
\begin{equation*}
\left(\begin{array}{cc} \lambda_{ij} & 0 \\ 0 & \lambda_{ik}\end{array}\right) \left(\begin{array}{cc} \sigma_j^2(W_i) & \sigma_{jk}(W_i) \\ \sigma_{jk}(W_i) & \sigma_k^2(W_i)\end{array}\right) 
= \left(\begin{array}{cc} \lambda_{ij} \sigma^2_j(W_i) & \lambda_{ij}\sigma_{jk}(W_i)\\ \lambda_{ik}\sigma_{jk}(W_i) & \lambda_{ik}\sigma^2_k(W_i)\end{array}\right)
\end{equation*}
be either positive semi-definite or negative semi-definite with probability one.
Hence, setting $\bar \lambda_i = (\lambda_{ij} + \lambda_{ik})/2$ and noting that for any $2\times 2$ matrix $A$ and $a\in \R^2$ we have $a^\prime A a = a^\prime (A+A^\prime)a/2$, it follows that a necessary condition for $\omega_{i,n}$ to be positive for all $1 \leq i \leq n$ with probability one is that the matrices $\Omega_i$ defined by
\begin{equation*}
\Omega_i \equiv \left(\begin{array}{cc} \lambda_{ij} \sigma^2_j(W_i) & \bar \lambda_i \sigma_{jk}(W_i)\\ \bar \lambda_i \sigma_{jk}(W_i) & \lambda_{ik}\sigma^2_k(W_i)\end{array}\right)
\end{equation*}
be either positive semi-definite with probability one or negative semi-definite with probability one.
However, in order for $\Omega_i$ to be positive semi-definite or negative-semidefinite with probability one we must have $\lambda_{ij}\lambda_{ik}\geq 0$ due to $\min\{\sigma_j^2(W_i),\sigma_k^2(W_i)\} > 0$ by Assumption \ref{ass:heterogssextra}(iii). 
Hence, since the determinant of $\Omega_i$ equals the product of its eigenvalues, a necessary condition for $\omega_{i,n}$ to be positive for all $i$ with probability one is 
\begin{equation*}
1 = P(\text{det}\{\Omega_i\} \geq 0 \text{ and } \lambda_{ij}\lambda_{ik} \geq 0).
\end{equation*}
However, setting $\rho_{jk}(W_i) \equiv \sigma_{jk}(W_i)/\sigma_j(W_i)\sigma_k(W_i)$ we obtain by direct calculation that
\begin{multline*}
 P\left(\text{det}\{\Omega_i\} \geq 0 \text{ and } \lambda_{ij}\lambda_{ik} \geq 0\right) = 
P\left(\lambda_{ij}\lambda_{ik} - \bar \lambda_i^2 \rho_{jk}^2(W_i) \geq 0 \text{ and } \lambda_{ij}\lambda_{ik} \geq 0\right) \\
\leq  P\left(1 \geq \frac{\rho_{jk}^2(W_i)}{4}\left(\frac{\lambda_{ij}}{\lambda_{ik}} + \frac{\lambda_{ik}}{\lambda_{ij}} + 2\right)\right)  < 1,
\end{multline*}
where the first inequality follows from $\lambda_{ij}\lambda_{ik}\neq 0$ with probability one by Assumption \ref{ass:heterogssextra}(iv), and the final from Assumption \ref{ass:heterogssextra}(ii), the support of $\lambda_{ij}/\lambda_{ik}$ being unbounded, and $\rho^2_{jk}(W_i) \neq 0$ with positive probability due to condition \eqref{cor:heteroimpdisp}. \qed

\begin{lemma}\label{lm:decomp}
Suppose that $\Sh$ and $\{Y_i,X_i,W_i,S_i\}_{i=1}^n$ satisfy equations \eqref{model:eq13} and \eqref{model:eq14}, $E[Z_i|W_i,\mathcal G_n] = W_i^\prime \pi_n$, and $(\Lambda_i,\beta_i,\varepsilon_i,\eta_i)\indep (\Sh,S_i)$ conditionally on $(W_i,\mathcal G_n)$. Then:
\begin{equation}\label{lm:decomp:disp}
E\left[X_i\dot Z_{i,n}\Big|\Lambda_i,W_i,\mathcal G_n\right] = E[(\Sh^\prime \Lambda_iS_i)\dot Z_{i,n}|\Lambda_i,W_i,\mathcal G_n].
\end{equation}
If in addition $\sum_{i=1}^n E[X_i \dot Z_{i,n}|\mathcal G_n]\neq 0$, then it also follows that equation \eqref{model:eq16} holds.
\end{lemma}

\noindent {\emph{Proof}.} 
First note that for any function $f$ of $(\Lambda_i,\beta_i,\varepsilon_i,\eta_i,W_i)$ it follows from the requirements that $(\Lambda_i,\beta_i,\varepsilon_i,\eta_i) \indep (\Sh,S_i)$ conditionally on $(\mathcal G_n,W_i)$ and $Z_i = \Sh^\prime S_i$ that
\begin{multline}\label{lm:decomp1}
E\left[\dot Z_{i,n}f(\Lambda_i,\beta_i,\varepsilon_i,\eta_i,W_i)\Big|\mathcal G_n\right] \\ =  E\left[E\left[\dot Z_{i,n}|W_i,\mathcal G_n\right]E\left[f(\Lambda_i,\beta_i,\varepsilon_i,\eta_i,W_i)|W_i, \mathcal G_n\right]\Big| \mathcal G_n\right] =0,   
\end{multline}
where the final equality follows from $E[Z_i|W_i,\mathcal G_n] = W_i^\prime \pi_n$.
Combining result \eqref{lm:decomp1} with the specifications of the first and second stages in \eqref{model:eq13} and \eqref{model:eq14} then yields that
\begin{equation}\label{lm:decomp2}
 E\left[Y_i \dot Z_{i,n} \Big|\mathcal G_n\right]  =  E\left[X_i \dot Z_{i,n}\beta_i\Big|\mathcal G_n\right] = E\left[(\Sh^\prime \Lambda_i S_i) \dot Z_{i,n}\beta_i\Big|\mathcal G_n\right].
\end{equation}
Letting $\Sh_j$ denote the $j^{th}$ coordinate of $\Sh$, $S_{ij}$ the $j^{th}$ coordinate of $S_i$, and recalling that $\Lambda_i = \text{diag}\{(\lambda_{i1},\ldots,\lambda_{ip})\}$ we may also employ \eqref{model:eq14} and \eqref{lm:decomp1} to obtain
\begin{align}\label{lm:decomp3}
E\left[X_i\dot Z_{i,n}|\Lambda_i,W_i,\mathcal G_n\right] & = E\left[(\Sh^\prime \Lambda_i S_i)\dot Z_{i,n}|\Lambda_i,W_i,\mathcal G_n\right] \notag     \\
& = \sum_{j=1}^p E\left[ (\Sh_j\lambda_{ij}S_{ij})\dot Z_{i,n}\Big|\Lambda_i,W_i,\mathcal G_n\right],
\end{align}
which verifies that \eqref{lm:decomp:disp} indeed holds. 
Finally, observe that $(\Lambda_i,\beta_i)$ being independent of $(\Sh,S_i)$ conditionally on $(W_i,\mathcal G_n)$ further allows us to conclude that
\begin{align}\label{lm:decomp4}
E[(\Sh^\prime \Lambda_iS_i)\dot Z_{i,n}\beta_i|\mathcal G_n] 
& = \sum_{j=1}^pE[E[\Sh_j  S_{ij}\dot Z_{i,n}|W_i,\mathcal G_n]E[\lambda_{ij}\beta_i|W_i,\mathcal G_n]|\mathcal G_n] \notag \\
& = \sum_{j=1}^pE[E[\Sh_j  S_{ij} \dot Z_{i,n}|W_i,\mathcal G_n]\lambda_{ij}E[\beta_i|\Lambda_i,W_i,\mathcal G_n]|\mathcal G_n]\notag \\
& = E[X_i \dot Z_{i,n}|
\Lambda_i,W_i,\mathcal G_n]E[\beta_i|\Lambda_i,W_i,\mathcal G_n]|\mathcal G_n],
\end{align}
where the final equality follows from result \eqref{lm:decomp3} and $\Lambda_i \indep (\Sh,S_i)$ conditionally on $(W_i,\mathcal G_n)$.
The claim that \eqref{model:eq16} holds therefore follows from the definition of $\beta_{0,n}$, the rank condition $\sum_{i=1}^n E[X_i\dot Z_{i,n}|\mathcal G_n]\neq 0$ and results \eqref{lm:decomp2}, \eqref{lm:decomp3}, and \eqref{lm:decomp4}. \qed

\begin{lemma}\label{lm:sign}
Let Assumption \ref{ass:heterogssextra} hold. Then, for $n$ sufficiently large, it follows 
\begin{equation}\label{lm:signdisp}
P\left(\min_{1\leq i \leq n} \Sh^\prime (\Lambda_i {\rm Var}\{S_i|W_i\})\Sh\} \geq 0 \text{ or } \max_{1\leq i \leq n} \Sh^\prime (\Lambda_i {\rm Var}\{S_i|W_i\})\Sh\} \leq 0 \right) = 1
\end{equation}
if and only if either $\Lambda_i {\rm Var}\{S_i|W_i\}$ is positive semi-definite with probability one for all $i$, or $\Lambda_i {\rm Var}\{S_i|W_i\}$ is negative semi-definite with probability one for all $i$.
\end{lemma}

\noindent \emph{Proof.} For notational simplicity we first set $\Gamma_i \equiv \Lambda_i \text{Var}\{S_i|W_i\}$ and use ``$\Gamma_i\geq 0$" and ``$\Gamma_i\leq 0$" to denote that $\Gamma_i$ is positive semi-defintie and negative-semidefinite respectively. 
Then note that if either $P(\Gamma_i \geq 0) = 1$ for all $i$  or $P(\Gamma_i \leq 0) = 1$ for all $i$, then  \eqref{lm:signdisp} immediately holds.
For the opposite direction, we next observe that
\begin{align}\label{lm:sign1}
P& \left( \min_{1\leq i \leq n} \Sh^\prime \Gamma_i \Sh \geq 0 \text{ or } \max_{1\leq i \leq n} \Sh^\prime \Gamma_i \Sh \leq 0\right) \notag \\
& = P\left(\min_{1\leq i \leq n} \Sh^\prime \Gamma_i \Sh \geq 0\right) + P\left(\max_{1\leq i \leq n} \Sh^\prime \Gamma_i \Sh \leq 0\right) - P\left(\Sh^\prime \Gamma_i \Sh = 0 \text{ for all }1\leq i \leq n\right) \notag \\
& = E\left[P^n\left(\Sh^\prime \Gamma_1 \Sh \geq 0|\Sh\right) + P^n\left(\Sh^\prime \Gamma_1 \Sh \leq 0|\Sh\right) - P^n\left(\Sh^\prime \Gamma_1 \Sh = 0 |\Sh\right)\right],
\end{align}
where the second equality follows from Assumption \ref{ass:heterogssextra}(i).
However, since $a^n = o(1) $ for any $a\in [0,1)$ as $n$ tends to infinity, result \eqref{lm:sign1} implies that in order for the claim in \eqref{lm:signdisp} to hold for $n$ sufficiently large we must have with probability one over $\Sh$ that
\begin{equation}\label{lm:sign2}
\max\{P(\Sh^\prime \Gamma_1\Sh \geq 0 |\Sh), P(\Sh^\prime \Gamma_1\Sh \leq 0|\Sh)\} = 1 .
\end{equation}
By Assumption \ref{ass:heterogssextra}(ii), $\supp\{(\Gamma_1,\Sh)\} = \supp\{\Gamma_1\}\times \supp\{\Sh\}$ and hence the distribution of $\Gamma_1$ is absolutely continuous with respect to the distribution of $\Gamma_1$ conditionally on $\Sh$. Defining the sets $S_+$ and $S_{-}$ to be given by
\begin{align}\label{lm:sign3}
    S_+ & \equiv \{z\in \R^p : P(z^\prime \Gamma_1 z \geq 0) = 1\} \notag\\
    S_- & \equiv \{z\in \R^p : P(z^\prime \Gamma_1 z \leq 0) = 1\} 
\end{align}
it therefore follows from result \eqref{lm:sign2} holding with probability one that we must have
\begin{equation}\label{lm:sign4}
P(\Sh \in S_+  \cup S_{-}) = 1.    
\end{equation}

Next, set $\Gamma^* \equiv \Lambda^* \Sigma^*$ with $\Sigma^*$ any point in the support of ${\rm Var}\{S_1|W_1\}$ and $\Lambda^* = \text{diag}\{\lambda_{1}^*,\ldots, \lambda_p^*\}$ for any point $(\lambda_1^*,\ldots, \lambda_p^*)$ at which the density of $(\lambda_{11},\ldots, \lambda_{1p})$ is strictly positive.
We next aim to show that $\Gamma^*$ is either positive semi-definite or negative semi-definite.
We proceed by contradiction: Suppose that $\Gamma^*$ is neither positive semi-definite nor negative semi-definite. 
Then, by Lemma \ref{aux:lin}, Assumption \ref{ass:heterogssextra}(iii), the density of $(\lambda_{11},\ldots,\lambda_{1p})$ being continuous and strictly positive at $(\lambda_1^*,\ldots, \lambda_p^*)$, and Assumption \ref{ass:heterogssextra}(v), there are $z_0 \in {\rm \supp}\{\Sh\}$  and $\Lambda_0,\tilde \Lambda_0 \in \supp\{\Lambda_1\}$ satisfying 
\begin{equation*}
z_0^\prime (\Lambda_0 \Sigma^*) z_0 > 0 \hspace{0.5 in }    z_0^\prime (\tilde \Lambda_0 \Sigma^*) z_0 < 0.
\end{equation*}
We may therefore select sufficiently small neighborhoods $N(z_0), N(\Lambda_0), N(\tilde \Lambda_0), N(\Sigma^*)$ of $z_0,\Lambda_0,\tilde \Lambda_0$, and $\Sigma^*$ such that the following inequalities are satisfied
\begin{align*}
z^\prime(\Lambda \Sigma) z > 0 & \text{ for all } (z,\Lambda,\Sigma) \in N(z_0)\times N(\Lambda_0)\times N(\Sigma^*) \notag \\        
z^\prime(\Lambda \Sigma) z < 0 & \text{ for all } (z,\Lambda,\Sigma) \in N(z_0)\times N(\tilde \Lambda_0)\times N(\Sigma^*). 
\end{align*}
Then note that since $\Lambda_0\Sigma^* \in \supp\{\Lambda_1\text{Var}\{S_1|W_1\}\}$  by Assumption \ref{ass:heterogssextra}(ii), it follows
\begin{multline}\label{lm:sign7}
P(z^\prime (\Lambda_1 \text{Var}\{S_1|W_1\}) z > 0 \text{ for all } z\in N(z_0)) \\ \geq P((\Lambda_1,\text{Var}\{S_1|W_1\}) \in N(\Lambda_0)\times N(\Sigma^*)) > 0.
\end{multline}
By identical arguments, but relying on $N(\tilde \Lambda_0)\times N(\Sigma^*)$ instead, similarly yield that
\begin{equation}\label{lm:sign8}
P(z^\prime (\Lambda_1 \text{Var}\{S_1|W_1\}) z < 0 \text{ for all } z\in N(z_0)) > 0.
\end{equation}
However, because $P(\Sh \in N(z_0)) > 0$ due to $z_0 \in \supp\{\Sh\}$, results \eqref{lm:sign7} and \eqref{lm:sign8} together contradict \eqref{lm:sign4} and therefore $\Gamma^*$ must be either positive semi-definite or negative semi-definite as claimed.
Thus, we have so far shown that
\begin{equation}\label{lm:sign9}
    P(\Gamma_1 \geq 0 \text{ or } \Gamma_1 \leq 0 ) =1.
\end{equation}

To conclude, note that by result \eqref{lm:sign4} and Lemma \ref{aux:lindep} either $S_+$ or $S_{-}$ (or both) must contain $p$ linearly independent vectors. 
If $S_+$ contains $p$ linearly independent vectors $\{s_j\}_{j=1}^p$, then we obtain by the definition of $S_+$ in \eqref{lm:sign3} that
\begin{multline}\label{lm:sign10}
P(\Gamma_1 = 0) \leq P(\Gamma_1 \leq 0) \leq P(z^\prime \Gamma_1 z \leq 0 \text{ for all } z\in S_+) \\ \leq P(s_j^\prime \Gamma_1 s_j \leq 0 \text{ for } 1\leq j \leq p) = P(s_j^\prime \Gamma_1 s_j = 0 \text{ for } 1\leq j \leq p).
\end{multline}
Next note that $z^\prime \Gamma_1 z = z^\prime(\Gamma_1 + \Gamma_1^\prime)z/2$ for any $z\in \R^p$ and use that $(\Gamma_1 + \Gamma_1^\prime)$ is diagonalizable and $s_j^\prime (\Gamma_1 + \Gamma_1^\prime)s_j = 0$ for all $1\leq j \leq p$ implies that all the eigenvalues of $(\Gamma_1+\Gamma_1^\prime)$ equal zero due to $\{s_j\}_{j=1}^p$ being linearly independent to conclude that
\begin{multline}\label{lm:sign11}
P(s_j^\prime \Gamma_1 s_j = 0 \text{ for } 1\leq j \leq p) = P(s_j^\prime(\Gamma_1 + \Gamma_1^\prime)s_j = 0 \text{ for } 1\leq j \leq p) \\
= P(\Gamma_1 = -\Gamma_1^\prime) \leq P((\lambda_{11},\ldots,\lambda_{1p}) = 0) \leq P(\Gamma_1 =0),    
\end{multline}
where the first inequality follows from $\text{Var}\{S_{1j}|W_i\}>0$ for all $1\leq j \leq p$ by Assumption \ref{ass:heterogssextra}(iii) and the second inequality by $\Gamma_1 = \text{diag}\{(\lambda_{11},\ldots, \lambda_{1p})\}\text{Var}\{S_1|W_1\}$.
Hence, combining results \eqref{lm:sign10} and \eqref{lm:sign11} we obtain that if $S_+$ contains $p$ linearly independent vectors, then $P(\Gamma_1 = 0) = P(\Gamma_1 \leq 0)$, which together with result \eqref{lm:sign9} yields
\begin{equation*}
1 = P(\Gamma_1 \geq 0 \text{ or } \Gamma_1 \leq 0) = P(\Gamma_1 \geq 0) + P(\Gamma_1\leq 0) - P(\Gamma_1=0) = P(\Gamma_1\geq 0).
\end{equation*}
Similarly, if $S_{-}$ instead contains $p$ linearly independent vectors, then it is possible to show that \eqref{lm:sign9} implies that $P(\Gamma_1\leq 0)=1$ and therefore the lemma follows. \qed

\begin{lemma}\label{aux:lin}
Let $\Gamma \equiv D \Sigma$ with $D$ a $p\times p$ diagonal matrix and $\Sigma$ a $p\times p$ positive semi-definite matrix with ${\rm rank}\{\Sigma\}\geq p-1$.
If $\Gamma$ is neither positive semi-definite nor negative semi-definite, then for any $\delta > 0$  there is a $z\in \R^p$ and diagonal $p\times p$ matrices $H_1,H_2$ satisfying $\|z\| \vee \|D-H_1\| \vee \|D-H_2\| < \delta$, $z^\prime (H_1\Sigma) z > 0$, and $z^\prime (H_2\Sigma) z < 0$.
\end{lemma}

\noindent \emph{Proof.} Let $\bar \Gamma = (\Gamma + \Gamma^\prime)/2$ and note that since $a^\prime \Gamma a = a^\prime \Gamma^\prime a$ for any $a \in \R^p$, we have 
\begin{equation}\label{aux:lin1}
a^\prime \Gamma a = a^\prime \bar \Gamma a \text{ for all } a \in \R^p.
\end{equation}
Next observe that $\bar \Gamma$ is a symmetric matrix and let $\{\pi_i\}_{i=1}^p$ denote its eigenvalues and $\{v_i\}_{i=1}^p$ its corresponding eigenvectors.
If $\Gamma$ is neither positive semi-definite nor negative semi-definite, then result \eqref{aux:lin1} implies that $\bar \Gamma$ must have a strictly positive eigenvalue and a strictly negative eigenvalue.
Assuming without loss of generality that $\pi_1 >0$ and $\pi_2 < 0$, let $\alpha >0$ and define the vectors $c_1$ and $c_2$ to be given by
\begin{equation*}
c_1 \equiv \alpha v_1 + \alpha \frac{\sqrt{\pi_1}}{\sqrt{|\pi_2|}}v_2 \hspace{0.5 in} c_2 \equiv \alpha v_1 - \alpha \frac{\sqrt{\pi_1}}{\sqrt{|\pi_2|}}v_2.
\end{equation*}
Next note that $\|c_1\|^2 = \|c_2\|^2 = \alpha^2(1+ \pi_1/|\pi_2|)$ and that therefore we may set $\|c_1\| =\|c_2\| < \delta$ by selecting $\alpha$ to be sufficiently small.
Moreover, by direct calculation we have
\begin{equation}\label{aux:lin3}
c_1^\prime \bar \Gamma c_1 = c_2^\prime \bar \Gamma c_2 = \alpha^2 \pi_1 + \alpha^2 \frac{\pi_1}{|\pi_2|}\pi_2 = 0,   
\end{equation}
where the final equality follows from $\pi_2 < 0$.
Next observe that since ${\rm rank}\{\Sigma^{1/2}\} = {\rm rank}\{\Sigma\} \geq p-1$ and $c_1,c_2\in \R^p$ are linearly independent, it follows that $\Sigma^{1/2} c_j \neq 0$ for some $j\in\{1,2\}$.
Assuming without loss of generality that $\Sigma^{1/2} c_1 \neq 0$, we then set $z = c_1$ and note that results \eqref{aux:lin1}, \eqref{aux:lin3}, and $\Sigma^{1/2}z \neq 0$ imply that
\begin{equation}\label{aux:lin4}
z^\prime \Gamma z =  0 \text{ and } z^\prime \Sigma z > 0.    
\end{equation}
Setting $x \equiv \Sigma z$ and letting $x_j$ and $z_j$ denote the $j^{th}$ coordinates of $x$ and $z$, then note
\begin{equation*}
0 < z^\prime\Sigma z = \sum_{j=1}^p z_j x_j,    
\end{equation*}
which implies that $z_{j} x_{j} > 0$ for some $1\leq j \leq p$. 
Assuming without loss of generality that $z_1 x_1 > 0$, then set $e_1 = (1,0,\ldots, 0)^\prime$ and for $\xi > 0$ define the matrices
\begin{equation*}
H_1 \equiv D + \xi e_1 e_1^\prime \hspace{0.5 in} H_2 \equiv D -\xi e_1 e_1^\prime ,   
\end{equation*}
noting that $\|D-H_1\| = \|D-H_2\| <\delta$ provided $\xi$ is chosen sufficiently small. 
In addition, $\Gamma \equiv D\Sigma$, $x \equiv \Sigma z$, result \eqref{aux:lin4}, and $x_1z_1 > 0$ together yield that
\begin{equation*}
z^\prime (H_1 \Sigma) z = z^\prime((D+\xi e_1e_1^\prime)\Sigma)z = z^\prime \Gamma z + \xi z_1 x_1 > 0.     
\end{equation*}
Identical arguments imply $z^\prime(H_2 \Sigma)z < 0$, and hence the lemma follows. \qed

\begin{lemma}\label{aux:lindep}
If Assumption \ref{ass:heterogssextra}(v) holds and $A_1, A_2$ satisfy $P(\Sh \in A_1 \cup A_2) = 1$, then it follows that ${\rm span}\{A_i\} = \R^p$ for some $i\in \{1,2\}$.    
\end{lemma}

\noindent {\emph Proof.} By way of contradiction, suppose that ${\rm span}\{A_i\} \neq \R^p$ for $i \in \{1,2\}$.
Then note that $A_i \subseteq V_i$ for some vector subspace $V_i \subset \R^p$ and that for each $i$ we may find a $v_i\in \R^p$ that is orthogonal to $V_i$ and satisfies $\|v_i\|^2 = 1$.
Next define $v^*$ to equal
\begin{equation*}
v^* \equiv \left\{\begin{array}{cl} v_1 & \text{ if } |\langle v_1,v_2\rangle| = 1\\
v_1 + v_2 & \text{ if } |\langle v_1,v_2\rangle| \neq 1
\end{array}\right.
\end{equation*} 
and note that $\langle v^*,v_i\rangle \neq 0$ for $i \in \{1,2\}$ due to $\|v_i\|^2 = 1$.
Since $v_i$ is orthogonal to $V_i$ it follows that $\xi v^* \notin V_i$ and hence $\xi v^*\notin V_1\cup V_2$ for any $\xi > 0$.
In particular, we obtain that $\xi v^*\notin A_1 \cup A_2$ due to $A_1\cup A_2 \subseteq V_1 \cup V_2$.
However, for $\xi$ sufficiently small, $\xi v^*$ belongs to the support of $\Sh$ by Assumption \ref{ass:heterogssextra}(v).
Thus, selecting $N(\xi v^*)$ to be a neighborhood of $\xi v^*$ sufficiently small to ensure that $N(\xi v^*) \subseteq (A_1\cup A_2)^c$ we obtain
\begin{equation}\label{aux:lindep2}
P(\Sh \in (A_1\cup A_2)^c) \geq P(\Sh \in N(\xi v^*)) > 0,    
\end{equation}
where the final inequality follows from $\xi v^*$ being in the support of $\Sh$. Since \eqref{aux:lindep2} contradicts the hypothesis that $P(\Sh \in A_1 \cup A_2) =1$, the lemma follows. \qed

\section{Proofs for Section \ref{sec:heteroakm}}


The results in Section \ref{sec:heteroakm} rely on the following regularity conditions.

\begin{assumption}\label{ass:heteroakmextra}
For $\mathcal G_n = \{S_i,W_i,\Lambda_i,\beta_i,\varepsilon_i,\eta_i\}_{i=1}^n$ and $\mathcal W_n = \{W_i\}_{i=1}^n$ we have:  (i) ${\rm Var}\{\Sh|\mathcal G_n\} = {\rm Var}\{\Sh|\mathcal W_n\}$ and ${\rm Var}\{\Sh_j|\mathcal W_n\} > 0$ for all $1\leq j \leq p$ with probability one; 
(ii) $\{\Lambda_i,S_i\}_{i=1}^n$ are i.i.d.\ conditionally on $\mathcal W_n$;
(iii) ${\rm supp}\{(\Lambda_i,S_i,\mathcal W_n)\} = {\rm supp}\{\Lambda_i\}\times {\rm supp}\{S_i\}\times {\rm supp}\{\mathcal W_n\}$;
(iv) ${\rm supp}\{S_i/\|S_i\|_1\} = \{s\in \R^p : s\geq 0 \text{ and } \|s\|_1 =1\}$.
\end{assumption}

Assumptions \ref{ass:heteroakmextra}(i)-(iii) impose conditions that simplify our analysis.
In turn,  Assumption \ref{ass:heteroakmextra}(iv) demands that the support of the shares be sufficiently rich -- in its statement, $\|\cdot\|_1$ denotes the $\ell_1$ norm $\|s\|_1 = \sum_{i=1}^n |s^{(i)}|$ for any vector $s = (s^{(1)},\ldots, s^{(p)})^\prime \in \R^p$.
We note that Assumption \ref{ass:heteroakmextra}(iv) allows the shares to sum to less than one, as is sometimes the case in empirical applications.

\noindent \emph{Proof of Proposition \ref{prop:heteroakm}.} 
First note that setting $\mathcal G_n = \{S_i,W_i,\Lambda_i,\beta_i,\varepsilon_i,\eta_i\}_{i=1}^n$ implies that  $(\Lambda_i,\beta_i,\varepsilon_i,\eta_i)\indep (\Sh,S_i)$ conditionally on $(W_i,\mathcal G_n)$. Hence, Assumption \ref{ass:heteroakm} yields 
\begin{multline}\label{prop:heteroakm1}
E\left[X_i \dot Z_{i,n}|\Lambda_i, W_i,\mathcal G_n\right] =  E\left[(\Sh^\prime \Lambda_iS_i)\dot Z_{i,n}\Big|\mathcal G_n\right]    \\
= S_i^\prime\Lambda_i E\left[ \Sh(\Sh^\prime - E[\Sh^\prime|\mathcal G_n])|\mathcal G_n \right]S_i = S_i^\prime (\Lambda_i \text{Var}\{\Sh|\mathcal G_n\}) S_i,
\end{multline}
where the first equality follows from Lemma \ref{lm:decomp} and $(\Lambda_i,W_i)\in\mathcal  G_n$, the second equality holds by $Z_i = S_i^\prime \Sh$ and Assumption \ref{ass:heteroakm}(ii), and the third equality holds by direct calculation.
The claim of the proposition therefore follows from result \eqref{prop:heteroakm1}, the definition of $\omega_{i,n}$ in \eqref{model:eq16}, and Lemma \ref{lm:signakm}. \qed

\noindent \emph{Proof of Corollary \ref{cor:heteroimpakm}.}
Suppose condition \eqref{cor:heteroimpakmdisp} holds for some $j\neq k$ and define
\begin{equation*}
\sigma_l^2(\mathcal G_n) \equiv \text{Var}\{\Sh_l|\mathcal G_n\} \hspace{0.3 in} \sigma_{jk}^2(\mathcal G_n) = \text{Cov}\{\Sh_j,\Sh_k|\mathcal G_n\} \hspace{0.3 in} \rho_{jk}(\mathcal G_n) \equiv \text{Corr}\{\Sh_j,\Sh_k|\mathcal G_n\}.
\end{equation*}
Also note that Assumptions \ref{ass:heteroakmextra}(i)(iii) and Proposition \ref{ass:heteroakm} together imply that a necessary condition for $\omega_{i,n}$ to be positive for all $i$ with probability one is that
\begin{equation*}
P(\Lambda_i \geq 0) =1 \text{ or } P(\Lambda_i \leq 0) = 1.
\end{equation*}
In what follows, we assume that $P(\Lambda_i\geq 0) =1$ and note that the case in which $P(\Lambda_i \leq 0) = 1$ follows by identical arguments.
Next note that by Proposition \ref{prop:heteroakm}, a necessary condition for $\omega_{i,n}$ to be positive for all $i$ with probability one is that 
\begin{equation}\label{cor:heteroimpakm3}
(\lambda_j s_j,\lambda_k s_k) \left(\begin{array}{cc} \sigma_j^2(\mathcal G_n) & \sigma_{jk}(\mathcal G_n) \\ \sigma_{jk}(\mathcal G_n) & \sigma_k^2(\mathcal G_n)\end{array}\right) \left(\begin{array}{c} s_j\\s_k\end{array}\right)
\end{equation}
be either positive for all $(s_j,s_k)\in \R_+$ with probability one over $\mathcal G_n$, or negative for all $(s_j,s_k)\in \R_+$ with probability one over $\mathcal G_n$.
However, Assumption \ref{ass:heteroakmextra}(i), $P(\lambda_{ij}=0) = 0$, and $P(\lambda_{ij}\geq 0)=1$ imply that setting $(s_j,s_k)=(1,0)$ yields a strictly positive number in \eqref{cor:heteroimpakm3} with probability one.
Hence, we can conclude that a necessary condition for $\omega_{i,n}$ to be positive for all $i$ with probability one is that we have
\begin{align}\label{cor:heteroimpakm4}
1 & = P\left(\inf_{s_j,s_k\geq 0} \left\{s_j^2 \lambda_j\sigma_j^2(\mathcal G_n) + s_k^2 \lambda_k \sigma_k^2(\mathcal G_n)+ s_j s_k \sigma_{jk}(\mathcal G_n)(\lambda_j+\lambda_k)\right\} \geq 0\right) \notag \\   
& = P\left(\inf_{s_j,s_k\geq 0} \left\{s_j^2 \lambda_j + s_k^2 \lambda_k + s_j s_k \rho_{jk}(\mathcal G_n)(\lambda_j+\lambda_k)\right\} \geq 0\right), \notag \\
& = P\left(\inf_{s_k\geq 0}\left\{ s_k^2\left(\lambda_k - \frac{\rho^2_{jk}(\mathcal G_n)}{4\lambda_j}(\lambda_j+\lambda_k)^2\right)\right\} \geq 0\right)
\end{align}
where in the second equality we employed Assumption \ref{ass:heteroakmextra}(i) and in the third equality we profiled out $s_j$ by setting $s_j = -s_k\rho(\mathcal G_n)(\lambda_j+\lambda_k)/2\lambda_j$.
However, note that
\begin{equation}\label{cor:heteroimpakm5}
P\left(\lambda_k - \frac{\rho^2_{jk}(\mathcal G_n)}{4\lambda_j}(\lambda_j + \lambda_k)^2 < 0\right) =  
P\left(1 < \frac{\rho^2_{jk}(\mathcal G_n)}{4}\left(\frac{\lambda_j}{\lambda_k} + \frac{\lambda_k}{\lambda_j}+2\right)\right) > 0
\end{equation}
where the inequality holds due to condition \eqref{cor:heteroimpakm}, Assumptions \ref{ass:heteroakmextra}(i)(iii), and the support of $\lambda_j/\lambda_k$ being unbounded.
Since result \eqref{cor:heteroimpakm5} implies the necessary condition in \eqref{cor:heteroimpakm4} cannot hold, the claim of the corollary follows. \qed

\begin{lemma}\label{lm:signakm}
Let Assumption \ref{ass:heteroakmextra} hold and $\mathcal G_n = \{S_i,W_i,\Lambda_i,\beta_i,\varepsilon_i,\eta_i\}_{i=1}^n$. 
Then,
\begin{equation}\label{lm:signakmdisp1}
P\left(\min_{1\leq i \leq n} S_i^\prime (\Lambda_i {\rm Var}\{\Sh|\mathcal G_n\})S_i \geq 0 \text{ or } \max_{1\leq i \leq n} S_i^\prime (\Lambda_i {\rm Var}\{\Sh|\mathcal G_n\})S_i \leq 0 \right) = 1
\end{equation}
for all $n$ sufficiently large if and only if the distribution of $\Lambda_i {\rm Var}\{\Sh|\mathcal G_n\}$ satisfies 
\begin{equation}\label{lm:signakmdisp2}
P\left(\sup_{s\in \R^p_+} s^\prime (\Lambda_i{\rm Var}\{\Sh|\mathcal G_n\}) s \leq 0\right) = 1  \text{ or } P\left(\inf_{s\in \R^p_+} s^\prime (\Lambda_i{\rm Var}\{\Sh|\mathcal G_n\}) s \geq 0\right) = 1.
\end{equation}
\end{lemma}

\noindent \emph{Proof.} First note that since $S_i \in \R^p_+$ by Assumption \ref{ass:heteroakmextra}(iv), it follows that condition \eqref{lm:signakmdisp2} implies \eqref{lm:signakmdisp1} holds.
For the reverse direction, set $\mathcal W_n = \{W_i\}_{i=1}^n$ and $\Gamma_{i,n} = \Lambda_i\text{Var}\{\Sh|\mathcal W_n\}$ for notational simplicity, and employ Assumptions \ref{ass:heteroakmextra}(i) to obtain
\begin{align}\label{lm:signakm1}
P& \left(\min_{1\leq i \leq n} S_i^\prime (\Lambda_i {\rm Var}\{\Sh|\mathcal G_n\})S_i \geq 0 \text{ or } \max_{1\leq i \leq n} S_i^\prime (\Lambda_i {\rm Var}\{\Sh|\mathcal G_n\})S_i \leq 0 \right) \notag \\
& = P\left(\min_{1\leq i \leq n} S_i^\prime\Gamma_{i,n} S_i \geq 0\right) + P\left(\max_{1\leq i \leq n}S_i^\prime \Gamma_{i,n}S_i \leq 0\right) - P\left(S_i^\prime \Gamma_{i,n}S_i = 0 ~ \text{ for all } i \right) \notag \\
& =E\left[P^n\left(S_i^\prime \Gamma_{i,n} S_i \geq 0|\mathcal W_n\right) + P^n\left(S_i^\prime \Gamma_{i,n}S_i \leq 0|\mathcal W_n\right) - P^n\left(S_i^\prime \Gamma_{i,n} S_i = 0 |\mathcal W_n\right)\right],
\end{align}
where the second equality follows from Assumption \ref{ass:heteroakmextra}(ii).
For $n\geq 2$, result \eqref{lm:signakm1} and the law of iterated expectations together imply that in order for the equality in \eqref{lm:signakmdisp1} to hold we must have with probability one over $(\Lambda_i,\mathcal W_n)$ that
\begin{equation}\label{lm:signakm2}
\max\left\{P\left(S_i^\prime \Gamma_{i,n} S_i \geq 0 | \Lambda_i,\mathcal W_n\right),P\left(S_i^\prime \Gamma_{i,n} S_i \leq 0 |\Lambda_i, \mathcal W_n\right)\right\} = 1.
\end{equation}
Next note that Assumption \ref{ass:heteroakmextra}(iii) implies that the distribution of $S_i$ is absolutely continuous with respect to the distribution of $S_i$ conditionally on $(\Lambda_i,\mathcal W_n)$.
Therefore, letting $\mathbb M_p$ denote the set of $p\times p$ real matrices and defining the sets
\begin{align*}
O_+ & \equiv \left\{ G \in \mathbb M_p: P\left(S^\prime G S \geq 0\right) =1 \right\}   \notag \\
O_{-} & \equiv \left\{ G\in \mathbb M_p: P\left(S^\prime G S \leq 0\right) = 1\right\}
\end{align*}
we obtain from result \eqref{lm:signakm2} holding with probability one over $(\Lambda_i,\mathcal W_n)$ that we have
\begin{equation}\label{lm:signakm4}
P\left(\Gamma_{i,n} \in O_+ \cup O_{-}\right) = 1.
\end{equation}
Moreover, since the sign of $a^\prime G a$ equals the sign of $a^\prime G a/\|a\|_1^2$ for any matrix $G\in \mathbb M^p$ and vector $0\neq a\in \R^p$, Assumption \ref{ass:heteroakmextra}(iv) allows us to conclude that
\begin{equation}\label{lm:signakm5}
O_+  = \left\{ G \in \mathbb M_p: \inf_{a\in \R^p_+} a^\prime G a \geq 0\right\}   \hspace{0.5 in}
O_{-}  = \left\{ G\in \mathbb M_p: \sup_{a\in \R^p_+} a^\prime G a \leq 0 \right\}.
\end{equation}
In particular, combining \eqref{lm:signakm4} with \eqref{lm:signakm5} yields that the distribution of $\Gamma_{i,n}$ satisfies
\begin{equation}\label{lm:signakm6}
P\left(\inf_{a\in \R^p_{+}} a^\prime \Gamma_{i,n} a \geq 0 \text{ or } \sup_{a\in \R^p_{+}} a^\prime \Gamma_{i,n} a \leq 0\right) = 1.
\end{equation}

To conclude, note that by Lemma \ref{lm:auxsign} we have $P(\Lambda_i \geq 0) =1$ or $P(\Lambda_i \leq 0) = 1$.
However, if $P(\Lambda_i \geq 0)=1$, then $\Gamma_{i,n} = \Lambda_i \text{Var}\{\Sh|\mathcal W_n\}$
and Assumption \ref{ass:heteroakmextra}(i) imply
\begin{multline}\label{lm:signakm7}
P\left(\sup_{a\in \R^p_+} a^\prime \Gamma_{i,n}a \leq 0\right) = P\left(\sup_{a\in \R^p_+} a^\prime \Gamma_{i,n}a \leq 0; \Lambda_i \geq 0\right) \\
\leq P\left(\Lambda_i = 0\right) \leq P\left(\sup_{a\in \R^p_+} |a^\prime \Gamma_{i,n}a| = 0\right)\leq P\left(\sup_{a\in \R^p_+} a^\prime \Gamma_{i,n}a \leq 0\right),
\end{multline}
where the final two inequalities hold by set inclusion.
By result \eqref{lm:signakm6} we thus obtain
\begin{align*}
1 & =  P\left(\inf_{a\in \R^p_{+}} a^\prime \Gamma_{i,n} a \geq 0 \text{ or } \sup_{a\in \R^p_{+}} a^\prime \Gamma_{i,n} a \leq 0\right) \notag \\
& = P\left(\inf_{a\in \R^p_{+}} a^\prime \Gamma_{i,n} a \geq 0\right) + P\left(\sup_{a\in \R^p_{+}} a^\prime \Gamma_{i,n} a \leq 0\right) - P\left(\sup_{a\in \R_+^p} |a^\prime \Gamma_{i,n}a| = 0\right)  \notag \\
& = P\left(\inf_{a\in \R^p_{+}} a^\prime \Gamma_{i,n} a \geq 0\right),
\end{align*}
where the final equality follows from \eqref{lm:signakm7}.
Similarly, it is possible to show that if $P(\Lambda_i \leq 0) =1$, then $P(\sup_{a\in \R^p_+}a^\prime \Gamma_{i,n}a\leq 0) = 1$, and therefore the lemma follows. \qed

\begin{lemma}\label{lm:auxsign}
Let Assumption \ref{ass:heteroakmextra} hold, set $\mathcal G_n = \{S_i,W_i,\Lambda_i,\beta_i,\varepsilon_i,\eta_i\}_{i=1}^n$, and 
\begin{equation}\label{lm:auxsigndisp}
P\left(\min_{1\leq i \leq n} S_i^\prime (\Lambda_i {\rm Var}\{\Sh|\mathcal G_n\}S_i\} \geq 0 \text{ or } \max_{1\leq i \leq n} S_i^\prime (\Lambda_i {\rm Var}\{\Sh|\mathcal G_n\})S_i\} \leq 0 \right) = 1
\end{equation}
for $n$ sufficiently large.
Then, it follows that either $P(\Lambda_i \geq 0) =1$ or $P(\Lambda_i \leq 0) = 1$.
\end{lemma}

\noindent \emph{Proof}. We will proceed by contradiction and instead suppose that in fact we have that
\begin{equation}\label{lm:auxsign1}
\min\left\{P\left(\max_{1\leq j \leq p} \lambda_{ij} > 0\right),P\left(\min_{1\leq j \leq p} \lambda_{ij} < 0\right)\right\} > 0.
\end{equation}
For notational simplicity, let $\mathcal W_n = \{W_i\}_{i=1}^n$, set $\Gamma_{i,n}\equiv \Lambda_i \text{Var}\{\Sh|\mathcal W_n\}$, and define 
\begin{align*}
A_+(\Gamma_{i,n}) & \equiv \left\{s \in \R^p : s\geq 0 \text{ and } s^\prime \Gamma_{i,n} s > 0\right\}  \notag \\
A_-(\Gamma_{i,n}) &\equiv \left\{s \in \R^p : s\geq 0 \text{ and } s^\prime \Gamma_{i,n} s < 0\right\}.
\end{align*}
Next note that if $\lambda_{ij} > 0$ for some $1\leq j \leq p$ (resp.\ $\lambda_{ij} < 0$ for some $1\leq j \leq p$) then $A_+(\Gamma_{i,n})$ (rep.\ $A_{-}(\Gamma_{i,n})$) has non-empty interior whenever $\text{Var}\{\Sh_j|\mathcal W_n\} > 0$ for all $1\leq j \leq p$.
Therefore, since $s \in A_+(\Gamma_{i,n})$ (resp.\ $A_{-}(\Gamma_{i,n})$) if and only if $s/\|s\|_1\in A_+(\Gamma_{i,n})$ (resp.\ $s/\|s\|_1\in A_{-}(\Gamma_{i,n})$), we obtain by Assumptions \ref{ass:heteroakmextra}(i)(iii) that
\begin{align}\label{lm:auxsign3}
P\left(S_i \in A_+(\Gamma_{i,n})\Big| \max_{1\leq j \leq p} \lambda_{ij} > 0, \mathcal W_n \right) & > 0 \notag \\
P\left(S_i \in A_{-}(\Gamma_{i,n}) \Big| \min_{1\leq j \leq p} \lambda_{ij} < 0, \mathcal W_n\right) &> 0.    
\end{align}
Similarly, also note that display \eqref{lm:auxsign1} holding and Assumption \ref{ass:heteroakmextra}(iii) together imply
\begin{equation}\label{lm:auxsign4}
\min\left\{P\left(\max_{1\leq j \leq p} \lambda_{ij} > 0 \Big|\mathcal W_n\right), P\left(\min_{1\leq j \leq p} \lambda_{ij} < 0 \Big|\mathcal W_n\right)\right\} > 0.    
\end{equation}
To conclude, observe that for $n \geq 2$, display \eqref{lm:auxsign} and Assumptions \ref{ass:heteroakmextra}(i)(ii) yield
\begin{align*}
0 & \geq P\left(S_1^\prime \Lambda_1 \text{Var}\{\Sh|\mathcal W_n\} S_1 > 0 \text{ and } S_2^\prime \Lambda_2 \text{Var}\{\Sh|\mathcal W_n\} S_2 < 0\right) \notag \\
& = E\left[P\left(S_1^\prime \Lambda_1 \text{Var}\{\Sh|\mathcal W_n\} S_1 > 0|\mathcal W_n\right)P\left(S_1^\prime \Lambda_1 \text{Var}\{\Sh|\mathcal W_n\} S_1 < 0|\mathcal W_n\right)\right] \notag \\
& \geq E\left[P\left(S_1\in A_+(\Gamma_{1,n}); \max_{1\leq j \leq p}\lambda_{1j} > 0 |\mathcal W_n\right)P\left(S_1\in A_-(\Gamma_{1,n}); \min_{1\leq j \leq p} \lambda_{1j} < 0|\mathcal W_n\right)\right]\\
& >0,
\end{align*}
where the final inequality follows from results \eqref{lm:auxsign3} and \eqref{lm:auxsign4}. 
Hence, we have arrived at a contradiction implying \eqref{lm:auxsign1} cannot hold. \qed

\section{Proofs for Section \ref{sec:longpanel}}\label{App:sec:longpanel}


In this appendix we collect the technical results behind Proposition \ref{prop-clt}.
Section \ref{sec_Joint_Conv_Def} introduces the framework, Section \ref{sec-intuition} provides an outline of the main argument, and Sections \ref{Sec_uncond_conv} and \ref{sec_Cond_Analysis} establish the two key building blocks for the main result.

\subsection{Preliminaries and Definitions\label{sec_Joint_Conv_Def}}

We consider a probability space $\left(  \Xi\times\Psi,\mathcal{F}\times\mathcal{C},P_{u}\times P_{f}\right)  $ where $\Xi$ and $\Psi$ are Polish spaces with their respective Borel $\sigma$-algebras $\mathcal{F}$ and
$\mathcal{C}$; see, p.\ 270 in \cite{dudley1989real}. 
An example of such spaces is $\Xi=\left(  \mathbb{R}^{\infty}\times....\times\mathbb{R}^{\infty}\right)$;
i.e.\ a finite number of products of $\mathbb{R}^{\infty}$. 
Let $\Omega\equiv\Xi\times\Psi,$ $\mathcal{X}\equiv \mathcal{F}\times\mathcal{C}$ and $P\equiv P_{u}\times P_{f}$ so that the probability space can be represented more compactly as $\left(\Omega,\mathcal{X},P\right)$. 
It is useful for later developments to impose further structure on $\left(  \Xi,\mathcal{F},P_{u}\right)  $. 
We assume that $\left\{  \left(  \Xi_{i},\mathcal{F}_{i},P_{u,i}\right)  \right\}_{i=1}^{\infty}$ is a sequence of probability spaces and define $\Xi \equiv \Xi_{1} \times \Xi_{2} \times....$, $\mathcal{F} \equiv \mathcal{F}_{1} \times \mathcal{F}%
_{2} \times...,$ and $P_{u} \equiv P_{u,1}\times P_{u,2}\times...$ such that $\left(  \Xi,\mathcal{F},P_{u}\right)  $ is an infinite dimensional product space. 
Let $\mathcal{F}_{t}^{i}\subset\mathcal{F}_i$ denote an array of filtrations, where ${\mathcal F}_{-\infty}^{i}   =\left\{  \mathcal{\emptyset},\Xi_{i}\right\}$,  ${\mathcal F}_{\infty}^{i}={\mathcal F}_{i}$, and $\mathcal{F}_{t}^{i}   \subset \mathcal{F}_{t+1}^{i}$ for $i\leq n$.\footnote{We can allow for triangular arrays, but suppress it for simplicity of notations.}  Similarly, let $\mathcal{C}_{t}\subset\mathcal{C}$ denote a triangular array of filtrations, where $\mathcal{C}_{-\infty}  =\left\{  \mathcal{\emptyset},\Psi\right\}$,  $\mathcal{C}_{\infty} = \mathcal{C}$, and $\mathcal{C}_{t}  \subset \mathcal{C}_{t+1}  $.
In addition, define the filtrations $\mathcal{H}_{t}^{i}\equiv\mathcal{F}_{t}^{i}\times\mathcal{C}_{t}$. 
Then, $\mathcal{F}\times\mathcal{C}_{t}\subset\mathcal{X}$ for all $t$ and $\mathcal{H}_{t}^{i}$ can be embedded in $\mathcal{X}$ for each $i$ and all $t\leq T$ by padding up additional coordinates; see, e.g., page 140 in \cite{halmos1976measure}. 
In addition, note that $\mathcal{H}%
_{t}^{i}\subset\mathcal{H}_{t+1}^{i}$ for all $i$, and $t$. 

In order to simplify our analysis, we will assume that $\mathcal{F}_{t}^{i}   =\sigma\left(  ....,\eta_{i,t-1}%
,\eta_{it}\right)$, where $\left(  \eta_{it}\right)
_{i=1,t=-\infty}^{\infty,\infty}$ is an array of \emph{some} random variables on
$\left(  \Xi,\mathcal{F},P_{u}\right)  $ where for each $i$, $\left(
\eta_{it}\right)  _{t=-\infty}^{\infty}$ is an array of
random variables defined on $\left(  \Xi_{i},\mathcal{F}_{i},P_{u,i}\right)
$. The product space structure of $\left(  \Xi,\mathcal{F},P_{u}\right)  $
immediately implies independence of $\left(  \eta_{it}\right)
_{t=-\infty}^{\infty}$ and $\left(  \eta_{jt}\right)  _{t=-\infty
}^{\infty}$ for any $i\neq j$. Likewise, we assume that $\mathcal{C}_{t}    =\sigma\left(  ...,v_{t-1},v_{t}\right)$, where $\left(  v_{t}\right)  _{t=-\infty}^{\infty}$ is
a sequence of some random variables on $\left(  \Psi,\mathcal{C}%
,P_{f}\right)  $. 

We will introduce mixing measures as in
\cite{mcleish1975}; see also \cite{andrews1988} for triangular array versions
of these measures. Recall $\mathcal{C}_{t}=\sigma\left(  ...,v_{t-1},v_t\right)  $ and define $\mathcal{C}_{t}^{t+m}\equiv
\sigma\left(  v_{t},v_{t+1},...,v_{t+m}\right)  $ and $\mathcal{C}%
_{t}^{\infty}\equiv\sigma\left(  v_{t},v_{t+1},...\right)  $. Similarly, we let
$$\mathcal{F}_{t}^{i,t+m} \equiv\sigma\left(  \eta_{it}
,...,\eta_{i,t+m}\right)  ~\text{for all }i\leq n. $$
Our asymptotic normality results employ $\alpha$-mixing coefficients, and we therefore define
\[
\alpha_{f}\left(  m\right)     \equiv \sup_{t}\sup_{A\in\mathcal{C}_{t}%
,B\in\mathcal{C}_{t+m}^{\infty}}\left\vert P\left(  A\cap B\right)
-P\left(  A\right)  P\left(  B\right)  \right\vert . \]

We will now impose the following    
restrictions on the measures $P_{u}$ and $P_{f}$. By Theorems 4.34 and A.46 in \cite{breiman1968probability}, a regular conditional distribution on
$\mathcal{X}$ given $\mathcal{C}\subset\mathcal{X}$ exists, and by Theorem 10.2.2 in 
\cite{dudley1989real}, the regular conditional distribution is
unique for $P$-almost all $\omega\in\Omega$. As in \cite{eaglson1975}, let
$\omega^{\prime}\in\Omega$ and consider the regular conditional (on
$\mathcal{C}$) probability denoted by $Q_{\omega^{\prime}}\left(
B,\mathcal{C}\right)  =Q_{\omega^{\prime}}\left(  B\right)  $. It follows that
for fixed $B\in\mathcal{X}$, $Q_{\omega^{\prime}}\left(  B,\mathcal{C}\right)
$ is a version of $P\left(  \left.  B\right\vert \mathcal{C}\right)  $ and for
fixed $\omega^{\prime}\in\Omega$, $Q_{\omega^{\prime}}\left(  \cdot\right)  $
is a probability measure on $\mathcal{X}$. Importantly, this means
$Q_{\omega^{\prime}}\left(  \cdot\right)  $ is countably additive which
ensures that the law of iterated expectations holds; see p.\ 270 in \cite{dudley1989real}. We note that the results in \cite{dudley1989real} are
established for Polish spaces.

Consider the measure space $\left(  \Omega,\mathcal{X},Q_{\omega^{\prime}%
}\right)  $ with expectation $E_{\omega^{\prime}}$, which formalizes the idea of treating the aggregate variables $\varkappa_{t}$ (to be defined later) as constants through the choice of $\omega^{\prime}$. By a lemma in p.\ 558 of \cite{eaglson1975}, the following holds:

\begin{lemma}
[Eagleson 1975]\label{lem_Eagleson} Let $\mathcal{Y}$ be a sub-sigma field of
$\mathcal{X}$ such that $\mathcal{C}\subseteq\mathcal{Y}$.\footnote{It is easy
to see that the proof of Lemma 1 in \cite{eaglson1975} goes through when
$\mathcal{F=G}$.} Then, for $P$ almost all $\omega^{\prime}\in\Omega$ and a
random variable $q$ with $E\left[  \left\vert q\right\vert \right]  <\infty$
it follows that $E_{\omega^{\prime}}\left[  \left.  q\right\vert
\mathcal{Y}\right]  \left(  \omega\right)  =E\left[  \left.  q\right\vert
\mathcal{Y}\right]  \left(  \omega\right)  $ $Q_{\omega^{\prime}}$-a.s. 
\end{lemma}

For arbitrary $t$ and arbitrary $m\geq0$, we next define the filtration $\mathcal Y_t^{i,t+m}$ by
\begin{equation}\label{Y-Filt}
\mathcal{Y}_{t}^{i,t+m}\mathcal{\equiv F}_{t}^{i,t+m}\times\mathcal{C}
\end{equation}
 To establish marginal convergence as $T\rightarrow\infty$ for each $i\leq n$
fixed we follow the strategy of the proof of Theorem 2 in \cite{eaglson1975}. 
This
requires modifying the regularity conditions to the measure $Q_{\omega
^{\prime}}$. Define the conditional $L_{q}$ norm $\left\Vert y\right\Vert
_{\left.  q\right\vert \mathcal{C}}=\left(  \int\left\vert y\right\vert
^{q}dP\left(  \left.  y\right\vert \mathcal{C}\right)  \right)  ^{1/q}$. By Lemma
\ref{lem_Eagleson} it follows that $\left\Vert y\right\Vert _{\left.
q\right\vert \mathcal{C}}=\left(  \int\left\vert y\right\vert ^{q}%
dQ_{\omega^{\prime}}\right)  ^{1/q}$ $Q_{\omega^{\prime}}$ a.s. for $P$-almost
all $\omega^{\prime}$. Similarly, define the conditional mixing coefficients
\begin{equation}
\alpha_{\left.  \xi\right\vert \mathcal{C}}\left(  m\right)  \equiv \sup_{t}%
\sup_{A\in\mathcal{Y}_{-\infty}^{i,t},B\in\mathcal{Y}_{t+m}^{i,\infty}%
}\left\vert P\left(  \left.  A\cap B\right\vert \mathcal{C}\right)  -P\left(
\left.  A\right\vert \mathcal{C}\right)  P\left(  \left.  B\right\vert
\mathcal{C}\right)  \right\vert ,\quad a.s. \label{alpha_Y}%
\end{equation}
where $\alpha_{\left.  \xi\right\vert
\mathcal{C}}\left(  m\right)  $ are $\mathcal{C}$-measurable random variables.

Now, recall that our analysis is predicated on the moment restriction in \eqref{eq:long2}. In order to facilitate our asymptotic analysis, we impose the following:

\begin{condition}
\label{Cond_Ind}
(i) There are some random vectors $f_{t}$ and $u_{it}$ defined on $(\Psi,{\mathcal C},P_f)$ and $(\Xi_i, {\mathcal F}_{i}, P_{u,i})$ such that 
$S_{i}%
\varepsilon_{it}=\psi\left(  f_{t},u_{it}\right)  $; (ii) The ${\mathcal Z}_t$ is a component of $f_t$; and  (iii) The joint distribution of $(u_{it})_{t=1}^{\infty}$ is identical across $i$.
\end{condition}

It is important to note that the product structure $P_{u}=P_{u,1}\times P_{u,2}\times...$ implies that the variables $(u_{it})_{t=1}^{\infty}$ are independent over $i$. 
Finally, we define the variables
\[
\xi_{it}\equiv(\mathcal{Z}_{t})^\prime \nu_{it}, \hspace{0.4 in} 
\varkappa_{t}\equiv(\mathcal{Z}_{t})^\prime \zeta_{t},
\]
where $\zeta_t$ and $\nu_{it}$ are given by $\zeta_{t} \equiv E\left[\left.  S_{it}\varepsilon_{it}\right| {\mathcal C}\right] = E\left[\left.  S_{it}\varepsilon_{it}\right| f_t \right]$ and $\nu_{it} \equiv S_{it}\varepsilon_{it} - \zeta_{t}$.

\subsection{Proof of Proposition \ref{prop:long-panel}  (Outline)} \label{sec-intuition}

Recall that Proposition \ref{prop:long-panel} derives an asymptotic distribution for the numerator (i.e.\ the score) of $(\hat \beta_n - \beta)$; see \eqref{eq:long3}. 
Given the introduced notation, we can decompose
\begin{equation*}
\frac{1}{nT}\sum_{i=1}^{n}\sum_{t=1}^{T}\left(
\mathcal{Z}_{t}^{\prime}S_{it}\right)  \varepsilon_{it}  = 
\frac{1}{T}\sum_{t=1}%
^{T} \varkappa_{t}+\frac{1}{nT}\sum_{t=1}^{T}\sum_{i=1}^{n}\xi_{it}.
\end{equation*}

We start the analysis by studying the conditional distribution given $\mathcal C$ of the term
\begin{equation}\label{eq:int1}
\frac{1}{\sqrt {nT}} \sum_{t=1}^T \sum_{i=1}^{n} \xi_{it}.
\end{equation}
By Condition \ref{Cond_Ind}, $\xi_{it}$ are independent over $i$ conditional on $\mathcal{C}$, and we therefore expect that \eqref{eq:int1} is asymptotically normal conditional on $\mathcal{C}$. 
Under additional regularity conditions, the limit distribution has a
variance that does not depend on $\mathcal{C}$. Indeed, this is formally established in Theorem \ref{Joint_Cond_CLT} below.
In  particular, defining
\begin{equation}\label{CCF}
\phi_{nT}\left(  \left.  \varsigma_{1}\right\vert \mathcal{C}\right)  \equiv
E\left[  \left.  \exp\left(  \iota\varsigma_{1}\frac{1}{\sqrt{nT}}\sum
_{t=1}^{T}\sum_{i=1}^{n}\xi_{it}\right)  \right\vert
\mathcal{C}\right]  ,
\end{equation}
we have that $\phi_{nT}\left(  \left.  \varsigma_{1}\right\vert \mathcal{C}\right)
\rightarrow\phi\left(  \varsigma_{1}\right)  $ almost surely in $\mathcal C$, where $\phi\left(  \varsigma_{1}\right)  $ denotes the characteristic function
of the limiting normal distribution.

Now, let's consider the joint characteristic function of the vector in display \eqref{vector-of-interest}:
$$
   E\left[  \exp\left(  \iota\varsigma_{2}\frac{1}{\sqrt{T}}\sum_{t=1}%
^{T}\varkappa_{t}+\iota\varsigma_{1}\frac{1}{\sqrt{nT}%
}\sum_{t=1}^{T}\sum_{i=1}^{n}\xi_{it}\right)
\right] 
 =E\left[  \exp\left(  \iota\varsigma_{2}\frac{1}{\sqrt{T}}\sum
_{t=1}^{T}\varkappa_{t}\right)  \phi_{nT}\left(  \left.
\varsigma_{1}\right\vert \mathcal{C}\right)  \right] 
$$
where we note that $\varkappa_t$ is measurable with respect to $\mathcal{C}$.  
Because $\sum_{t=1}^T E[\varkappa_{t}]=0$ due to the moment restriction in \eqref{eq:long2}, a time series CLT should apply to the term
\begin{equation}\label{eq:int2}
\frac{1}{\sqrt T} \sum_{t=1}^T \varkappa_t.
\end{equation}
Indeed, the asymptotic normality of \eqref{eq:int2} is formally established in Theorem \ref{Thm_JointCLT-TSterm} in Appendix \ref{Sec_uncond_conv} below. 
Letting $\varphi\left(\varsigma_{2}\right)  $ denote the characteristic function of the limiting normal
distribution of the term in \eqref{eq:int2}, we then obtain that
\begin{align}
  E & \left[  \exp\left(  \iota\varsigma_{2}\frac{1}{\sqrt{T}}\sum_{t=1}%
^{T}\varkappa_{t}+\iota\varsigma_{1}\frac{1}{\sqrt{nT}%
}\sum_{t=1}^{T}\sum_{i=1}^{n}\xi_{it}\right)
\right]  -\varphi\left(  \varsigma_{2}\right)  \phi\left(  \varsigma
_{1}\right)  \notag \\
  & =  E\left[  \exp\left(  \iota\varsigma_{2}\frac{1}{\sqrt{T}}\sum_{t=1}%
^{T}\varkappa_{t}\right)  \phi_{nT}\left(  \left.
\varsigma_{1}\right\vert \mathcal{C}\right)  \right]  -E\left[  \exp\left(
\iota\varsigma_{2}\frac{1}{\sqrt{T}}\sum_{t=1}^{T}\varkappa_{t}\right)  \phi\left(  \varsigma_{1}\right)  \right] \label{eq:int3} \\
&  +E\left[  \exp\left(  \iota\varsigma_{2}\frac{1}{\sqrt{T}}\sum_{t=1}%
^{T}\varkappa_{t}\right)  \phi\left(  \varsigma
_{1}\right)  \right]  -\varphi\left(  \varsigma_{2}\right)  \phi\left(
\varsigma_{1}\right) . \label{eq:int4}
\end{align}
Because characteristic functions are bounded by one and $\phi_{nT}\left(  \left.  \varsigma_{1}\right\vert \mathcal{C}%
\right)  \rightarrow\phi\left(  \varsigma_{1}\right)  $ almost surely in $\mathcal C$ we obtain from the dominated convergence theorem that
\begin{multline*}
\left\vert \eqref{eq:int3}  \right\vert \leq E\left[  \exp\left(
\iota\varsigma_{2}\frac{1}{\sqrt{T}}\sum_{t=1}^{T}\varkappa_{t}\right)  \left\vert \phi_{nT}\left(  \left.  \varsigma_{1}\right\vert
\mathcal{C}\right)  -\phi\left(  \varsigma_{1}\right)  \right\vert \right]\\
\leq E\left[  \left\vert \phi_{nT}\left(  \left.  \varsigma_{1}\right\vert
\mathcal{C}\right)  -\phi\left(  \varsigma_{1}\right)  \right\vert \right]  \rightarrow 0.
\end{multline*}
Similarly, since \eqref{eq:int2} is asymptotically normally distributed and $\varphi\left(\varsigma_{2}\right)$ denotes the characteristic function of its limiting distribution, we also have that
\begin{multline*}
\left\vert \eqref{eq:int4}  \right\vert  
  \leq\left\vert E\left[  \exp\left(  \iota\varsigma_{2}\frac{1}{\sqrt{T}%
}\sum_{t=1}^{T}\varkappa_{t}\right)  \right]
-\varphi\left(  \varsigma_{2}\right)  \right\vert \left\vert \phi\left(
\varsigma_{1}\right)  \right\vert \\
  \leq\left\vert E\left[  \exp\left(  \iota\varsigma_{2}\frac{1}{\sqrt{T}%
}\sum_{t=1}^{T}\varkappa_{t}\right)  \right]
-\varphi\left(  \varsigma_{2}\right)  \right\vert 
  \rightarrow 0.
\end{multline*}
To conclude, we see that \eqref{eq:int1} and \eqref{eq:int2} are jointly asymptotically normally distributed and independent under the conditions laid out in Theorems \ref{Thm_JointCLT-TSterm} and \ref{Joint_Cond_CLT}.

\begin{remark} \rm
While we have focused on the numerator of $(\hat \beta_n -\beta)$, the denominator can be studied under similar conditions. 
In particular it is possible to show that
$$\frac{1}{nT}\sum_{t=1}^T \sum_{i=1}^n (\Sh_t^\prime S_{it})X_{it} = \frac{1}{nT}\sum_{t=1}^T \sum_{i=1}^n E[(\Sh_t^\prime S_{it})X_{it}] + o_P(1),$$
though, for conciseness, we do not provide the details of the relevant argument. \qed
\end{remark}

\subsection{Unconditional Convergence}\label{Sec_uncond_conv}

The main purpose of this section is to establish a time series central limit
theorem for 
$$\frac{1}{\sqrt T}\sum_{t=1}^{T}\varkappa_{t},$$ which is formally established in Theorem \ref{Thm_JointCLT-TSterm} below. For this purpose, it is convenient to assume
that $\varkappa_{t}$ is $L_{2}$ near-epoch-dependent (NED).
The concept of NED sequences was introduced by \cite{billingsley1968convergence}.
\cite{mcleish1975} or \cite{andrews1988} show that a NED process is also a
mixingale. Therefore assumptions that impose NED type conditions lead to
strong laws by showing that these processes also satisfy the requirements for
strong laws of related mixingales; see Theorem 3 in \cite{mcleish1975} for the first result of this nature.

We impose the following condition that establishes the NED property for
$\varkappa_{t}$, and imposes sufficient conditions for $\varkappa_{t}$
to satisfy the conditions for the SLLN in \cite{DEJONG1996} and the
CLT in \cite{dejong1997}.
In the statement below, we say that a sequence $\delta_{m}$ is of size $-\lambda$ if $\delta_{m}=O\left(  m^{-\lambda-\omega}\right)  $ for some $\omega>0$; see p.\ 335 in \cite{dejong1997}.

\begin{condition}
\label{Cond_NED}For $r>2$, assume that $\sup_{t}\left\Vert
\varkappa_{t}\right\Vert _{r}<\infty$, and that there exists a bounded array of non-random constants $c_{t}$ such that
\[
\left\Vert \varkappa_{t}-E\left[  \varkappa_{t}|\mathcal{C}_{t-m}^{t+m}\right]
\right\Vert _{2}\leq c_{t}\beta_{f}\left(  m\right)  \text{ for all }t 
\]
where $\beta_{f}\left(  m\right)  $ is of
size $-1/2$. Further assume that for $r>2$,
$\alpha_{f}\left(  m\right)  $ is of
size $r/\left(  r-2\right)  $. 
\end{condition}

To establish the CLT we adopt the conditions given in \cite{dejong1997}. The
result is based on a blocking scheme that needs to be defined. Let $b_{T}$ and
$l_{T}$ be positive, non-decreasing integer valued sequences that are the
lengths of included and discarded blocks. Assume $b_{T}\geq l_{T}+1,$
$l_{T}\rightarrow\infty$, $l_{T}\geq1,$ $b_{T}\leq T$, $b_{T}/T\rightarrow0$
and $l_{T}/b_{T}\rightarrow0$. Let $r_{T}\equiv\left[  T/b_{T}\right]  $ and
define $\tilde{\varkappa}_{T,t}\equiv \varkappa_{t}/\sigma_{T,\varkappa}$, where
where $\sigma_{T,\varkappa}^{2}$ is given by
$$\sigma_{T,\varkappa}^{2}\equiv E\left[  \left(  \sum_{t=1}^{T}\varkappa_{t}\right)
^{2}\right].  $$ 

The following condition is sufficient for obtaining the desired CLT.

\begin{condition}
\label{Cond_NED_CLT_Z}Assume that $\left\{  \varkappa_{t},\mathcal{C}_{t}\right\}
$ satisfies Condition \ref{Cond_NED}. There exists a positive constant array
$e_{T,t}$ such that $c_{t}/e_{T,t}$ is uniformly bounded in $t\leq T$ and
$T\geq1$ and $\left\{  \tilde{\varkappa}_{T,t}/e_{T,t}\right\}  $  is $L_{r}$ bounded
for $r>2$ uniformly in $t\leq T$ and $T\geq1$. Let $M_{T,j}\equiv\max_{\left(
j-1\right)  b_{T}+1\leq t\leq jb_{T}}e_{T,t}$ for $1\leq j\leq r_{T}$ and
$M_{T,r_{T}+1}\equiv\max_{r_{T}b_{T}+1\leq t\leq T}e_{T,t}$. Then,
$\max_{1\leq j\leq r_{T}+1}M_{T,j}=o\left(  b_{T}^{-1/2}\right)  $ and
$\sum_{j=1}^{r_{T}}M_{T,j}^{2}=O\left(  b_{T}^{-1}\right)  .$
\end{condition}

\begin{theorem}\label{Thm_JointCLT-TSterm}
Assume that Condition \ref{Cond_NED_CLT_Z} holds and suppose that there is some nonrandom constant $\sigma_{\varkappa}^{2}$ such that $\lim_{T \to \infty} \sigma_{T,\varkappa}^2/T = \sigma_{\varkappa}^{2}$. 
It then follows that
$$\frac{1}{\sqrt T} \sum_{t=1}^{T}\varkappa_{t} \stackrel{d}{\rightarrow}N\left(  0,\sigma_{\varkappa}^{2}\right).  $$
\end{theorem}

\begin{proof}
Under Condition \ref{Cond_NED_CLT_Z}, we can use Theorem 2 in \cite{dejong1997} to conclude that $\sum_{t=1}^{T}\tilde{\varkappa}_{T,t}\stackrel{d}{\rightarrow} N\left(  0,1\right)$. Under the additional assumption that $T^{-1}\sigma_{T,\varkappa}^{2}\rightarrow
\sigma_{\varkappa}^{2}$ for
some non-random constant $\sigma_{\varkappa}^{2}$,  it follows
from the continuous mapping theorem that
\begin{equation}
\frac{1}{\sqrt T}\sum_{t=1}^{T}\varkappa_{t}\stackrel{d}{\rightarrow} N\left(  0,\sigma_{\varkappa}^{2}\right), \label{CLT_Z_uncond}
\end{equation}
which establishes Theorem \ref{Thm_JointCLT-TSterm}. 
\end{proof}

\subsection{Conditional Analysis}\label{sec_Cond_Analysis}

The primary purpose of this section is to establish the asymptotic normality of the term
$$\frac{1}{\sqrt{nT}}\sum_{t=1}^{T}\sum_{i=1}^{n}\xi_{it}$$ conditional on $\mathcal{C}$, which is formally shown in Theorem \ref{Joint_Cond_CLT}.
The proof  uses ideas similar to the development in \cite{HKM},
which in turn relies on arguments in \cite{eaglson1975}, to handle the
conditioning step in the proof of the CLT; see also Lemma 2.9.5 in \cite{vandervaart:wellner:1996}. 
However, the dependence structure of the panel
is more complicated here than in \cite{HKM} and requires a different approach
to prove the CLT. \cite{HKM} consider a scenario where conditional on
$\mathcal{C}$, a cross-sectional average over i.i.d.\ draws is analyzed. 
Here, we need to extend these results to a panel setting with joint asymptotics as
$N,T\rightarrow\infty$ and where we allow for general dependence and possible
non-stationarity in the time series direction. We extend the notation from
\cite{HKM} to account for these extensions. The result of \cite{eaglson1975}
formalizes the intuition that conditional on $\mathcal{C}$ the processes
$\mathcal{Z}_{t}$ and $f_{t}$ can be treated as a fixed constant in deriving
limiting results.

We assume that $\alpha_{\left.  \xi\right\vert \mathcal{C}}\left(  m\right)
\rightarrow0$ almost surely as $m\rightarrow\infty.$ 
This is done in the next condition, which introduces conditional NED, by requiring $\alpha_{\xi|\mathcal C}$ be of size $r/(r-2)$. 

\begin{condition}
\label{Conditional_NED}
For $r> 2$, assume that $\sup \left\Vert \xi_{it}\right\Vert _{\left.  r\right\vert \mathcal{C}}%
<\infty$ a.s., and that there exists an array of $\mathcal{C}$-measurable
random variables $d_{it}\left(  \omega\right)  $ and a sequence of
$\mathcal{C}$-measurable random variables $\beta_{\left.  \xi\right\vert
\mathcal{C}}\left(  m\right)  $ of size $-1/2$ a.s.\ such that the following inequality holds:
\[
\left\Vert \xi_{it}-E\left[  \left.  \xi_{it}\right\vert \mathcal{Y}%
_{t-m}^{i,t+m}\right]  \right\Vert _{\left.  2\right\vert \mathcal{C}}  
\leq d_{it}\left(  \omega\right)  \beta_{\left.  \xi\right\vert
\mathcal{C}}\left(  m\right). 
\]
Further assume that for $r>2,$ $\alpha_{\left.
\xi\right\vert \mathcal{C}}\left(  m\right)  $ is of size $r/\left(
r-2\right)  $ a.s.
\end{condition}

We now turn to establishing our main result. The argument is based  on the one-to-one mapping $h_T$ of the double index $\left(  i,t\right) $ into a single index $s=h_T(i,t) \equiv \left(  i-1\right)  T+t$ for $i\leq n$ and $L\equiv nT$. We will let $(\bar{\iota}_{T}(s), \bar{t}_{T}(s))$ denote the $(i,t)$ that corresponds to $s$, i.e., $(\bar{\iota}_{T}(s), \bar{t}_{T}(s)) \equiv h_{T}^{-1}(s)$.\footnote{Note $\bar{\iota}_{T}(s)$ is the smallest integer larger than or equal to $s/T$, and $\bar{t}_{T}(s) = s - (\bar{\iota}_{T}(s)-1)\cdot T$.}
With some abuse of notation, we then set
\[
\left(  nT\right)  ^{-1/2}\sum_{i=1}^{n}\sum_{t=1}^{T}\xi_{it}=L^{-1/2}%
\sum_{s=1}^{L}\xi_{L,s}.
\]
The sum on the right hand side sums over components starting with $i=1$, $t=1,\ldots,T$
followed by $i=2,$ $t=1,\ldots,T$, and so on. 
Similarly, we construct the array of
filtrations based on \eqref{Y-Filt}, using coordinate identification rules of
\cite{halmos1976measure} (see p.\ 151 and p.\ 155) to organize coordinates, by defining
\begin{align*}
\mathcal{K}_{L,s}  &  \equiv\left(  \Times_{j=1}^{\bar{\iota}_{T}(s)-1
}\mathcal{F}_{-\infty}^{j,\infty}\right)  \times\mathcal{F}_{-\infty
}^{\bar{\iota}_{T}(s),\bar{t}_{T}(s)}\times\left(
\Times _{j=\bar{\iota}_{T}(s)  +1}^{\infty}\mathcal{F}_{-\infty}%
^{j}\right)  \times\mathcal{C\quad}\text{if }\bar{t}_{T}(s) < T,\\
\mathcal{K}_{L,s}  &  \equiv\left(  \Times_{j=1}^{\bar{\iota}_{T}(s)
}\mathcal{F}_{-\infty}^{j,\infty}\right)  \times\left(  \Times_{j=\bar{\iota}_{T}(s) +1}^{\infty}\mathcal{F}_{-\infty}^{j}\right)  \times
\mathcal{C\quad}\text{if }\bar{t}_{T}(s) = T.
\end{align*}
 The construction guarantees that $\mathcal{K}_{L,s}\subset
\mathcal{X}$ by padding missing coordinates with the trivial field
$\mathcal{F}_{-\infty}^{j}=\left\{  \emptyset,\Xi_{j}\right\}  $.

We introduce the binary operator $\ominus$ that maps
two $\sigma$-fields generated by a sequence of random variables into
a $\sigma$-field generated by the non-overlapping portion of members
of the sequence. In particular, for $0<s<s^{\prime}$ and $\mathcal{F}_{-\infty}^{j,s^{\prime}},$
$\mathcal{F}_{-\infty}^{j,s}$ we define $\mathcal{F}_{-\infty}^{j,s^{\prime}}\ominus\mathcal{F}_{-\infty}^{j,s}\equiv\mathcal{F}_{s}^{j,s^{\prime}}$,
where $\mathcal{F}_{t}^{i,t+m}=\sigma\left(\eta_{it},...,\eta_{i,t+m}\right).$
We further define $\ominus$ to have the properties $\left(\Times_{j=1}^{\infty}\mathcal{A}_{nT}^{j}\right)\ominus\left(\Times_{j=1}^{\infty}\mathcal{B}_{nT}^{j}\right)\equiv\Times_{j=1}^{\infty}\left(\mathcal{A}_{nT}^{j}\ominus\mathcal{B}_{nT}^{j}\right)$
for $\mathcal{A}_{nT}^{j},\mathcal{B}_{nT}^{j}\subset\mathcal{F}_{-\infty}^{j,\infty}$,
where we understand $\mathcal{F}_{-\infty}^{j,\infty}\ominus\mathcal{F}_{-\infty}^{j,\infty}\equiv\mathcal{F}_{-\infty}^{j}$,
$\mathcal{F}_{-\infty}^{j}\ominus\mathcal{F}_{-\infty}^{j}\equiv\mathcal{F}_{-\infty}^{j}$
and ${\mathcal{C}}\ominus{\mathcal{C}}\equiv{\mathcal{C}}$. Using these properties we define the sigma fields $\mathcal{K}_{L,s}^{s^{\prime}}$
by $\mathcal{K}_{L,s}^{s^{\prime}}\equiv\mathcal{K}_{L,s^{\prime}}\ominus\mathcal{K}_{L,s}$.

Because $\mathcal{K}_{L,s}$ contains the coordinate $\mathcal{C}$ for all $L$
and $s$ the construction implies that when conditioning on $\mathcal{K}_{L,s}
$, all $\varkappa_{t}$ and $\zeta_{t}$ are held fixed. We have the following Lemma
relating mixing coefficients for $\left\{  \xi_{L,s},\mathcal{K}%
_{L,s}\right\}  $ to mixing coefficients for $\left\{  \xi_{it}%
,\mathcal{Y}_{-\infty}^{i,t}\right\}  $.\footnote{Condition
\ref{Conditional_NED} includes uniformity of the mixing and approximation
coefficients, which is used here.}

\begin{lemma}
\label{Cond_Joint_NED}Assume that $\left\{  \xi_{it},\mathcal{Y}%
_{-\infty}^{i,t}\right\}  $ satisfies Condition \ref{Conditional_NED} for
each $i$. Assume that
Condition \ref{Cond_Ind} holds. Then, for $r>2,$
$\left\{  \xi_{L,s},\mathcal{K}_{L,s}\right\}  $ satisfies $\sup_{L,s\leq
L}\left\Vert \xi_{L,s}\right\Vert _{\left.  r\right\vert \mathcal{C}}<\infty$
a.s., and there exists an array of $\mathcal{C}$-measurable, almost surely bounded,  random
variables $d_{L,s}\left(  \omega\right)  $ and a sequence of $\mathcal{C}%
$-measurable random variables $\beta_{\xi|\mathcal{C}}\left(  m\right)  $ such
that
\begin{align}
\left\Vert \xi_{L,s}-E\left[  \xi_{L,s}|\mathcal{K}_{L,s-m}^{s+m}\right]
\right\Vert _{\left.  2\right\vert \mathcal{C}}  &  \leq d_{L,s}\left(
\omega\right)  \beta_{\left.  \xi\right\vert \mathcal{C}}\left(  m\right)
\text{ } \text{for each } s\leq L \label{NED_joint}
\end{align}
where $\beta_{\left.  \xi\right\vert \mathcal{C}}\left(  m\right)  $ is of
size $-1/2$ a.s.. Furthermore, for $r>2,$ and
\begin{equation}
\alpha_{\left.  \xi\right\vert \mathcal{C}}\left(  m\right)  =\sup_{t}%
\sup_{A\in\mathcal{K}_{L,-\infty}^{t},B\in\mathcal{K}_{L,t+m}^{\infty}%
}\left\vert P\left(  A\cap B|\mathcal{C}\right)  -P\left(  A|\mathcal{C}%
\right)  P\left(  B|\mathcal{C}\right)  \right\vert \label{alpha_mix_joint}%
\end{equation}
it follows that $\alpha_{\left.  \xi\right\vert \mathcal{C}}\left(  m\right)$ is of size $r/\left(  r-2\right)  $ a.s.
\end{lemma}

\begin{proof}
The first claim follows immediately from $\xi_{L,s}=\xi_{it}$, where the $(i,t)=h_T^{-1}(s)$, and by applying
Condition \ref{Conditional_NED} to $\xi_{it}$.
To establish \eqref{NED_joint}, we first consider the case where
$m>\max\left(  T-\bar{t}_T(s)  ,\bar{t}_T(s)\right)  $. The way the filtration $\mathcal{K}_{L,s-m}^{s+m}$ is
constructed, it follows that $\xi_{L,s}$ is measurable with respect to $\mathcal{K}%
_{L,s-m}^{s+m}$. This then leads to%
\[
\left\Vert \xi_{L,s}-E\left[  \left.  \xi_{L,s}\right\vert \mathcal{K}%
_{L,s-m}^{s+m}\right]  \right\Vert _{\left.  2\right\vert \mathcal{C}}=0.
\]
Now, consider the case where $0\leq m\leq\max\left(  T-\bar{t}_T(s)  ,\bar{t}_T(s)\right)   $. We distinguish two
cases (a) $T-\bar{t}_T(s)  \geq \bar{t}_T(s) $, and (b) $T-\bar{t}_T(s)  <\bar{t}_T(s) $. For Case (a), if $0\leq m<\bar{t}_T(s) $, then\footnote{ For simplicity of notations, we adopt the convention that $(i,t)=(\bar{\iota}_{T}(s),\bar{t}_{T}(s))$.}
\[
E\left[  \left.  \xi_{L,s}\right\vert \mathcal{K}_{L,s-m}^{s+m}\right]
=E\left[  \xi_{it}\left\vert \mathcal{F}_{t-m}^{i,t+m}\vee\mathcal{C}%
\right.  \right]  ,
\]
and if $\bar{t}_T(s) \leq m<T-\bar{t}_T(s)  $, then by independence
\[
E\left[  \left.  \xi_{L,s}\right\vert \mathcal{K}_{L,s-m}^{s+m}\right]
=E\left[  \left.  \xi_{it}\right\vert \mathcal{F}_{-\infty}^{i,t+m}%
\vee\mathcal{C}\right]  .
\]
In the latter case, we have%
\begin{align*}
\left\Vert \xi_{L,s}-E\left[  \left.  \xi_{L,s}\right\vert \mathcal{K}%
_{L,s-m}^{s+m}\right]  \right\Vert _{\left.  2\right\vert \mathcal{C}}  &
=\left\Vert \xi_{it}-E\left[  \left.  \xi_{it}\right\vert \mathcal{F}%
_{nT,-\infty}^{i,t+m}\vee\mathcal{C}\right]  \right\Vert _{\left.
2\right\vert \mathcal{C}}\\
&  \leq\left\Vert \xi_{it}-E\left[  \left.  \xi_{it}\right\vert
\mathcal{F}_{t-m}^{i,t+m}\vee\mathcal{C}\right]  \right\Vert _{\left.
2\right\vert \mathcal{C}},
\end{align*}
because the residual in the projection on $\mathcal{F}_{t-m}^{i,t+m}%
\vee\mathcal{C}\subset\mathcal{F}_{-\infty}^{i,t+m}\vee\mathcal{C}$ has smaller variance, whereas in the former case we have
\[
\left\Vert \xi_{L,s}-E\left[  \left.  \xi_{L,s}\right\vert \mathcal{K}%
_{L,s-m}^{s+m}\right]  \right\Vert _{\left.  2\right\vert \mathcal{C}%
}=\left\Vert \xi_{it}-E\left[  \left.  \xi_{it}\right\vert
\mathcal{F}_{t-m}^{i,t+m}\vee\mathcal{C}\right]  \right\Vert _{\left.
2\right\vert \mathcal{C}}.
\]
Therefore, in both cases, we have%
\[
\left\Vert \xi_{L,s}-E\left[  \left.  \xi_{L,s}\right\vert \mathcal{K}%
_{n,s-m}^{s+m}\right]  \right\Vert _{\left.  2\right\vert \mathcal{C}}\leq
d_{nT,it}\left(  \omega\right)  \beta_{\left.  \xi\right\vert \mathcal{C}%
}\left(  m\right)
\]
by Condition \ref{Conditional_NED}. Relabeling the term $d_{nT,it}\left(
\omega\right)  =d_{L,s}\left(  \omega\right)  $ then yields the claimed
inequality in \eqref{NED_joint}.
For Case (b), if $0\leq m\leq T-\bar{t}_T(s)   $,
then
\[
E\left[  \left.  \xi_{L,s}\right\vert \mathcal{K}_{L,s-m}^{s+m}\right]
=E\left[  \xi_{it}\left\vert \mathcal{F}_{t-m}^{i,t+m}\vee\mathcal{C}%
\right.  \right]  ,
\]
and if $T-\bar{t}_T(s)   <m\leq \bar{t}_T(s) $, then 
\[
E\left[  \left.  \xi_{L,s}\right\vert \mathcal{K}_{L,s-m}^{s+m}\right]
=E\left[  \left.  \xi_{it}\right\vert \mathcal{F}_{t-m}^{i,\infty}%
\vee\mathcal{C}\right]  .
\]
In both cases, we have by similar reasoning
\[
\left\Vert \xi_{L,s}-E\left[  \left.  \xi_{L,s}\right\vert \mathcal{K}%
_{L,s-m}^{s+m}\right]  \right\Vert _{\left.  2\right\vert \mathcal{C}}\leq
d_{nT,it}\left(  \omega\right)  \beta_{\left.  \xi\right\vert \mathcal{C}%
}\left(  m\right).
\]

Finally, to establish \eqref{alpha_mix_joint} note that if $m>T-\bar{t}_T(s) $, it follows for $A\in\mathcal{K}_{L,-\infty
}^{s}$ and $B\in\mathcal{K}_{L,s+m}^{\infty}$ that
\begin{equation}
\left\vert P\left(  \left.  A\cap B\right\vert \mathcal{C}\right)  -P\left(
\left.  A\right\vert \mathcal{C}\right)  P\left(  \left.  B\right\vert
\mathcal{C}\right)  \right\vert =0\text{ a.s.} \label{alpha_zero}%
\end{equation}
because in that case $A$ and $B$ are independent by Condition \ref{Cond_Ind}. When $m\leq T-\bar{t}_T(s)  $, which implies that $\bar{\iota}_{T}(s) = \bar{\iota}_{T}(s+m)$, or in other words that conditioning happens within the same cross-sectional unit,  we consider
the following. The product space structure of $\mathcal{K}_{L,-\infty}^{s}$
implies that $A=A_{1}\times A_{2}\times A_{3}\times A_{4}$,
where\footnote{See \cite[p.154 and theorem 38.A]{halmos1976measure}.} $A_{1}    \in\Times_{j=1}^{\bar{\iota}_{T}(s)-1  }\mathcal{F}_{-\infty
}^{j}$, $A_{2}    \in\mathcal{F}_{-\infty}^{\bar{\iota}_{T}(s),\bar{t}_T(s)}$, $A_{3}   \in\Times_{j=\bar{\iota}_{T}(s) +1}^{\infty}\mathcal{F}%
_{nT,-\infty}^{j}$, and $ A_{4}    \in\mathcal{C}$.
In the same way, $B=B_{1}\times B_{2}\times B_{3}\times B_{4}$, where
$B_{1}   \in\Times_{j=1}^{\bar{\iota}_{T}(s+m)-1  }\mathcal{F}_{-\infty}^{j}$, $B_{2}   \in\mathcal{F}_{\bar{t}_{T}(s+m)}^{\bar{\iota}_{T}(s+m),\infty}$, $B_{3}   \in\Times_{j=\bar{\iota}_{T}(s+m) +1}^{\infty}\mathcal{F}_{-\infty}^{j}$, $B_{4}  \in\mathcal{C} $.
Note that
\begin{equation}
A_{2}\times A_{4}\in\mathcal{Y}_{-\infty}^{\bar{\iota}_{T}(s),\bar{t}_{T}(s)}=\mathcal{F}_{-\infty}^{\bar{\iota}_{T}(s),\bar{t}_{T}(s)}\times\mathcal{C} \label{A2_A4}%
\end{equation}
and
\begin{equation}
B_{2}\times B_{4}\in\mathcal{Y}_{\bar{t}_{T}(s+m)}^{\bar{\iota}_{T}(s+m),\infty}=\mathcal{F}_{\bar{t}_{T}(s+m)}^{\bar{\iota}_{T}(s+m),\infty}\times\mathcal{C}. \label{B2_B4}%
\end{equation}
Since $A\cap B=\Times_{j=1}^{4}\left(  A_{j}\cap B_{j}\right)  $ by p.\ 139 in \cite{halmos1976measure}, and
\[
P\left(  A\cap B|\mathcal{C}\right)  =\left(
{\textstyle\prod\nolimits_{j=1}^{3}}
\tilde{P}_{j}\left(  A_{j}\cap B_{j}\right)  \right)  P_{f}\left(  A_{4}\cap
B_{4}|\mathcal{C}\right)
\]
where $\tilde{P}_{j}$ are themselves products of measures $P_{u,j}$
corresponding to the respective coordinates in $A_{j}$ and $B_{j}$ and where
$A_{j},B_{j}$ for $j\leq3$ are independent of $\mathcal{C}$. We note that
$A_{1}$ and $B_{3}$ consist of coordinates that are either $\Xi_{j}$ or
$\emptyset.$ By Theorem 33.A in \cite{halmos1976measure}, $B=\emptyset$ if and only if at
least one coordinate is equal to $\emptyset.$ Thus, by the properties of
conditional expectations, $P\left(  A\cap B|\mathcal{C}\right)  =0$ if one of
the coordinates of $B$ or $A$ is equal to $\emptyset.$ Thus, the mixing
coefficient is zero in this case.
We therefore assume without loss of generality that for $A_{1}$ and $B_{3}$
all coordinates are equal to $\Xi_{j}.$ This implies that $A_{1}\cap
B_{1}=B_{1}$ and $A_{3}\cap B_{3}=A_{3}$, and we have $\tilde{P}_{1}\left(
A_{1}\cap B_{1}\right)  =\tilde{P}_{1}\left(  B_{1}\right)  $ and $\tilde
{P}_{1}\left(  A_{3}\cap B_{3}\right)  =\tilde{P}_{3}\left(  A_{3}\right)  $. 
Note that $P_{f}\left(  A_{4}\cap B_{4}|\mathcal{C}\right)  =E\left[
1_{A_{4}}1_{B_{4}}|\mathcal{C}\right]  =1_{A_{4}}1_{B_{4}}$ since $A_{4}$ and
$B_{4}$ are measurable with respect to $\mathcal{C}$. Finally, note that $\tilde{P}%
_{2}\left(  A_{2}\cap B_{2}\right)  =P_{u,\bar{\iota}_{T}(s)}\left(
A_{2}\cap B_{2}\right)  $. (Recall that $\bar{\iota}_{T}(s) = \bar{\iota}_{T}(s+m)$.)
These arguments lead to 
\begin{equation}
P\left(  A\cap B|\mathcal{C}\right)  =\tilde{P}_{1}\left(  B_{1}\right)
\tilde{P}_{3}\left(  A_{3}\right)  P_{u,\bar{\iota}_{T}(s)}\left(
A_{2}\cap B_{2}\right)  1_{A_{4}}1_{B_{4}}. \label{P_AxB}%
\end{equation}
Similar arguments also show that
\begin{equation}
P\left(  A|\mathcal{C}\right)  =\tilde{P}_{1}\left(  A_{1}\right)
P_{u,\bar{\iota}_{T}(s)}\left(  A_{2}\right)  1_{A_{4}} \label{P_A}%
\end{equation}
and
\begin{equation}
P\left(  B|\mathcal{C}\right)  =\tilde{P}_{3}\left(  B_{3}\right)
P_{u,\bar{\iota}_{T}(s)}\left(  B_{2}\right)  1_{B_{4}} \label{P_B}%
\end{equation}
such that
\begin{align}
&  \left\vert P\left(  A\cap B|\mathcal{C}\right)  -P\left(  A|\mathcal{C}%
\right)  P\left(  B|\mathcal{C}\right)  \right\vert \label{alpha_approx}\\
&  =\left\vert P_{u,\bar{\iota}_{T}(s)}\left(  A_{2}\cap B_{2}\right)
-P_{u,\bar{\iota}_{T}(s)}\left(  A_{2}\right)  P_{u,\bar{\iota}_{T}(s)}\left(  B_{2}\right)  \right\vert \tilde{P}_{1}\left(  B_{1}\right)
\tilde{P}_{3}\left(  A_{3}\right)  1_{A_{4}}1_{B_{4}}\nonumber\\
&  \leq\left\vert P_{u,\bar{\iota}_{T}(s)}\left(  A_{2}\cap B_{2}\right)
-P_{u,\bar{\iota}_{T}(s)}\left(  A_{2}\right)  P_{u,\bar{\iota}_{T}(s)}\left(  B_{2}\right)  \right\vert 1_{A_{4}}1_{B_{4}}\nonumber\\
&  \leq\sup_{A\in\mathcal{Y}_{-\infty}^{\bar{\iota}_{T}(s),\bar{t}_{T}(s)},B\in\mathcal{Y}_{\bar{t}_{T}(s+m)}^{\bar{\iota}_{T}(s),\infty}}\left\vert P\left(  A\cap B|\mathcal{C}%
\right)  -P\left(  A|\mathcal{C}\right)  P\left(  B|\mathcal{C}\right)
\right\vert \nonumber\\
&  \leq\alpha_{\left.  \xi\right\vert \mathcal{C}}\left(  m\right)  ,\nonumber
\end{align}
where the equality combines results \eqref{P_AxB}, \eqref{P_A} and \eqref{P_B}, the
first inequality uses $\tilde{P}_{1}\left(  B_{1}\right)  \tilde{P}_{3}\left(
A_{3}\right)  \leq1$, the second inequality uses
\begin{align*}
P_{u,\bar{\iota}_{T}(s)}\left(  A_{2}\cap B_{2}\right)  1_{A_{4}}%
1_{B_{4}}  &  =P_{u,\bar{\iota}_{T}(s)}\left(  A_{2}\cap B_{2}\right)
P_{f}\left(  A_{4}\cap B_{4}|\mathcal{C}\right) \\
&  =P_{u,\bar{\iota}_{T}(s)}\left(  \left(  A_{2}\times A_{4}\right)
\cap\left(  B_{2}\times B_{4}\right)  |\mathcal{C}\right)
\end{align*}
as well as \eqref{A2_A4} and \eqref{B2_B4}. The last inequality in
\eqref{alpha_approx} uses \eqref{alpha_Y}. Then, combining \eqref{alpha_zero}
and \eqref{alpha_approx} show that \eqref{alpha_mix_joint} holds.
\end{proof}

Lemma \ref{Cond_Joint_NED} is the basis for establishing a joint CLT as both
$n,T\rightarrow\infty.$ For this, the earlier blocking definitions need to be adjusted.
Let $L\equiv nT,$ and $n,T\rightarrow\infty$. Define $b_{L}$ and $l_{L}$ to be
positive, non-decreasing integer valued sequences that are the length of
included and discarded blocks. Assume $b_{L}\geq l_{L}+1,$ $l_{L}%
\rightarrow\infty,$ $l_{L}\geq1,$ $b_{L}\leq n$, $b_{L}/L\rightarrow0$ and
$l_{L}/b_{L}\rightarrow0$. Let $r_{L}\equiv\left[ L/b_{L}\right]  $ and
define $\tilde{\xi}_{L,s}\equiv\xi_{L,s}/\sigma_{L,\xi}(\omega)$,\footnote{ Note that $\sigma_{L,\xi}$ in \eqref{sigma_n} is a random variable. In order to emphasize it, we write $\tilde{\xi}_{L,s}\equiv\xi_{L,s}/\sigma_{L,\xi}(\omega)$. } where
\begin{equation}
\sigma_{L,\xi}^{2}=E\left[ \left. \left(  \sum_{s=1}^{L}\xi
_{L,s}\right)  ^{2} \right\vert
\mathcal{C}\right]  . \label{sigma_n}%
\end{equation}

We impose the following conditions that are modifications of the conditions
for the case of marginal convergence when $T\rightarrow\infty.$

\begin{condition}
\label{Conditional_NED_CLT_joint}Assume that $\left\{  \xi_{it}%
,\mathcal{Y}_{-\infty}^{i,t}\right\}  $ satisfies Condition
\ref{Conditional_NED} for each $i\leq n$ and all $T,n\geq1.$ For $\left\{
\xi_{L,s},\mathcal{K}_{L,s}\right\}  $ impose the following additional
restrictions. There exists a positive array of random variables $g_{L,s}%
\left(  \omega\right)  $ such that $d_{L,s}\left(  \omega\right)
/g_{L,s}\left(  \omega\right)  $ is uniformly bounded $P$-a.s. for all $s\leq
L$ and all $L\geq1$ and
\begin{equation}
\left\Vert \tilde{\xi}_{L,s}/g_{L,s}\left(  \omega\right)  \right\Vert
_{\left.  r\right\vert \mathcal{C}}<\infty\label{Cond_UI}%
\end{equation}
for $r>2$ for all $s\leq L$ and all $L\geq1$. Let $M_{L,j}\equiv\max_{\left(
j-1\right)  b_{L}+1\leq s\leq jb_{L}}g_{L,s}\left(  \omega\right)  $ for
$1\leq j\leq r_{L}$, and $M_{L,r_{L}+1}\equiv\max_{r_{L}b_{L}+1\leq s\leq
L}g_{L,s}\left(  \omega\right)  $. Then,
\[
b_{L}^{1/2}\left(\max_{1\leq j\leq r_{L}+1}M_{L,j}\right)\rightarrow0\text{ \ensuremath{P}-a.s.}
\]
and
\[
\left|b_{L}\left(\sum_{j=1}^{r_{L}}M_{L,j}^{2}\right)\right|\leq C_{M}<\infty\text{ \ensuremath{P}-a.s}.
\]
for some constant $C_{M}$ and all $L$ large enough. 

\end{condition}

We now state the following conditional CLT for joint convergence.

\begin{theorem}
\label{Joint_Cond_CLT}Assume that Conditions \ref{Cond_Ind} and  
\ref{Conditional_NED_CLT_joint} hold. Then, as $n,T\rightarrow\infty,$ 
\[
\lim_{T,n\rightarrow\infty}E\left[  \left.  e^{\iota\tau L^{-1/2}\sum
_{s=1}^{L}\tilde\xi_{L,s}}\right\vert
\mathcal{C}\right]  =e^{-\frac{1}{2}\tau^{2}}\text{ a.s.}%
\]
If in addition, $L^{-1}\sigma_{L,\xi}^{2}\left(  \omega\right)
\rightarrow\sigma_{\xi}^{2}$ P-a.s. where $\sigma_{\xi}^{2}$ is a constant, then it follows that
\[
\lim_{T,n\rightarrow\infty}E\left[  \left.  e^{\iota\tau L^{-1/2}\sum
_{s=1}^{L}\xi_{L,s}}\right\vert \mathcal{C}\right]  =e^{-\frac{1}{2}\tau
^{2}\sigma_{\xi}^{2}}\text{ a.s.}%
\]

\end{theorem}

\begin{proof}
The proof follows the proof of Theorem 2 in \cite{eaglson1975}. Lemma \ref{Cond_Joint_NED}
shows that $\left\{ \xi_{L,s},\mathcal{K}_{L,s}\right\} $ is a NED
process with coefficients $\beta_{\xi|\mathcal{C}}\left(m\right)$
of size $-1/2$ and is based on a mixing sequence with coefficients
$\alpha_{\left.\xi\right\vert \mathcal{C}}\left(m\right)$ of size
$-1/2$, where $\beta_{\xi|\mathcal{C}}\left(m,\omega\right)$ and
$\alpha_{\left.\xi\right\vert \mathcal{C}}\left(m,\omega\right)$
are $\mathcal{C}$-measurable random variables and where we add the
$\omega$ argument to emphasize this fact. By Lemma \ref{lem_Eagleson},
the conditional expectations on $\left(\Omega,\mathcal{X},P\right)$
and $\left(\Omega,\mathcal{X},Q_{\omega^{\prime}}\right)$ coincide
with $P$-probability one. Therefore, for all $\omega^{\prime}\in\Omega$,
except for a set of $P$-measure zero, $\left\{ \xi_{it},\mathcal{Y}_{-\infty}^{i,t}\right\} $
satisfies the same NED process conditions with  the same mixing coefficients $Q_{\omega}$-a.s. The  quantities  $\beta_{\xi|\mathcal{C}}\left(m,\omega'\right)$ and $\alpha_{\left.\xi\right\vert \mathcal{C}}\left(m,\omega'\right)$
are constants for the probability space $\left(\Omega,\mathcal{X},Q_{\omega^{\prime}}\right)$.
Therefore, the remaining assumptions in Condition \ref{Conditional_NED_CLT_joint}
guarantee that for all $\omega^{\prime}\in\Omega,$ except for a set
of $P$-measure zero, $\left\{ \xi_{n,s},\mathcal{K}_{L,s}\right\} $
satisfies the conditions of Theorem 2 in \cite{dejong1997} for the probability
space $\left(\Omega,\mathcal{X},Q_{\omega^{\prime}}\right)$, except
on a set of $Q_{\omega^{\prime}}$-probability of zero. Then, for
$P$-almost all $\omega^{\prime}$ 
\begin{equation}\label{Joint_Cond_CLT:eq1}
\lim_{T,n\rightarrow\infty}E\left[\left.e^{\iota\tau L^{-1/2}\sum_{s=1}^{L}\tilde \xi_{L,s}}\right\vert \mathcal{C}\right]=\lim_{T,n\rightarrow\infty}E_{\omega^{\prime}}\left[e^{\iota\tau L^{-1/2}\sum_{s=1}^{L}\tilde \xi_{L,s}}\right]=e^{-\frac{1}{2}\tau^{2}}\text{ a.s.}
\end{equation}
where the last equality follows from applying Theorem 2 in \cite{dejong1997} on $\left(\Omega,\mathcal{X},Q_{\omega^{\prime}}\right)$ to a
set of $Q_{\omega^{\prime}}$ probability one, which establishes the
first result.

For the second statement, define the $\left(\Omega,\mathcal{X},Q_{\omega^{\prime}}\right)$-constant random variable
$$\sigma_{L,\xi}^{2}(\omega')=E_{\omega'}\left[  \left(  \sum_{s=1}^{L}\xi
_{L,s}\right)  ^{2}\right].$$  By  Lemma \ref{lem_Eagleson} $\sigma_{L,\xi}^{2}(\omega')$  is identical to $\sigma_{L,\xi}^{2}$ defined in (\ref{sigma_n})  except on a set of $Q_{\omega^{\prime}}$ probability zero.
Then, it follows that $L^{-1}\sigma_{L,\xi}^{2}\left(\omega'\right)\rightarrow\sigma_{\xi}^{2}$
with $Q_{\omega^{\prime}}$-probability one on $\left(\Omega,\mathcal{X},Q_{\omega^{\prime}}\right)$
because of the maintained assumption of the theorem that $L^{-1}\sigma_{L,\xi}^{2}\left(\omega\right)\rightarrow\sigma_{\xi}^{2}$  $P$ -a.s. The result follows
from applying the continuous mapping theorem to \eqref{Joint_Cond_CLT:eq1}.  
\end{proof}

\section{Proofs for Section \ref{sec:clt}\label{section-proof-test}}


In this appendix we provide a formal justification of the inference procedure proposed in Section \ref{sec:clt} by relying on the high dimensional central limit theorem of \cite{chernozhuokov2022improved}.
To this end, we impose the following assumptions.

\begin{assumption}\label{ass:cltcond1}
There are $\{\psi_i\}_{i=1}^{b_n}$ that are independent conditionally on $\mathcal G_n$ and satisfy $\psi_i \in \R^q$, $E[\psi_i|\mathcal G_n] = 0$ for all $1\leq i \leq b_n$, and for any $\varrho >0$
\begin{equation*}
P\left(  \left.  \left\Vert \sum_{i=1}^{n}f\left(  V_{i},\hat{\theta}%
_{n}\right)  -\sum_{i=1}^{b_{n}}\psi_{i}\right\Vert _{\infty}>\frac
{\varrho\sqrt{b_{n}}}{\sqrt{\log(q)}}\right\vert \mathcal{G}_{n}\right)  \leq
r_{1n}\left(  \varrho,\mathcal{G}_{n}\right) .
\end{equation*}
\end{assumption}

\begin{assumption}\label{ass:cltcond2}
For some $B_n(\mathcal G_n)$ satisfying $1\leq B(\mathcal G_n) < \infty$ with probability one 
\begin{equation*}
\max_{1\leq i\leq b_{n}}\max_{1\leq k\leq q}E\left[  \left.  \exp\left(
\frac{\left\vert \psi_{ki}\right\vert }{B_{n}(\mathcal{G}_{n})}\right)
\right\vert \mathcal{G}_{n}\right]  \leq2.
\end{equation*}
\end{assumption}

\begin{assumption}\label{ass:cltcond3}
There are  $\underline{\sigma}>0$ and $\bar{\sigma}<\infty$ not depending on $\mathcal G_n$ satisfying
\begin{equation*}
P\left(  \left.  \underline{\sigma}^{2}\leq\min_{1\leq k\leq q}\frac{1}{b_{n}%
}\sum_{i=1}^{b_{n}}E[\psi_{ki}^{2}|\mathcal{G}_{n}],\max_{1\leq k\leq q}%
\frac{1}{b_{n}}\sum_{i=1}^{b_{n}}E[\psi_{ki}^{4}|\mathcal{G}_{n}]\leq
\bar{\sigma}^{2}B_{n}(\mathcal{G}_{n})\right\vert \mathcal{G}_{n}\right)  =1.
\end{equation*}
\end{assumption}

\begin{assumption}\label{ass:cltcond4}
For $\mathcal D_n \equiv (\{Y_i,X_i,S_i\}_{i=1}^n,\Sh,\mathcal G_n)$ we have for any $\varrho > 0$ that
\begin{equation*}
P\left(  \left.  \left\Vert \sum_{i=1}^{b_{n}}\omega_{i}\left(  (\hat{\psi}_{i}-\psi_i) -\frac{1}{b_{n}} \sum_{j=1}^{b_{n}} (\hat{\psi}_{j}-\psi_j)  \right)  \right\Vert _{\infty}>\frac{\varrho\sqrt{b_{n}}}{\sqrt
{\log(q)}}\right\vert \mathcal{D}_{n}\right)  \leq r_{2n}(\varrho
,\mathcal{D}_{n}). 
\end{equation*}
\end{assumption}

\begin{assumption}\label{ass:cltweights}
The weights $\{\omega_i\}_{i=1}^{b_n}$ are independent of $\{\mathcal D_n\}$ and either: (i) Follow a multinomial distribution with parameter $b_n$ and probabilities $(1/b_n,\ldots, 1/b_n)$, or (ii) Are i.i.d.\ with $\omega_i = \omega_{i,1}+\omega_{i,2}$ with $\omega_{i,1} \sim N(0,\sigma^2_\omega)$ for some $\sigma^2_\omega \geq 0$ and $|\omega_{i,2}|\leq 3$. 
\end{assumption}

Assumption \ref{ass:cltcond1} demands that the vector of moments be asymptotically equivalent to a $q$-dimensional sample mean of random variables $\{\psi_{i}\}_{i=1}^{b_n}$. 
We note that Assumption \ref{ass:cltcond1} effectively requires that the null hypothesis be true by requiring that the variables $\{\psi_i\}_{i=1}^{b_n}$ have mean zero.
In turn, Assumptions \ref{ass:cltcond2} and \ref{ass:cltcond3} imposes moment restrictions on the variables $\{\psi_i\}_{i=1}^{b_n}$ that ensure that the high dimensional central limit theorem of \cite{chernozhuokov2022improved} is applicable.
Finally, Assumption \ref{ass:cltcond4} demands a linearization requirement on our bootstrap statistic, while Assumption \ref{ass:cltweights} states requirements on the weights $\{\omega_i\}_{i=1}^{b_n}$ that we may employ.
We note, in particular, that Assumption \ref{ass:cltweights} allows for the empirical bootstrap (through Assumption \ref{ass:cltweights}(i)) and the use of Standard Gaussian, Rademacher, or \cite{mammen:1993} weights (through Assumption \ref{ass:cltweights}(ii)).

Our next result encompasses Proposition \ref{prop-clt} as a special case. 
The first and second parts of the result provide conditions under which the level of the test is $1-\alpha$ \emph{unconditionally} on $\mathcal G_n$ and \emph{conditionally} on $\mathcal G_n$ respectively.
We view the unconditional result as appropriate for the asymptotic framework in \cite{adao2019shift} (in which elements of $\mathcal G_n$ are resampled), and the conditional result as more suitable for the asymptotic framework in \cite{goldsmith2020bartik} (in which $\mathcal G_n$ is not resampled). 

\begin{lemma}\label{lm:clt}
Let Assumptions \ref{ass:cltcond1}, \ref{ass:cltcond2}, \ref{ass:cltcond3}, \ref{ass:cltcond4}, and \ref{ass:cltweights} hold.
(i) If $r_{1n}(\varrho,\mathcal G_n) \vee r_{2n}(\varrho,\mathcal D_n) = o_P(1)$ for any $\varrho > 0$ and $B_n^2(\mathcal G_n) \log^5(qb_n)/b_n = o_P(1)$, then 
$$\lim_{n \to \infty} P(T_n \leq \hat c_n) = 1-\alpha.$$
(ii) If for any constants $\epsilon, \varrho > 0$ we have $r_{1n}(\varrho,\mathcal G_n) = o_{as}(1)$ and $P(r_{2n}(\varrho,\mathcal D_n) > \epsilon|\mathcal G_n) = o_{as}(1)$, and in additiong $B_n^2(\mathcal G_n) \log^5(qb_n)/b_n = o_{as}(1)$, then it follows that
$$\lim_{n\to \infty} P(T_n \leq \hat c_n|\mathcal G_n) = 1-\alpha ~ \text{ a.s. }$$
\end{lemma}

\noindent \emph{Proof.} We begin by defining the conditional covariance matrix $\Sigma(\mathcal G_n)$ to be given by
\begin{equation*}
\Sigma(\mathcal G_n) \equiv \frac{1}{b_n}\sum_{i=1}^{b_n} E\left[\psi_i \psi_i^\prime|\mathcal G_n\right]
\end{equation*}
and letting $\mathbb T_n \equiv \|\mathbb G_n\|_\infty$ for $\mathbb G_n\in \R^q$ a Gaussian vector satisfying $\mathbb G_n \sim N(0,\Sigma(\mathcal G_n))$.
Further denote the linearized versions of $T_n$ and $T_n^*$ by letting $L_n$ and $L_n^*$ equal
\begin{equation}\label{lm:clt2}
L_n \equiv   \left\| \frac{1}{\sqrt {b_n}} \sum_{i=1}^{b_n} \psi_i\right\|_\infty \hspace{0.5 in }    
L_n^* \equiv \left\| \frac{1}{\sqrt {b_n}} \sum_{i=1}^{b_n} \omega_i\left(\psi_i - \frac{1}{b_n}\sum_{j=1}^{b_n}\psi_j\right)\right\|_\infty 
\end{equation}
and for notational convenience set $\delta_n(\mathcal G_n) \equiv (B_n^2(\mathcal G_n)\log^5(qb_n)/b_n)^{1/4}$.
By Theorem 2.1 in \cite{chernozhuokov2022improved} there then exists a $C_1$ not depending on $\mathcal G_n$ such that
\begin{equation}\label{lm:clt3}
\sup_{t\in\R}\left\vert P(\left.  L_n\leq t\right\vert \mathcal{G}_{n})-P(\left.  \mathbb{T}_{n}\leq t\right\vert \mathcal{G}_{n})\right\vert \leq C_{1}\delta_{n}(\mathcal{G}_{n}) .
\end{equation}
Next note that, for any constant $\varrho > 0$, result \eqref{lm:clt2} allows us to conclude that
\begin{align}\label{lm:clt4}
P & \left(  \left.  b_{n}^{-1/2}T_{n}\leq t\right\vert \mathcal{G}_{n}\right) \notag\\
&  \geq P\left(  \left. L_n\leq t-\varrho\right\vert \mathcal{G}_{n}\right)  -P\left(  \left.  |b_{n}^{-1/2}T_{n}-L_{n}|>\varrho\right\vert \mathcal{G}_{n}\right) \notag\\
&  \geq P\left(  \left.  \mathbb{T}_{n}\leq t-\varrho\right\vert
\mathcal{G}_{n}\right)  -C_{1}\delta_{n}(\mathcal{G}_{n})-P\left(  \left.
|b_{n}^{-1/2}T_{n}-L_n|>\varrho\right\vert \mathcal{G}_{n}\right) \notag\\
&  \geq P\left(  \left.  \mathbb{T}_{n}\leq t\right\vert \mathcal{G}_{n}\right)  -C_{2}\varrho\sqrt{\log(q)}-C_{1}\delta_{n}(\mathcal{G}_{n})-P\left(  \left.  |b_{n}^{-1/2}T_{n}-L_n|>\varrho\right\vert
\mathcal{G}_{n}\right) \notag\\
&  \geq P\left(  \left.  L_n\leq t\right\vert \mathcal{G}_{n}\right)
-C_{2}\eta-2C_{1}\delta_{n}(\mathcal{G}_{n})-r_{1n}(\eta,\mathcal{G}_{n}),
\end{align}
where the third inequality holds holds for some constant $C_{2}$ not depending on $\mathcal{G}_{n}$ by Lemma
J.3 in \cite{chernozhuokov2022improved}, and the final inequality holds for any $\eta > 0$ by result \eqref{lm:clt3}, Assumption \ref{ass:cltcond1}, and setting $\varrho = \eta/\sqrt{\log(q)}$.

For $\mathcal D_n \equiv (\{Y_i,X_i,S_i\}_{i=1}^n,\Sh,\mathcal G_n)$ and any constant $C_3 > 0$ next define the event
\begin{equation*}
\mathcal{E}(\mathcal{D}_{n})\equiv 1\left \{\sup_{t\in \R}|P(L_n \leq t|\mathcal{G}_{n})-P(L_n^* \leq t|\mathcal{D}_{n})|\leq C_{3}\delta_{n}(\mathcal{G}_{n})\right\}. 
\end{equation*}
By Assumption \ref{ass:cltweights} and Lemmas 4.5 and 4.6 in \cite{chernozhuokov2022improved} it then follows that we may select a constant $C_3$ not depending on $\mathcal G_n$ under which we have
\begin{equation}\label{lm:clt6}
P\left(  \left.  \mathcal{E}(\mathcal{D}_{n})=1\right\vert \mathcal{G}_{n}\right)  \geq1-\frac{2}{b_{n}}-3v_{n}(\mathcal{G}_{n}),\hspace{0.3in}v_{n}^{2}(\mathcal{G}_{n})=\frac{B_{n}^{2}(\mathcal{G}_{n})\log^{3}(qb_{n})}{b_{n}}. 
\end{equation}
Applying \eqref{lm:clt6} and the same arguments as in \eqref{lm:clt4} we obtain that if $\mathcal{E}(\mathcal{D}_{n})=1$, then
\begin{align}\label{lm:clt7}
 P &  \left(  \left.  b_{n}^{-1/2}T_{n}^{\ast}\leq t\right\vert \mathcal{D}_{n}\right) \notag\\
&  \leq P\left(  \left.  L_{n}^{\ast}\leq t+\varrho\right\vert \mathcal{D}_{n}\right)  +P\left(  \left.  |b_{n}^{-1/2}T_{n}^{\ast}-L_{n}^{\ast}|>\varrho\right\vert \mathcal{D}_{n}\right) \notag\\
&  \leq P\left(  \left.  L_{n}^{\ast}\leq t\right\vert \mathcal{D}_{n}\right)  +C_{2}\varrho\sqrt{\log(q)}+2(C_{1}+C_{3})\delta_{n}(\mathcal{G}_{n})+P\left(  \left.  |b_{n}^{-1/2}T_{n}^{\ast}-L_{n}^{\ast}|>\varrho\right\vert \mathcal{D}_{n}\right) \notag\\
&  \leq P\left(  \left.  L_{n}^{\ast}\leq t\right\vert \mathcal{D}_{n}\right)  +C_{2}\eta+2(C_{1}+C_{3})\delta_{n}(\mathcal{G}_{n})+r_{2n}(\eta,\mathcal{D}_{n}) 
\end{align}
for any constant $\eta>0$. 
Next, plug in $t=\hat{c}_{n}b_{n}^{-1/2}$ into result \eqref{lm:clt7} and note that result \eqref{lm:clt4} then implies that whenever $\mathcal E(\mathcal D_n) = 1$ we must have
\begin{align}\label{lm:clt8}
1-\alpha &  \leq P\left(  \left.  L_{n}^{\ast}\leq\hat{c}_{n}b_{n}^{-1/2}\right\vert \mathcal{D}_{n}\right)  +C_{2}\eta+2(C_{1}+C_{3})\delta_{n}(\mathcal{G}_{n})+r_{2n}(\eta,\mathcal{D}_{n})\notag\\
&  \leq P\left(  \left.  L_n \leq\hat{c}_{n}b_{n}^{-1/2}\right\vert
\mathcal{G}_{n}\right)  +C_{2}\eta+(2C_{1}+3C_{3})\delta_{n}(\mathcal{G}_{n})+r_{2n}(\eta,\mathcal{D}_{n})\notag\\
&  \leq P\left(  \left.  T_{n}\leq\hat{c}\right\vert \mathcal{G}_{n}\right)
+2C_{2}\eta+4(C_{1}+C_{3})\delta_{n}(\mathcal{G}_{n})+r_{2n}(\eta
,\mathcal{D}_{n})+r_{1n}(\eta,\mathcal{G}_{n})\notag\\
&  \equiv P(\left.  T_{n}\leq\hat{c}_{n}\right\vert \mathcal{G}_{n}%
)+s_{n}(\eta,\mathcal{D}_{n}), 
\end{align}
where the final equality is definitional. 

To conclude, let $\mathcal A_n$ be a sigma field satisfying $\mathcal G_n \subseteq \mathcal A_n$ and note that for any $\varrho > 0$ we obtain from result \eqref{lm:clt6} and the law of iterated expectations that
\begin{multline}\label{lm:clt9}
P(\mathcal E(\mathcal D_n)=1|\mathcal A_n)  \geq E[1\{v_n(\mathcal G_n) \leq \varrho\}P(\mathcal E(\mathcal D_n) = 1|\mathcal G_n)|\mathcal A_n] \\ \geq P(v_n(\mathcal G_n) \leq \varrho|\mathcal A_n)\left(1-\frac{2}{b_n} - 3\varrho\right) .
\end{multline}
Similarly, for any $\varrho > 0$ we may select $\eta > 0$ sufficiently small so as to ensure that
\begin{equation}\label{lm:clt10}
P(s_n(\eta,\mathcal D_n) < \varrho|\mathcal A_n) \geq P( \delta_n(\mathcal G_n) \vee r_{1n}(\eta,\mathcal G_n) \vee r_{2n}(\eta,\mathcal D_n) > \epsilon \varrho|\mathcal A_n)
\end{equation}
for some $\epsilon > 0$.
Finally, observe that  \eqref{lm:clt8} and the law of iterated expectations yield
\begin{align}\label{lm:clt11}
P(T_{n}\leq\hat{c}_{n}|\mathcal A_n)  &  \geq E[1\{T_{n}\leq\hat{c}_{n}\}1\{P(T_{n}\leq \hat{c}_{n}|\mathcal{G}_{n})\geq1-\alpha-\varrho\}|\mathcal A_n] \notag\\
&  \geq(1-\alpha-\varrho)E[1\{P(T_{n}\leq\hat{c}_{n}|\mathcal{G}_{n}
)\geq1-\alpha-\varrho\}|\mathcal A_n]\nonumber\\
&  \geq(1-\alpha-\varrho)P(s(\eta,\mathcal{D}_{n})<\varrho,~\mathcal{E}%
(\mathcal{D}_{n})=1|\mathcal A_n).
\end{align}
Part (i) of the lemma therefore follows from $\varrho$ being arbitrary, results \eqref{lm:clt9}, \eqref{lm:clt10}, \eqref{lm:clt11}, and setting $\mathcal A_n$ to be the trivial sigma field. 
Part (ii) of the lemma similarly follows from $\varrho$ being arbitrary, results \eqref{lm:clt9}, \eqref{lm:clt10}, \eqref{lm:clt11}, and setting $\mathcal A_n = \mathcal G_n$. \qed

\addcontentsline{toc}{section}{References}

{\small

\singlespace

}




\setcounter{footnote}{0}

\renewcommand{\thesection}{S.\arabic{section}}
\renewcommand{\theequation}{S.\arabic{equation}}
\renewcommand{\thelemma}{S.\arabic{section}.\arabic{lemma}}
\renewcommand{\thecorollary}{S.\arabic{section}.\arabic{corollary}}
\renewcommand{\thetheorem}{S.\arabic{section}.\arabic{theorem}}
\renewcommand{\theassumption}{S.\arabic{section}.\arabic{assumption}}
\setcounter{lemma}{0}
\setcounter{theorem}{0}
\setcounter{corollary}{0}
\setcounter{equation}{0}
\setcounter{remark}{0}
\setcounter{section}{0}
\setcounter{assumption}{0}

\emptythanks

\title{Supplemental Appendix}

\author{Jinyong Hahn\\hahn@econ.ucla.edu \\ UCLA
\and Guido Kuersteiner\\gkuerte@umd.edu\\University of Maryland
\and Andres Santos\\andres@econ.ucla.edu\\UCLA
\and Wavid Willigrod\\dwwilligrod@gmail.com\\ UCLA}

\maketitle

This supplemental appendix includes: (i) Calculations that justify the asymptotic validity of the proposed overidentifiation tests; and (ii) A set of Monte Carlo experiments evaluating the finite sample performance of such tests.

\section{Influence Functions} \label{sec:suppcalc}

In this section, we discuss how to verify the conditions of Lemma \ref{lm:clt}, and hence Proposition \ref{prop-clt}, in the context of the overidentification tests of Sections \ref{test:gss} and \ref{test:akm}.

\subsection{Conditioning on Shocks}\label{sec:suppcalcgss}


We first examine the overidentification test introduced in Section \ref{test:gss}, which is designed for applications in which $\mathcal G_n$ denotes a set of aggregate shocks that include $\Sh$.
Recall that in Section \ref{test:gss} we set  $\psi_{ij} \equiv U_{ij}/\sigma_j$ and $\hat \psi_{ij} \equiv \hat U_{ij}/\hat \sigma_j$ with $U_{ij}$ and $\hat U_{ij}$  given by
\begin{align}\label{supp:gss1}
U_{ij} & \equiv (S_{ij}  \varepsilon_{i} - (E[S_{ij} (X_i,W_i^\prime)|\mathcal G_n])(E[A_i(X_i,W_i)^\prime|\mathcal G_n])^{-1} A_i \varepsilon_i \notag\\
\hat U_{ij} & \equiv (S_{ij} \hat \varepsilon_{i} - (\frac{1}{n}\sum_{i=1}^n S_{ij} (X_i,W_i^\prime))(\frac{1}{n}\sum_{i=1}^n A_i(X_i,W_i)^\prime)^{-1} A_i \hat \varepsilon_i ,
\end{align}
where $A_i = (Z_i,W_i^\prime)^\prime$ and $\sigma_j^2$ and $\hat \sigma_j^2$ denote the population and sample variances
\begin{equation}\label{supp:gss2}
\sigma_j^2 \equiv \text{var}\{U_{ij}|\mathcal G_n\} \hspace{0.4 in}     \hat \sigma_j^2 \equiv \frac{1}{n}\sum_{i=1}^n (\hat U_{ij} - \frac{1}{n}\sum_{k=1}^n \hat U_{kj})^2.
\end{equation}
It will also prove helpful to define the variables $\{R_{ij}\}_{i=1}^n$ for $1\leq j \leq p$ according to
\begin{equation}\label{supp:gss3}
R_{ij} \equiv  S_{ij}\varepsilon_{i}-(\frac{1}{n}\sum_{k=1}^{n}S_{kj}(  X_{k},W_{k}^{\prime}) )  (\frac{1}{n}\sum_{k=1}^{n}A_{k}(  X_{k},W_{k}^{\prime})  )^{-1}A_i\varepsilon_i.
\end{equation}
In particular, by definition of $T_n$ and standard manipulations it then follows that
\begin{equation}\label{supp:gss4}
T_n = \max_{1\leq j \leq p}\left|\sum_{i=1}^n \frac{R_{ij}}{\hat \sigma_j}\right| .
\end{equation}

In order to apply Lemma \ref{lm:clt}(ii), to justify the asymptotic validity of the overidentification test of Section \ref{test:gss}, we require suitable moment conditions and that
\begin{align}
\max_{1\leq j \leq p} |\frac{\sqrt{\log(p)}}{\sqrt n} \sum_{i=1}^n \frac{R_{ij}-U_{ij}}{\sigma_j}| & = o_P(1)  \label{supp:gss5}\\
\max_{1\leq j \leq p}  \frac{\log^2(p)}{n} \sum_{i=1}^n \frac{(\hat U_{ij} - U_{ij})^2}{\sigma_j^2} & = o_P(1) \label{supp:gss6},
\end{align}
where probability statements are understood to be conditionally on $\mathcal G_n$ and requirements \eqref{supp:gss5} and \eqref{supp:gss6} to hold almost surely in $\mathcal G_n$.
By relying on Lemma D.5 in \cite{chernozhukov2019inference}, it is possible to show that requirement \eqref{supp:gss6} in fact implies that
\begin{equation}\label{supp:gss7}
\log(p) \times \max_{1\leq j \leq p} |\frac{\sigma_j}{\hat \sigma_j} - 1| = o_P(1),
\end{equation}
where, again, probabilities are understood to be conditionally on $\mathcal G_n$ and \eqref{supp:gss7} to hold almost surely in $\mathcal G_n$.
Moreover, by the triangle inequality and condition \eqref{supp:gss5} we have
\begin{multline}\label{supp:gss8}
\max_{1\leq j \leq p}|\frac{\sqrt{\log(p)}}{\sqrt n}\sum_{i=1}^n \frac{R_{ij}}{\hat \sigma_j} - \frac{U_{ij}}{\sigma_j}| \\ \leq (\max_{1\leq j \leq p}|\frac{\sqrt{\log(p)}}{\sqrt n}\sum_{i=1}^n \frac{U_{ij}}{\sigma_j}|+o_P(1))\times \max_{1\leq j \leq p}|\frac{\sigma_j}{\hat \sigma_j}-1| = o_P(1),
\end{multline}
where the final result follows from result \eqref{supp:gss7} and a standard maximal inequality; see, e.g., Lemma 2.2.2 in \cite{vandervaart:wellner:1996}.
Result \eqref{supp:gss8} together with \eqref{supp:gss4} and $\psi_{ij} \equiv U_{ij}/\sigma_j$ imply that Assumption \ref{ass:cltcond1} holds with $b_n = n$, $q = p$, and $r_{1n}(\varrho,\mathcal G_n)$ satisfying $r_{1n}(\varrho,\mathcal G_n) = o_{as}(1)$ for any $\varrho > 0$.

In order to verify Assumption \ref{ass:cltcond3}, recall that $\mathcal D_n\equiv (\{Y_i,X_i,S_i\}_{i=1}^n,\Sh,\mathcal G_n)$ and note that if $\{\omega_i\}_{i=1}^n$ are i.i.d.\ standard normal random variables independent of $\mathcal D_n$, then a standard maximal inequality yields that
\begin{multline}\label{supp:gss9}
E[\max_{1\leq j \leq p} |\frac{\sqrt{\log(p)}}{\sqrt n}\sum_{i=1}^n\omega_i((\hat \psi_{ij}-\psi_{ij}) - \frac{1}{n}\sum_{k=1}^n (\hat \psi_{kj} - \psi_{kj})|] \\
\lesssim \sqrt{\log(p)}\times \max_{1\leq j \leq p} (\frac{\log(p)}{n}\sum_{i=1}^n (\hat \psi_{ij} - \psi_{ij})^2)^{1/2}.
\end{multline}
Moreover, employing that $\hat \psi_{ij} \equiv \hat U_{ij}/\hat \sigma_j$ and $\psi_{ij} \equiv U_{ij}/\sigma_j$ we can bound \eqref{supp:gss9} by
\begin{multline*}
    \max_{1\leq j \leq p} \frac{\log^2(p)}{n}\sum_{i=1}^n (\hat \psi_{ij} - \psi_{ij})^2 \\ \lesssim 
     \frac{\log^2(p)}{n}\times(\max_{1\leq j \leq p} \sum_{i=1}^n \frac{(\hat U_{ij}-U_{ij})^2}{\sigma_j^2}(1+o_P(1)) + \max_{1\leq j \leq p} \sum_{i=1}^n \frac{U_{ij}^2}{\sigma_j^2} (\frac{\sigma_j}{\hat \sigma_j} - 1)^2) = o_P(1),
\end{multline*}
where the final result holds by \eqref{supp:gss6} and a standard maximal inequality. 
By Markov's inequality and result \eqref{supp:gss9} it follows that Assumption \ref{ass:cltcond3} holds for some $r_{2n}(\varrho,\mathcal D_n)$ satisfying $P(r_{2n}(\varrho,\mathcal D_n) > \epsilon|\mathcal D_n) = o_{as}(1)$ for any $\varrho > 0$ as required by Lemma \ref{lm:clt}.

\subsection{Identification Through Shocks}\label{sec:suppcalakm}


We next discuss the overidentification test of Section \ref{test:akm}, which is designed for applications in which identification is driven by exogeneity of the shocks $\Sh$.
Recall that in the corresponding asymptotic framework, originally developed by \cite{adao2019shift}, we set $\mathcal G_n = \{S_i,W_i,\varepsilon_i\}_{i=1}^n$.
Following the notation in Section \ref{test:akm}, we further set 
\begin{align}\label{supp:akm1}
\delta_j & \equiv \left(\sum_{i=1}^n W_iW_i^\prime\right)^{-1} \sum_{i=1}^n W_i g_j(\varepsilon_i,W_i,S_i) \notag\\
\kappa_j & \equiv \left(\sum_{i=1}^n E[S_i^\prime \mathcal E X_i|\mathcal G_n]\right)^{-1}\left(\sum_{i=1}^n E[S_i^\prime \mathcal E X_i|\mathcal G_n] \frac{\partial} {\partial \varepsilon} g_j(\varepsilon_i,W_i,S_i)\right),
\end{align}
where $\mathcal E \equiv \Sh - E[\Sh|\mathcal G_n]$ and we note that $\delta_j$ and $\kappa_j$ depend on $n$ (through $\mathcal G_n$), but we suppress the dependence from the notation.
As estimators for $\delta_j$ and $\kappa_j$ we employ
\begin{align*}
\hat \delta_j & \equiv \left(\sum_{i=1}^n W_iW_i^\prime\right)^{-1} \sum_{i=1}^n W_i g_j(\hat \varepsilon_i,W_i,S_i) \\
\hat \kappa_j & \equiv \left(\sum_{i=1}^n (Z_i - W_i^\prime \hat \pi_n) X_i\right)^{-1} \left(\sum_{i=1}^n (Z_i - W_i^\prime \hat \pi_n) X_i \frac{\partial}{\partial \varepsilon} g_j(\hat \varepsilon_i,W_i,S_i)\right), 
\end{align*}
where $\hat \pi_n$ denotes the  coefficient from regressing $\{Z_i\}_{i=1}^n$ on $\{W_i\}_{i=1}^n$.
In addition let
\begin{align*}
 U_{ij} & \equiv \mathcal E_i \times \sum_{k=1}^n S_{ki}(g_j(\varepsilon_k,W_k,S_k) - W_k^\prime \delta_j - \varepsilon_k \kappa_j) \\   
 \hat U_{ij} & \equiv \hat {\mathcal E}_i \times \sum_{k=1}^n S_{ki}(g_j(\hat \varepsilon_k,W_k,S_k)-W_k^\prime \hat \delta_j - \hat \varepsilon_k \hat \kappa_j)
\end{align*}
for $\hat {\mathcal E}_i$ an estimator of $\mathcal E_i$ (see Remark \ref{rm:akmest}), and recall that $\psi_{ij} \equiv U_{ij}/\sigma_j$ and $\hat \psi_{ij} \equiv \hat U_{ij}/\hat \sigma_j$, where $\sigma_j$ and $\hat \sigma_j$ respectively denote the population and finite sample variances
\begin{equation*}
\sigma_j^2 \equiv \frac{1}{p}\sum_{i=1}^p \text{Var}\{U_{ij}|\mathcal G_n\} \hspace{0.5 in} \hat \sigma_j^2 \equiv \frac{1}{p}\sum_{i=1}^p \left( \hat U_{ij} - \frac{1}{p}\sum_{k=1}^p \hat U_{kj}\right)^2.
\end{equation*}
It will also prove convenient to define the variables $\{R_{ij}\}_{i=1}^n$ for $1\leq j \leq q$ according to
\begin{equation*}
R_{ij} \equiv g_j(\hat \varepsilon_i,W_i,S_i)(Z_i-W_i^\prime \hat \pi_n).
\end{equation*}

The asymptotic validity of the overidentification test of Section \ref{test:akm} may be justified by employing Lemma \ref{lm:clt}(i).
In order to appeal to Lemma \ref{lm:clt}(i), first note that if $\{\mathcal E_i\}_{i=1}^p$ are (uniformly) Sub-Gaussian almost surely in $\mathcal G_n$, then Assumption \ref{ass:cltcond2} can be verified by setting $B_n(\mathcal G_n) = KC_n$ for $K$ large enough and $C_n$ given by
\begin{equation*}
C_n \equiv \max_{1\leq i \leq p} \max_{1\leq j \leq q}  \left(\frac{\text{Var}\{U_{ij}|\mathcal G_n\}}{\frac{1}{p}\sum_{i=1}^p \text{Var}\{U_{ij}|\mathcal G_n\}}\right)^{1/2}.    
\end{equation*}
In turn, Assumptions \ref{ass:cltcond1} and \ref{ass:cltcond2} can be verified under the key requirements
\begin{align}
\max_{1\leq j \leq q} \frac{\sqrt{\log(q)}}{\sigma_j\sqrt p} |\sum_{i=1}^n R_{ij} - \sum_{i=1}^p U_{ij}|  & =  o_P(1) \label{supp:akm7}\\
\max_{1\leq j \leq q} \frac{\log^2(q)}{p} \sum_{i=1}^p \frac{(\hat U_{ij}-U_{ij})^2}{\sigma_j^2} & =  o_P(1) \label{supp:akm8},
\end{align}
where the convergence in probability statement should be understood as jointly over all the data (rather than conditionally on $\mathcal G_n$).
In particular, under the condition that
\begin{equation}\label{supp:akm9}
\frac{C_n^2\log^2(q)}{p^{(1-c)/2}} = o_P(1)    
\end{equation}
for some $0<c<1$, it is possible to argue by relying on Lemma D.5 in \cite{chernozhukov2019inference} that requirement \eqref{supp:akm8} in fact implies that
\begin{equation}\label{supp:akm10}
\log(q)\times \max_{1\leq j \leq q} |\frac{\sigma_j}{\hat \sigma_j}-1| = o_P(1).    
\end{equation}
Moreover, by applying Lemma D.3 in \cite{chernozhukov2019inference} and relying on the rate condition in \eqref{supp:akm9} it is also possible to obtain the rate bounds
\begin{equation}\label{supp:akm11}
\max_{1\leq j \leq q}|\frac{1}{\sqrt p} \sum_{i=1}^p \frac{U_{ij}}{\sigma_j}| = O_P(\sqrt{\log(q)}) \hspace{0.3 in} \max_{1\leq j \leq q} |\frac{1}{p}\sum_{i=1}^p \frac{U_{ij}}{\sigma_j}|^2 = O_P(1).
\end{equation}
Combining results \eqref{supp:akm10} and \eqref{supp:akm11} with the same arguments employed in Section \ref{sec:suppcalcgss}, it is then straightforward to show that conditions \eqref{supp:akm7} and \eqref{supp:akm8} imply Assumptions \ref{ass:cltcond1} and \ref{ass:cltcond4} hold with $b_n = p$ and $r_{1n}(\varrho,\mathcal G_n) \vee r_{2n}(\varrho,\mathcal G_n) = o_P(1).$
To conclude verifying the main requirements of Lemma \ref{lm:clt}(i) we note that the condition $B_n^2(\mathcal G_n) \log^5(qp)/p = o_P(1)$ is implied by requirement \eqref{supp:akm9} (up to logs).

Condition \eqref{supp:akm7} is more challenging to verify than its analogue in Section \ref{sec:suppcalcgss} (i.e.\ \eqref{supp:gss5}) because there are $n$ terms $\{R_{ij}\}$ but $p$ terms $\{U_{ij}\}$.
Fortunately, as we next outline, it is possible to establish that \eqref{supp:akm7} holds by building on the assumptions and arguments in \cite{adao2019shift}.
To this end, we start with the decomposition
\begin{align}
    \sum_{i=1}^n R_{ij} = & \sum_{i=1}^n g_j(\varepsilon_i,W_i,S_i)(Z_i - W_i^\prime \hat \pi_n) \label{supp:akm12}\\
    & + \sum_{i=1}^n (g_j(\hat \varepsilon_i,W_i,S_i)- g_j(\varepsilon_i,W_i,S_i))(Z_i - W_i^\prime \pi_n) \label{supp:akm13}\\
    & + \sum_{i=1}^n (g_j(\hat \varepsilon_i,W_i,S_i) - g_j(\varepsilon_i,W_i,S_i))W_i^\prime (\hat \pi_n - \pi_n) . \label{supp:akm14}
\end{align}
It is also helpful to note that since $Z_i = S_i^\prime \Sh$, $E[Z_i|\mathcal G_n] = W_i^\prime \pi_n$ under the null hypothesis, and $E[Z_i|\mathcal G_n] = S_i^\prime E[\Sh|\mathcal G_n]$ due to $S_i\in \mathcal G_n$, it follows that
\begin{equation}\label{supp:akm15}
Z_i - W_i^\prime \pi_n = S_i^\prime \mathcal E.    
\end{equation}
Next, note that \eqref{supp:akm15}, the definition of $\delta_j$ in \eqref{supp:akm1}, the equality $S_k^\prime \mathcal E = \sum_{i=1}^p S_{ki} \mathcal E_i$, and some algebra allows us to express term \eqref{supp:akm12} as being equal to
\begin{align*}
\sum_{i=1}^n g_j(\varepsilon_i,W_i,S_i)(Z_i - W_i^\prime \hat \pi_n) &  = \sum_{i=1}^n g_j(\varepsilon_i,W_i,S_i)\{(Z_i - W_i^\prime \pi_n) + W_i^\prime(\pi_n - \hat \pi_n)\} \notag \\
& = \sum_{i=1}^n g_j(\varepsilon_i,W_i,S_i)\{S_i^\prime \mathcal E - W_i^\prime (\sum_{k=1}^n W_kW_k^\prime)^{-1}\sum_{l=1}^n W_l S_l^\prime \mathcal E\} \notag \\
& = \sum_{i=1}^p \mathcal E_i \times (\sum_{k=1}^n S_{ki}(g_j(\varepsilon_k,W_k,S_k) - W_k^\prime \delta_j)).
\end{align*}

We analyze the term in \eqref{supp:akm13} through a linearization argument.
To this end, we set
\begin{equation*}
M_{1q} \equiv \max_{1\leq j \leq q} \|\frac{\partial}{\partial \varepsilon} g_j\|_\infty \hspace{0.5 in}  M_{2q} \equiv \max_{1\leq j \leq q}    \|\frac{\partial^2}{\partial \varepsilon^2} g_j\|_\infty
\end{equation*}
and, following \cite{adao2019shift}, we let $n_k \equiv \sum_{i=1}^n S_{ki}$ and set $r_n = (\sum_{k=1}^p n_k^2 )^{-1}$.
Then note that by result \eqref{supp:akm15} and a standard Taylor expansion we obtain that
\begin{align}\label{supp:akm18}
\sum_{i=1}^n (g_j(\hat \varepsilon_i,W_i,S_i)& -g_j(\varepsilon_i,W_i,S_i))(Z_i-W_i^\prime \pi_n) \notag \\
= & \sum_{i=1}^n \frac{\partial}{\partial \varepsilon} g_j(\varepsilon_i,W_i,S_i)S_i^\prime \mathcal E(X_i(\beta - \hat \beta_n) + W_i^\prime(\gamma_{\mathtt{s}} - \hat \gamma_n)) \notag \\
& + \sum_{i=1}^n \frac{\partial^2}{\partial \varepsilon^2} g_j(\tilde \varepsilon_i,W_i,S_i)(\hat \varepsilon_i - \varepsilon_i)^2 S_i^\prime \mathcal E,
\end{align}
where $\tilde \varepsilon_i$ is some intermediate value between $\hat \varepsilon_i$ and $\varepsilon_i$.
If the covariates $W_i$ are bounded almost surely, then a maximal inequality (applied conditionally on $\mathcal G_n$) yields
\begin{multline*}
E[\max_{1\leq j \leq q}\|\sum_{i=1}^p \mathcal E_i\{\sum_{k=1}^n S_{ik} \frac{\partial }{\partial \varepsilon} g_j(\varepsilon_k,W_k,S_k)W_k^\prime\| | \mathcal G_n] \\ \lesssim \sqrt{\log(q)} M_{1q} \times (\sum_{i=1}^p (\sum_{k=1}^n S_{ik})^2)^{1/2} = \frac{\sqrt{\log(q)}M_{1q}}{\sqrt{r_n}},
\end{multline*}
where the final equality follows by definition of $r_n$.
We can therefore conclude that
\begin{equation*}
\max_{1\leq j \leq q}|\sum_{i=1}^n \frac{\partial }{\partial \varepsilon} g_j(\varepsilon_i,W_i,S_i)S_i^\prime \mathcal E W_i^\prime (\gamma_{\mathtt{s}} - \hat \gamma_n) |= O_P(\frac{\sqrt {\log(q)}M_{1q}}{\sqrt{r_n}}\|\gamma_{\mathtt{s}} - \hat \gamma_n\|),
\end{equation*}
where the probability is understood to be over the entire data.
Similarly, adapting the arguments in the proof of Proposition 3 in \cite{adao2019shift} (see in particular the proof of their display (A.4)) and employing a maximal inequality for degenerate U-statistics (see, e.g., equation (3.5) in \cite{gine2000exponential}) it is possible to establish that
\begin{equation*}
\max_{1\leq j \leq q}|\sum_{i=1}^n \frac{\partial}{\partial \varepsilon} g_j(\varepsilon_i,W_i,S_i)(S_i^\prime \mathcal E X_i - E[S_i^\prime \mathcal E X_i|\mathcal G_n])(\beta - \hat \beta_n)| = O_P(\frac{\log(q)M_{1q}}{\sqrt{r_n}}|\hat \beta_n - \beta|).
\end{equation*}
Moreover, the arguments in \cite{adao2019shift} can additionally be used to conclude that
\begin{equation}\label{supp:akm22}
 \hat \beta_n - \beta = \frac{\sum_{i=1}^n S_i^\prime \mathcal E \varepsilon_i}{\sum_{i=1}^n E[S_i^\prime \mathcal E X_i|\mathcal G_n]} + O_P(\frac{\|\hat \gamma - \gamma_{\mathtt{s}}\|}{n\sqrt{r_n}} + \frac{1}{n^2 r_n}),   
\end{equation}
while sup norm bound on the quadratic term in \eqref{supp:akm18} and the definition of $M_{2q}$ imply
\begin{equation*}
\max_{1\leq j \leq q}|\sum_{i=1}^n \frac{\partial^2}{\partial \varepsilon^2} g_j(\tilde \varepsilon_i,W_i,S_i)(\hat \varepsilon_i - \varepsilon_i)^2 S_i^\prime \mathcal E |= O_P(nM_{2q}(|\hat \beta_n - \beta|^2 + \|\hat \gamma - \gamma_{\mathtt{s}}\|^2)).
\end{equation*}
Thus, since result \eqref{supp:akm22} implies that $|\hat \beta_n - \beta| = O_P((n\sqrt{r_n})^{-1})$, our analysis so far yields
\begin{multline*}
\max_{1\leq j \leq q}|\sum_{i=1}^n(g_j(\hat \varepsilon_i,W_i,S_i) - g_j(\varepsilon_i,W_i,S_i))(Z_i - W_i^\prime \pi_n) - \kappa_j\sum_{i=1}^n S_i^\prime \mathcal E\varepsilon_i| \\= O_P(\frac{
\sqrt{\log(q)}M_{1q}}{\sqrt{r_n}}\|\hat \gamma - \gamma_{\mathtt{s}}\| + nM_{2q}\|\hat \gamma - \gamma_{\mathtt{s}}\|^2 + \log(q)\frac{M_{1q}\vee M_{2q}}{nr_n}).
\end{multline*}
Finally, using that $\|\hat \pi_n - \pi_n\|\vee |\hat \beta - \beta| = O_P((n\sqrt{r_n})^{-1})$ and relying on the mean value theorem allows to bound in probability the term in \eqref{supp:akm14} by 
\begin{equation*}
\max_{1\leq j \leq q}|\sum_{i=1}^n (g_j(\hat \varepsilon_i,W_i,S_i) - g_j(\varepsilon_i,W_i,S_i))W_i^\prime (\hat \pi_n - \pi_n)| = O_P(\frac{M_{1q}}{\sqrt{r_n}}(\frac{1}{n\sqrt{r_n}} + \|\hat \gamma - \gamma_{\mathtt{s}}\|)).
\end{equation*}
To simplify our bounds, we suppose that $\|\hat \gamma - \gamma_{\mathtt{s}}\|$ has the same rate of convergence as $|\hat \beta - \beta|$ so that $\|\hat \gamma - \gamma_{\mathtt{s}}\| = O_P((n\sqrt{r_n})^{-1})$. 
Combining our analysis of the terms in \eqref{supp:akm12}-\eqref{supp:akm14} together with the definition of $R_{ij}$ and $U_{ij}$ we can then conclude that
\begin{equation*}
\max_{1\leq j \leq q} |\sum_{i=1}^n R_{ij} - \sum_{i=1}^p U_{ij}| = O_P(\log(q)\frac{M_{1q}\vee M_{2q}}{n r_n} ) .  
\end{equation*}
Thus, finally setting $\underline{\sigma} \equiv \min_{1\leq j \leq q} \sigma_j$ we obtain that \eqref{supp:akm7} is implied by the condition
\begin{equation*}
\frac{\log^{3/2}(q)(M_{1q}\vee M_{2q})}{\underline{\sigma} \sqrt{p} n r_n } = o_P(1).    
\end{equation*}

\section{Simulation Evidence} \label{sec:mc}

We next conduct a series of Monte Carlo simulations to evaluate the finite sample performance of the overidentification tests proposed in Section \ref{sec:test}.
With the goal of informing the implementation of our tests in the empirical application of Section \ref{sec:china}, we employ simulation designs based on the \cite{david2013china} dataset.
In particular, as in \cite{david2013china}, our designs consist of short panels with $T = 2$ time periods, $n = 722$ commuting zones, and $p = 397$ sectors defined by four digit SIC codes.

In what follows, we incorporate the short panel structure into our notation by letting $Y_{it}$, $X_{it}$, $W_{it}$, and $Z_{it}$ respectively denote the outcome, regressor, controls, and instrument for commuting zone $i$ at time periods $t$.
We also note that in \cite{david2013china} both the regressor $X_{it}$ and instrument $Z_{it}$ have a Bartik structure and hence we now index shares and aggregate shocks by subscripts $x$ and $z$.
Concretely, we have
\begin{equation}\label{sec:mc1}
X_{it} = S_{xit}^\prime \Sh_{xt} \hspace{0.5 in} Z_{it} = S_{zit}^\prime \Sh_{zt}
\end{equation}
where $S_{xit}$ and $S_{zit}$ represent share vectors for commuting zone $i$ at time $t$ and $\Sh_{xt}$ and $\Sh_{zt}$ denote aggregate shocks at time $t$.
Finally, because \cite{david2013china} weight all observations by the start of period commuting zone population, in our simulations we employ the same weights throughout the analysis.

\subsection{Conditioning on Shocks}\label{mc:gss}

We begin by examining the finite sample performance of the overidentification test proposed in Section \ref{test:gss}, which recall was designed for applications that implicitly condition on the aggregate shocks -- i.e.\ that employ asymptotic approximations based on only $n$ growing.
As discussed in Remark \ref{rm:panelGSS}, implicitly conditioning on aggregate shocks in short panels yields the overidentifying moment restrictions
\begin{equation}\label{mc:gss1}
E[S_{zit}\varepsilon_{it}] = 0 \hspace{0.5 in} \text{ for } 1\leq t \leq T.
\end{equation}
Since $T =2$ and $p = 397$ in the context of \cite{david2013china}, result \eqref{mc:gss1} represents a total of 794 possible moment restrictions.
Moreover, because \cite{david2013china} cluster observations at the state level, their effective number of observations is 48.

In designing our simulations, we aimed to reflect the clustering structure in \cite{david2013china} by employing a heteroskedastic version of the group shock model of \cite{moulton1986random}.
To this end, we let $c$ denote a cluster, which consists of the commuting zone time pairs $(i,t)$ for which $i$ belongs to the state represented by $c$, and let $C$ denote the collection of all clusters.
Employing the \cite{david2013china} dataset, we then estimate a model in which the errors $\varepsilon_{it}$ are assumed to have the structure
\begin{equation*}
\varepsilon_{it} = \eta_c + \zeta_{it}
\end{equation*}
where $\eta_{c}$ are i.i.d.\ cluster level shocks and $\zeta_{it}$ are i.i.d.\ shocks and independent of $\eta_{c}$. 
We further impose a parsimonious heteroskedasticity specification by supposing that
\begin{align*}
E[\eta_c^2|\mathcal A_n] & = a_\eta + s_\eta (\sum_{(i,t)\in c} S^\prime_{zit}S_{zit}) \notag \\
E[\zeta_{it}^2|\mathcal A_n] & = a_\zeta + s_\zeta (S_{zit}^\prime S_{zit})
\end{align*}
for some constants $a_\eta,s_\eta,a_\zeta,$ and $s_\zeta$ and $\mathcal A_n \equiv \{S_{zit},S_{xit},W_{it}\}$.
In order to estimate this model, we employ the fitted residuals $\{\hat e_{it}\}$ from the weighted instrumental variable estimation in the main specification of \cite{david2013china} and let
\begin{align}\label{mc:gss4}
(\hat a_\eta,\hat s_\eta) & \equiv \arg\min_{a,s\in \mathbf R} \sum_{c\in C}     \sum_{(it)\neq (\tilde i,\tilde t) \in c} (\hat e_{it} \hat e_{\tilde i \tilde t} - a- s(\sum_{(i^\prime t^\prime)\in c} S_{zi^\prime t^\prime}^\prime S_{zi^\prime t^\prime}))^2 \notag \\
(\hat a_\zeta, \hat s_\zeta) & \equiv \arg\min_{a,s \in \mathbf R} \sum_{c\in C}\sum_{(i,t)\in c} (\hat e_{it}^2 - \hat \sigma_{\eta,c}^2 - a - s(S^\prime_{zit} S_{zit}))^2
\end{align}
where
\begin{equation*}
\hat \sigma_{\eta,c}^2 \equiv \max\{\hat a_\eta + \hat s_\eta (\sum_{(i, t)\in c} S_{zi t}^\prime S_{zi t}),0\}.
\end{equation*}

Given these estimates, we generate our Monte Carlo samples as follows:

\noindent \textsc{Step 1.} We employ the same controls $\{W_{it}\}$, aggregate shocks $\mathcal{Z}_{zt}$ and $\mathcal{Z}_{xt}$, and regression weights as in the main specification of \cite{david2013china}, which we keep fixed across all the simulations. \qed

\noindent \textsc{Step 2.} For each $t\in \{1,2\}$ we draw $n$ observations $\{S_{xit}^{\ast
},S_{zit}^{\ast}\}_{i=1}^n$ with replacement from the original full sample set of shares $\{S_{xit},S_{zit}\}_{i=1}^n$. \qed

\noindent \textsc{Step 3.} Given the sample $\{S_{xit}^*,S_{zit}^*\}$, we create a sample of instruments $\{Z_{it}^{\ast}\}$ and endogenous variables $\{X_{it}^\ast\}$ by setting $Z_{it}^{\ast}\equiv (S_{zit}^{\ast})^{\prime}\mathcal{Z}_{zt}$ and letting $X_{it}^{\ast}\equiv (S_{xit}^{\ast})^{\prime}\mathcal{Z}_{xt}$. \qed

\noindent \textsc{Step 4.} For each commuting zone time pair $(i,t)$ and cluster $c$ we create the variances
\begin{align*}
(\hat \sigma_{\zeta,it}^*)^2 & \equiv \max\{\hat a_\zeta + \hat s_\zeta (S_{zit}^*)^\prime S_{zit}^*,0\}   \\
(\hat \sigma_{\eta,c}^*)^2 & \equiv \max\{\hat a_\eta + \hat s_\eta \sum_{(i,t)\in c}(S^*_{zit})^\prime S_{zit}^*,0\}
\end{align*}
by employing the full sample estimates $\hat a_\eta,\hat s_\eta,\hat a_\zeta$ and $\hat s_\zeta$ from \eqref{mc:gss4}. \qed

\noindent \textsc{Step 5.} To create a sample of outcomes for our simulations, we draw $|C|$ i.i.d.\ standard normal variables $\{V_c\}_{c\in C}$, $n\times T$ i.i.d.\ standard normal variables $\{U_{it}\}$, and set
\begin{equation*}
Y_{it}^* \equiv X_{it}^* \hat \beta + W_{it}^\prime \hat \gamma_{\mathtt{s}} + V_c \hat \sigma_{\eta,c}^* + U_{it} \hat \sigma^*_{\zeta,it},
\end{equation*}
where $(\hat \beta,\hat \gamma_{\mathtt{s}})$ denote the weighted instrumental variable estimators from the main specification in \cite{david2013china} and the ``$c$" subscript is understood to refer to the state to which commuting zone $i$ belongs. \qed

By repeating Steps 1-5 we generate one thousand samples $\{Y_{it}^*,X_{it}^*,Z_{it}^*,W_{it}\}$ on which we evaluate the finite sample properties of our test.
We note that the number of moment restrictions ($p\times T =794$) in \eqref{mc:gss1} far exceeds the number of clusters in the simulations ($48$ states).
Since any linear combination of the moment restrictions in \eqref{mc:gss1} is also a valid moment restriction, we also examine the performance of our test when adding restrictions across time periods and/or different levels of SIC codes.
Table \ref{tab:gss} reports the finite sample rejection probabilities for tests based on different choices of moments and significance levels. 
The final column of Table \ref{tab:gss} additionally reports the p-value obtained when the test is implemented in the data of \cite{david2013china}.
All critical values were obtained by following the procedure in Section \ref{test:gss} with one thousand bootstrap draws.

\begin{table}[t!]
\begin{center}
\caption{Rejection Probabilities}													\label{tab:gss}
\begin{tabular}{lc c cccc}
\hline \hline
             &              & & \multicolumn{3}{c}{Significance Level} &   \\
Moment Restrictions                            & $\#$ Moments & & $1\%$   & $5\%$     & $10\%$  & p-value \\ \hline
Four Digit SIC $\&$ all time periods    & 794          & & 0.000   & 0.004     &  0.013 & 0.1274  \\     
Four Digit SIC $\&$ time aggregated     & 397          & & 0.000   & 0.008     &  0.031  & 0.168 \\
Three Digit SIC $\&$ all time periods   & 272          & & 0.000   & 0.007     &  0.029 & 0.0488  \\
Three Digit SIC $\&$ time aggregated    & 136          & & 0.000   & 0.007     & 0.025 & 0.0748   \\
Two Digit SIC $\&$ all time periods     & 40           & & 0.001   & 0.030     & 0.072  & 0.0054  \\
Two Digit SIC $\&$ time aggregated      & 20           & & 0.004   & 0.034     & 0.095 & 0.0018   \\  
\hline
\hline
\end{tabular}
\end{center}
\begin{flushleft}
\scriptsize{Finite sample rejection probabilities for overidentification tests of the validity of the moment restrictions imposed when asymptotic approximation implicitly condition on aggregate shocks.
}		
\end{flushleft}\vspace{-0.3 in}
\end{table}

Overall we find that the test is able to control size across all specifications.
However, for larger values of the number of moments, the finite sample rejection probability of the test is significantly below its nominal level.
The test performs best when aggregating across time periods and two digit SIC codes.
For this specification, which consists of twenty moments, the finite sample rejection probabilities of the test are close to the nominal levels.
Because the design only has 48 clusters, we view this specification as still employing a large number of moments relative to the sample size.

\subsection{Identification Through Shocks}\label{mc:akm}

We conclude by examining the finite sample performance of the overidentification test proposed in Section \ref{test:akm}, which was designed for applications in which the exogeneity of the instrument is due to the exogeneity of the aggregate shocks.
Recall that in these applications $\mathcal G_n = \{S_{zit},S_{xit},W_{it},\varepsilon_{it}\}$ and the overidentifying restriction is given by
\begin{equation}\label{mc:akm1}
E[Z_{it}|\mathcal G_n] = W_{it}^\prime \pi_n    
\end{equation}
where
\begin{equation*}
\pi_n = (\sum_{t=1}^T\sum_{i=1}^n W_{it}W_{it}^\prime)^{-1}\sum_{t=1}^T\sum_{i=1}^n W_{it} (S_{zit}^\prime E[Z_{it}|\mathcal G_n]).
\end{equation*}

In order to ensure that the null hypothesis holds in our simulation design, we rely on a model proposed by \cite{adao2019shift} as a sufficient condition for \eqref{mc:akm1}.
Specifically, we suppose that for some $p\times d$ matrix of shock $\Sh_{wt}$ and $d\times 1$ vector $\Gamma$ we have
\begin{equation}\label{mc:akm3}
E[\Sh_{zt}|\mathcal G_n] = \Sh_{wt} \Gamma \hspace{0.5 in} W_{it} = \Sh_{wt}^\prime S_{zit}.
\end{equation}
To estimate this model in the original \cite{david2013china} dataset, we first compute a ridge regression of the coordinates of $W_{it}\in \R^{d}$ on $S_{zit}$ by setting
\begin{equation*}
\hat \delta_{jt} \equiv (\sum_{i=1}^n S_{zit}S_{zit}^\prime +  \lambda I_p)^{-1}\sum_{i=1}^n S_{zit} W_{itj},  
\end{equation*}
for each $j$ and $t$, where $W_{itj}$ denotes the $j^{th}$ coordinate of the vector $W_{it}\in \R^{d}$, $I_p$ is a $p\times p$ identity matrix, and we set the penalty $\lambda$ to equal $0.1$.
Given these estimates, we let $\hat \Sh_{wt} \equiv [\hat \delta_{1t},\ldots, \hat \delta_{dt}]$ and estimate $\Gamma$ through the regression
\begin{equation*}
\hat \Gamma \equiv \arg\min_{g \in \R^{d}} \sum_{t=1}^T \|\Sh_{zt} - \hat \Sh_{wt}g\|^2.    
\end{equation*}
In what follows, it will be helpful to define $\hat E[\Sh_{zt}|\mathcal G_n] \equiv \hat \Sh_{wt}\hat \Gamma$ and $\hat \nu_t \equiv \Sh_{zt}-\hat \Sh_{wt}\hat \Gamma$.

We further aim to reflect the clustering structure in \cite{david2013china}.
To this end, we follow \cite{adao2019shift} and \cite{borusyak2022quasi} who in re-examining the empirical analysis of \cite{david2013china} cluster shocks at the three digit SIC code.
As in Section \ref{mc:gss}, we estimate a common shock model in which $\Sh_{zt}$ satisfies
\begin{equation*}
\Sh_{ztj} = \eta_c + \zeta_{tj}   
\end{equation*}
where $\Sh_{ztj}$ denotes the $j^{th}$ coordinate of $\Sh_{zt}$, $\eta_c$ are i.i.d.\ cluster level shocks, and $\zeta_{tj}$ are i.i.d.\ shocks independent of $\eta_c$.
We estimate the variance of these shocks by setting
\begin{align*}
\hat \sigma_\eta^2 & \equiv \frac{1}{|C|}\sum_{c\in C}\frac{1}{n_c(n_c-1)}\sum_{(t,j)\neq (\tilde t,\tilde j)\in c} \hat \nu_{tj}\hat \nu_{\tilde t \tilde j}\notag \\
\hat \sigma_\zeta^2 & \equiv \frac{1}{|C|}\sum_{c\in C} \frac{1}{n_c} \sum_{(s,t)\in c} \hat \nu_{tj}^2 - \hat \sigma_\eta^2,
\end{align*}
where $n_c$ denotes the number of observations in cluster $c$ and $\nu_{tj}$ denotes the $j^{th}$ coordinate of the vector $\hat \nu_t \in \R^p$.

Finally, in order to reflect the strength of the instrument in \cite{david2013china} in our simulation design, we run the following regression on the aggregate shocks
\begin{equation*}
(\hat \alpha,\hat \kappa ) \equiv \arg\min_{a,k\in \R} \sum_{t=1}^T\sum_{j=1}^p (\Sh_{xtj} - a - k \Sh_{ztj})^2
\end{equation*}
and let $\hat \sigma^2_\xi$ denote the sample variance of the residuals $\hat \xi_{tj} \equiv \Sh_{xtj} - \hat \alpha - \hat \kappa \Sh_{ztj}$.

Given these estimates, we generate our Monte Carlo samples as follows:

\noindent \textsc{Step 1}. We first create controls $\hat W_{it} \equiv \hat\Sh_{wt}^\prime S_{zit}$, which we note have the structure required by model \eqref{mc:akm3}.
We combine the controls $\{\hat W_{it}\}$ with the shares $\{S_{zit},S_{xit}\}$ in \cite{david2013china} and keep all of them fixed across simulation designs. \qed

\noindent \textsc{Step 2}. To generate our instrument, we first draw $|C|$ i.i.d.\ standard normal variables $\{V_c\}_{c\in C}$, $p\times T$ i.i.d.\ standard normal variables $\{U_{ztj}\}$, and set
\begin{equation*}
\Sh_{ztj}^* = \hat E[\Sh_{ztj}|\mathcal G_n] + V_c \hat \sigma_\eta + U_{ztj} \hat \sigma_\zeta,
\end{equation*}
where the ``$c$" subscript refers to the three digit SIC code to which sector $j$ belongs.
As our instrument we then employ $Z_{it}^* \equiv S_{zt}^\prime \Sh^*_{zt}$.
Note that, because $\hat E[\Sh_{zt}|\mathcal G_n] \equiv \hat \Sh_{wt} \hat \Gamma$, the shocks $\Sh^*_{zt}$ have the structure required by model \eqref{mc:akm3}. \qed

\noindent \textsc{Step 3}. Similarly, in order to generate aggregate shocks for our regressor, we draw $p\times T$ i.i.d.\ standard normal random variables $\{U_{xtj}\}$ and let
\begin{equation*}
\Sh^*_{xtj} = \hat \alpha + \hat \kappa \Sh_{ztj}^* + V_{xtj} \hat \sigma_\xi.
\end{equation*}
As our regressor, we then employ $X_{it}^* \equiv S_{xit}^\prime \Sh_{xt}^*$. \qed

\noindent \textsc{Step 4}. Finally, we generate a sample of outcomes $Y_{it}^*$ by simply setting $Y_{it}^*$ to equal
\begin{equation*}
Y_{it}^* = X_{it}^* \hat \beta + \hat W_{it}^\prime \hat \gamma_{\mathtt s} + \hat e_{it},   
\end{equation*}
where $(\hat \beta,\hat \gamma_{\mathtt s})$ and $\{\hat e_{it}\}$ denote the estimators and residuals obtained from replacing $W_{it}$ with $\hat W_{it}$ in the main specification of \cite{david2013china}. \qed

By repeating Steps 1-4 we generate one thousand samples $\{Y_{it}^*,X_{it}^*,Z_{it}^*,\hat W_{it}\}$ on which we evaluate the performance of the overidentification test proposed in Section \ref{test:akm}.
In order to implement the test, we need to select the moments to employ (i.e.\ the functions $g_j$ in \eqref{test:akm2}) and an estimator $\hat {\mathcal E}^*_t$ for the demeaned shock
\begin{equation*}
\mathcal E^*_t \equiv \Sh^*_{zt}-E[\Sh^*_{zt}|\mathcal G_n].    
\end{equation*}
In the main specification of \cite{david2013china}, there are no control variables with the structure required by the estimator for $\mathcal E_t^*$ proposed by \cite{borusyak2022quasi}.
We therefore instead adapt the estimator of $\mathcal E^*_t$ advocated by \cite{adao2019shift} by employing
\begin{equation}\label{mc:akm12}
\hat {\mathcal E}_t^* = (\sum_{i=1}^n S_{zit} S_{zit}^\prime + \lambda I_p)^{-1}\sum_{i=1}^n S_{zit}(Z_{it}^*-\hat W_{it}^\prime\hat \pi_n^*)
\end{equation}
for $\hat \pi_n^*$ the coefficient obtained from a weighted regression of $\{Z_{it}^*\}$ of $\{\hat W_{it}\}$.
We introduce ridge regression in \eqref{mc:akm12} because the design matrix is ill-conditioned. 
In this regard, our estimator differs from that in \cite{adao2019shift} who use ordinary regression (i.e.\ $\lambda = 0$), but instead drop sectors from the regression to address the ill conditioning of the design matrix.\footnote{See page 3 in https://github.com/kolesarm/ShiftShareSE/blob/master/doc/ShiftShareSE.pdf}
The p-values of the test can depend on $\lambda$, and we employ the simulations to inform the choice of penalty $\lambda$ for our application.\footnote{Analogously, as noted by \cite{adao2019shift}, dropping sectors to ensure that the design matrix is well conditioned similarly affects the resulting standard errors.}

Finally, for our moments we select the square of the residual and moments based on the pdf of the Logit distribution, which may be interpreted as different kernel estimators.
Specifically, we employ a total of twenty moments by setting
\begin{equation*}
g_j(\varepsilon,w,s) = \left\{\begin{array}{cl} \varepsilon^2 & \text{ if } j = 1\\
\frac{\exp(\varepsilon -a_j)}{(\exp(\varepsilon-a_j)+1)^2} & \text{ if } 2\leq j \leq 20
\end{array}\right.
\end{equation*}
with $a_2,\ldots, a_{20} = -2.25, -2, \ldots, 2, 2.25$.
Table \ref{tab:akm} reports the finite sample rejection probabilities of the resulting test for different choices of the ridge parameter $\lambda$.
The final column of Table \ref{tab:akm} additionally reports the p-value obtained when the test is implemented in the data of \cite{david2013china}.
The results are based on one thousand simulations with the bootstrap implementation relying on one thousand replications.
Overall, we find that the rejection probability is close to the nominal level of the test provided that the ridge parameter is sufficiently small.

\begin{table}[t!]
\begin{center}
\caption{Rejection Probabilities}													\label{tab:akm}
\begin{tabular}{l  c  ccc  c}
\hline \hline
                          & & \multicolumn{3}{c}{Significance Level}  & \\
Ridge Parameter           & & $1\%$   & $5\%$     & $10\%$            & p-value \\ \hline
$\lambda = 1{\rm e-3}$          & & 0.029   & 0.103     & 0.177    & 0.0012    \\ 
$\lambda = 1{\rm e-4}$          & & 0.011   & 0.061     & 0.127     & 0.0074   \\ 
$\lambda = 1{\rm e-5}$          & & 0.008   & 0.049     & 0.114     & 0.0368   \\ 
$\lambda = 1{\rm e-6}$          & & 0.008   & 0.047     & 0.100      & 0.0650  \\ 
\hline
\hline
\end{tabular}
\end{center}
\begin{flushleft}
\scriptsize{Finite sample rejection probabilities for overidentification tests of the validity of the moment restrictions imposed when identification is driven by the exogeneity of the aggregate shocks. 
}		
\end{flushleft}\vspace{-0.3 in}
\end{table}

\phantomsection
\addcontentsline{toc}{section}{References}

{\small
\singlespace

}


\end{document}